\documentclass[journal,compsoc,9.5pt]{IEEEtran}

\usepackage[british]{babel}
\usepackage[T1]{fontenc}
\usepackage[utf8x]{inputenc}

\DeclareFontFamily{U}{mathx}{\hyphenchar\font45}
\DeclareFontShape{U}{mathx}{m}{n}{
      <5> <6> <7> <8> <9> <10>
      <10.95> <12> <14.4> <17.28> <20.74> <24.88>
      mathx10
      }{}
\DeclareSymbolFont{mathx}{U}{mathx}{m}{n}
\DeclareFontSubstitution{U}{mathx}{m}{n}
\DeclareMathAccent{\widecheck}{0}{mathx}{"71}

\usepackage{amsmath,amsthm,amssymb,amsfonts}
\usepackage{amsbsy}    
\usepackage{stmaryrd}  
\usepackage{calrsfs}                             
\DeclareMathAlphabet{\pazocal}{OMS}{zplm}{m}{n}  
\usepackage[mathscr]{eucal}                      
\usepackage{cutwin}  
\usepackage{xspace,xifthen}
\usepackage[dvipsnames,table]{xcolor}
\usepackage{subcaption,graphicx,wrapfig}
\usepackage[inline]{enumitem}
\usepackage{booktabs}
\usepackage[vlined]{algorithm2e} 
\usepackage{tikz,pgfplots}
\usepackage{tokcycle,marginnote}
\usepackage{twoopt}  
\usepackage{relsize}
\usepackage{breakcites}
\usepackage[%
	breaklinks=true,%
	colorlinks=false,
	implicit=false]{hyperref}
\usepackage[nameinlink]{cleveref}

\newboolean{tosubmit}
\setboolean{tosubmit}{true}

\newboolean{showhelpers}
\setboolean{showhelpers}{false}

\graphicspath{{imgs/},}

\usetikzlibrary{shapes.geometric}  

\makeatletter

\def\orcidID#1{\href{http://orcid.org/#1}{\protect\raisebox{-1.25pt}{\protect\includegraphics[height=8pt]{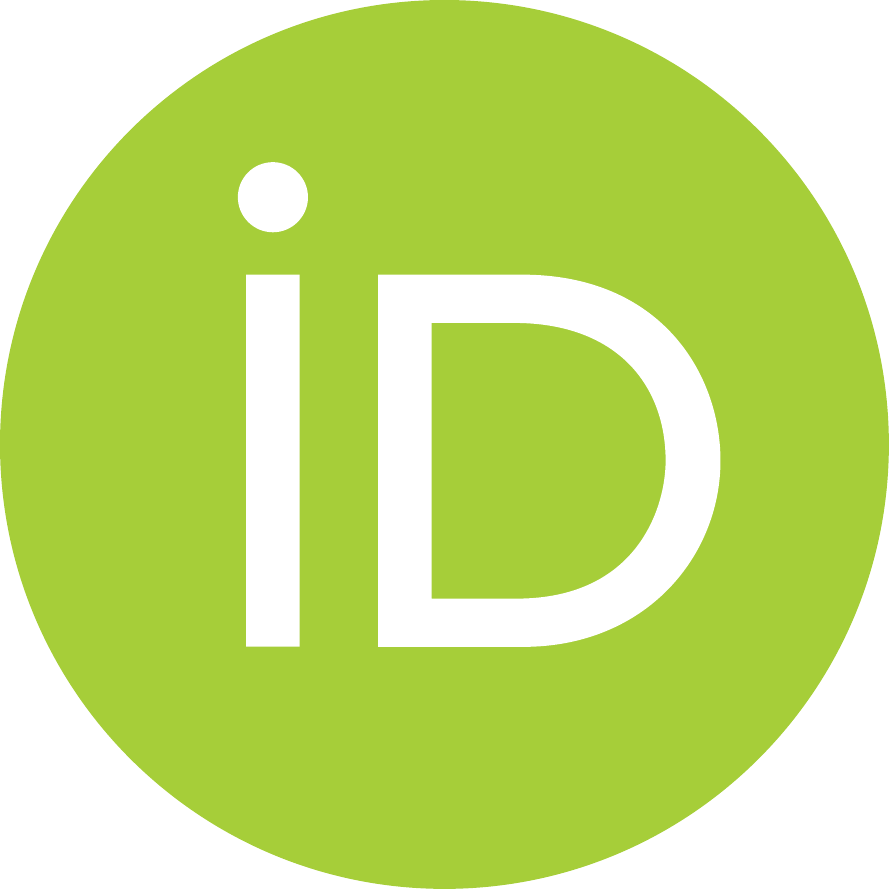}}}}
\makeatother

\makeatletter
\newtheorem*{rep@theorem}{\rep@title}
\newcommand{\newreptheorem}[2]{%
\newenvironment{rep#1}[1]{%
 \def\rep@title{#2 \ref{##1}}%
 \begin{rep@theorem}}%
 {\end{rep@theorem}}}
\makeatother

\theoremstyle{plain}
\newtheorem{theorem}{Theorem}[section]
\newreptheorem{theorem}{Theorem}  
\newtheorem{lemma}[theorem]{Lemma}
\newreptheorem{lemma}{Lemma}  
\newtheorem{corollary}[theorem]{Corollary}

\newreptheorem{proposition}{Proposition}
\newtheorem*{question}{Question}
\theoremstyle{definition}
\newtheorem{definition}[theorem]{Definition}
\newtheorem{example}[theorem]{Example}
\newtheorem*{answer}{Answer}

\Crefname{figure}{Fig.}{Figs.}
\crefname{figure}{figure}{figures}
\Crefname{tabular}{Table}{Tables}
\Crefname{algorithm}{Algo.}{Algos.}
\crefname{algorithm}{algorithm}{algorithms}
\Crefname{definition}{Def.}{Defs.}
\crefname{definition}{definition}{definitions}
\Crefname{theorem}{Theo.}{Theorems}
\crefname{theorem}{theorem}{theorems}
\Crefname{lemma}{Lemma}{Lemmas}
\Crefname{corollary}{Cor.}{Corollaries}
\crefname{corollary}{corollary}{corollaries}
\Crefname{proposition}{Prop.}{Propositions}
\crefname{proposition}{proposition}{propositions}
\Crefname{example}{Example}{Examples}
\Crefname{section}{Sec.}{Secs.}
\crefname{section}{section}{sections}
\Crefname{equation}{eq.}{eqs.}
\crefname{equation}{equation}{equations}

\SetAlCapFnt{\mdseries}  
\SetAlCapSkip{\medskipamount}
\SetStartEndCondition{\ \ }{\ \ }{}

\SetCommentSty{algocomment}

\makeatletter
\renewcommand{\paragraph}{\@startsection{paragraph}{5}{0em}%
  {.7ex plus .2ex minus .1ex}%
  {-.5em}%
  {\bfseries}}
\makeatother

\newcommand{\hlbox}[1]{%
  \smallskip\begin{center}
  \fboxrule1pt\fboxsep3pt\fcolorbox{black!45}{black!8}{%
  \begin{minipage}{.96\linewidth}#1\end{minipage}}
  \end{center}\smallskip}

\ifthenelse{\boolean{showhelpers}}{%
  \usepackage[pagewise,switch,modulo]{lineno}  
  \linenumbers
  
}{}

\ifthenelse{\boolean{tosubmit}}{}{ 
  \usepackage[firstpage]{draftwatermark}
}

\makeatletter
\def\THICKhrulefill{\leavevmode \leaders \hrule height 5pt\hfill \kern \z@}
\def\getfirst#1#2\relax{\tctestifnum{\count@stringtoks{#1}>1}{ERROR}{#1}}
\makeatother
\newcommand{\colorpar}[3]{\colorbox{#1}{\parbox{#2}{#3}}}
\newcommand{\marginremark}[3]{%
  \ifthenelse{\boolean{tosubmit}}{}{
	\marginnote{\colorpar{#2}{1.3\linewidth}{\raggedright\color{#1}#3}}
}}
\newcommand{\highlightedremark}[4]{%
  \ifthenelse{\boolean{tosubmit}}{}{
	\begin{center}\fcolorbox{#1}{#2}{%
	\begin{minipage}{.98\linewidth}\color{#1}%
	\textbf{\THICKhrulefill[ #3 ]\THICKhrulefill}%
	\par\noindent#4\end{minipage}}\end{center}%
}}
\newcommand{\hey}[4]{%
  \ifthenelse{\boolean{tosubmit}}{}{
  \leavevmode\marginnote{\sffamily\Large\color{#1}@\getfirst#3\relax\relax}
  \colorbox{#2}{\sffamily\bfseries{@#3:}}~{\sffamily\color{#1}#4}}}
\newcommand{\todo}[1]{%
  \ifthenelse{\boolean{tosubmit}}{}{
  \noindent\textsf{\color{Red}\textbf{TODO:} #1}%
  \begingroup\renewcommand*{\marginnotevadjust}{-3ex}%
  \marginnote{\textsf{\color{red}\bfseries TODO}}\endgroup}}
\newcommand{\tocite}[1][??]{%
  \ifthenelse{\boolean{tosubmit}}{}{
  \noindent\textbf{\sffamily\textcolor{blue!85}{[#1]}}%
  \begingroup\renewcommand*{\marginnotevadjust}{-3ex}%
  \marginnote{\textsf{\color{blue}\bfseries CITE!}}\endgroup}}
\colorlet{MS-fg}{WildStrawberry}
\colorlet{MS-bg}{Plum!6}
\colorlet{CEB-fg}{TealBlue!75!green!75!black}
\colorlet{CEB-bg}{Aquamarine!8}
\colorlet{MLZ-fg}{orange}
\colorlet{MLZ-bg}{orange!8}
\let\oldtop\top
\let\oldbot\bot
\renewcommand{\top}{\mathtt{{1}}}
\renewcommand{\bot}{\mathtt{{0}}}
\renewcommand{\emptyset}{\varnothing}
\renewcommand{\vec}{\mathaccent "017E\relax}  
\newcommand\restr[2]{{
  \left.\kern-\nulldelimiterspace #1 \vphantom{\big|} \right|_{#2}}}
\newcommand{\BB}{\ensuremath{\mathbb{B}}\xspace}  
\newcommand{\NN}{\ensuremath{\mathbb{N}}\xspace}  
\newcommand{\QQ}{\ensuremath{\mathbb{Q}}\xspace}  
\newcommand{\RR}{\ensuremath{\mathbb{R}}\xspace}  
\newcommand{\from}{\colon}

\newcommand{\card}[1]{\ensuremath{\vert{#1}\vert}\xspace}

\newcommand{\acronym}[1]{\ensuremath{\text{\uppercase{#1}}}\xspace}
\newcommand{\mathobject}[1]{\ensuremath{\text{\scalebox{.92}{$#1$}}}}

\newcommand{\Vars}{\ensuremath{\mathit{Vars}}\xspace}
\DeclareMathOperator{\bddOp}{\mathobject{B}}
\DeclareMathOperator{\BDDlab}{\mathit{Lab}}
\DeclareMathOperator{\low}{\mathit{Low}}
\DeclareMathOperator{\high}{\mathit{High}}
\newcommand{\Root}{\mathobject{R}}
\newcommand{\BDDnodes}{\mathobject{W}\xspace}
\newcommand{\BDDnodesN}{\ensuremath{\BDDnodes_{\mkern-5mu n \mkern 1mu}}\xspace}
\newcommand{\BDDnodesT}{\ensuremath{\BDDnodes_{\mkern-4mu t \mkern 1mu}}\xspace}
\newcommand{\BDDroot}[1][\bddOp]{\ensuremath{\Root_{#1}}\xspace}

\makeatletter  
\newcommand{\precneq}{\mathrel{\text{\prec@eq}}}
\newcommand{\prec@eq}{%
  \oalign{%
    \hidewidth$\m@th\prec$\hidewidth\cr
    \noalign{\nointerlineskip\kern1ex}%
    $\m@th\smash{\raisebox{0.65ex}{\rotatebox{90}{%
      \scalebox{1.1}[-1.1]{$\nshortmid$}}}}$\cr
    \noalign{\nointerlineskip\kern-.5ex}%
}}
\makeatother

\newcommand{\LINTIME}{\acronym{lintime}}     
\newcommand{\EXPTIME}{\acronym{exptime}}     

\newcommand{\PSPACE}{\acronym{pspace}}
\newcommand{\AT}{\acronym{at}}               
\newcommand{\ATs}{\acronym{at}{s}\xspace}
\newcommand{\SAT}{\acronym{sat}}             
\newcommand{\SATs}{\acronym{sat}{s}\xspace}
\newcommand{\DAT}{\acronym{dat}}             
\newcommand{\DATs}{\acronym{dat}{s}\xspace}

\newcommand{\DAG}{\acronym{dag}}             
\newcommand{\DAGs}{\acronym{dag}{s}\xspace}
\newcommand{\BDD}{\acronym{bdd}}             
\newcommand{\BDDs}{\acronym{bdd}{s}\xspace}
\DeclareMathOperator{\pos}{\mathit{pos}}
\newcommand{\BU}[1][]{\ensuremath{\mathtt{BU%
  \ifthenelse{\isempty{#1}}{}{_{\mkern1mu#1}}}}\xspace}
\newcommand{\BUSAT}{\BU[SAT]}
\newcommand{\BUDAT}{\BU[DAT]}
\newcommand{\BUBDD}{\ensuremath{\mathtt{BDD_{DAG}}}\xspace} 
\DeclareMathOperator{\kshortest}{\mathtt{shortest\_paths}}

\newcommand{\nodeType}[1]{\ensuremath{\mathtt{#1}}\xspace}
\newcommand{\tBAS}{\nodeType{BAS}}
\newcommand{\tOR}{\nodeType{OR}}
\newcommand{\tAND}{\nodeType{AND}}
\newcommand{\tSAND}{\nodeType{SAND}}
\newcommand{\allATs}{\ensuremath{\mathcal{T}}\xspace}
\DeclareMathOperator{\typOp}{\mathit{t}}       
\DeclareMathOperator{\chOp}{\mathit{ch}}       
\newcommand{\type}[1]{\ensuremath{\typOp({#1})}\xspace}
\newcommand{\child}[1]{\ensuremath{\chOp({#1})}\xspace}

\DeclareMathOperator{\desc}{\BAS}
\newcommand{\Suc}[1]{\operatorname{Suc}_{#1}}  
\DeclareMathOperator{\sfun}{\mathit{f}}        
\newcommand{\sfunT}[1][\T]{\ensuremath{\sfun_{\!#1}}\xspace}
\newcommand{\sem}[1]{\ensuremath{
  \llbracket{#1}\rrbracket}\xspace}
\newcommand{\semin}[1]{\ensuremath{
  \llfloor{#1}\rrfloor}\xspace}
\newcommand{\ssem}[1]{\sem{#1}\xspace}         
\newcommand{\dsem}[1]{\ssem{#1}}               
\newcommand{\dsemin}[1]{\semin{#1}}            
\newcommand{\before}[1][]{\mathbin{%
  \ifthenelse{\isempty{#1}}{%
    \tikz[baseline=-.6ex]{\draw[->,thin,x=1ex,y=1ex](0,0)--(1.7,0);}%
  }{%
    \tikz[baseline=-.48ex]{\draw[->,thin,x=1ex,y=1ex](0,0)--(1.7,0);%
    \node[x=1ex,y=1ex,inner sep=0pt]at(.75,.75){$\scriptscriptstyle{#1}$};}}}
}
\newcommand{\attack}[1][A]{\mathobject{#1}\xspace}
\newcommand{\suite}[1][S]{\ensuremath{\text{\relsize{-.5}$\pazocal{#1}$}}\xspace}
\newcommand{\allAttacks}{\ensuremath{\mathcal{A}}\xspace}
\newcommand{\allSuites}{\ensuremath{\rotatebox[origin=c]{-15}{$\mathscr{S}$\!}}\xspace}

\newcommandtwoopt{\poset}[2][\attack][\prec]{%
  \ensuremath{\langle{#1},{#2}\rangle}\xspace}
\newcommandtwoopt{\Hasse}[2][\attack][\prec]{%
  \ensuremath{\mathobject{H}_{\mkern-2mu#1}^{#2}}\xspace}
\DeclareMathOperator{\MC}{\mathrm{MC}}
\newcommand{\ccomp}{\mathobject{C}\xspace}

\DeclareMathOperator{\attrOp}{\alpha}                     
\DeclareMathOperator{\metrSOp}{
  \mkern-1.2mu\vec{\mkern1.2mu\attrOp\mkern-1.2mu}\mkern1.2mu}
\DeclareMathOperator{\metrAOp}{\widehat{\attrOp}}         
\DeclareMathOperator{\metrOp}{\widecheck{\attrOp}}        
\newcommand{\attr}[1]{\ensuremath{\attrOp(#1)}\xspace}    
\newcommand{\metrA}[1]{\ensuremath{\metrAOp(#1)}\xspace}  
\newcommand{\metr}[1]{\ensuremath{\metrOp(#1)}\xspace}    
\newcommand{\Vdom}{\mathobject{V}\xspace}
\newcommand{\domain}{\mathobject{D}\xspace}
\newcommand{\ndomain}{\ensuremath{\domain_{\mkern-1mu\star}}\xspace}
\newcommand{\operOR}{\mathbin{\triangledown}}
\newcommand{\operAND}{\mathbin{\vartriangle}}
\newcommand{\operSAND}{\mathbin{\vartriangleright}}
\DeclareMathOperator*{\bigoperOR}{\bigtriangledown}
\DeclareMathOperator*{\bigoperAND}{\bigtriangleup}
\DeclareMathOperator*{\bigoperSAND}{\mathbin{\text{\raisebox{-.7ex}{\rotatebox{90}{$\bigtriangledown$}}}}\vphantom{\bigtriangledown}}
\newcommand{\neutral}[1]{\ensuremath{1_{\!#1}}\xspace}

\newcommand{\ntOR}{\neutral{\operOR}}
\newcommand{\ntAND}{\neutral{\operAND}}
\newcommand{\ATnodes}{\mathobject{N}\xspace}
\newcommand{\ATroot}[1][\T]{\ensuremath{\Root_{#1}}\xspace}
\newcommand{\TLA}{\acronym{tla}}             
\newcommand{\OR}{\acronym{or}}               
\newcommand{\AND}{\acronym{and}}             
\newcommand{\SAND}{\acronym{sand}}           
\newcommand{\BAS}{\acronym{bas}}             
\newcommand{\BASs}{\acronym{bas}{es}\xspace}
\newcommand{\BASes}{\BASs}
\newcommand{\T}[1][]{\ensuremath{\mathobject{T}\ifthenelse{\isempty{#1}}{}{_{\hspace{-2pt}#1}}}\xspace}
\newcommand{\sampleTs}{\ensuremath{\T[\text{\larger[.5]{$s$}}]}\xspace}
\newcommand{\sampleTd}{\ensuremath{\T[\text{\larger[.5]{$d$}}]}\xspace}
\newcommand{\ff}{\ensuremath{\mathit{ff}}\xspace}
\newcommand{\ww}{\ensuremath{\mathit{w}}\xspace}
\newcommand{\cc}{\ensuremath{\mathit{c\mkern-3mu c}}\xspace}
\DeclareMathOperator{\logicformula}{\mathobject{L}}
\newcommand{\LT}[1][\T]{\ensuremath{\logicformula_{#1}}\xspace}
\newcommand{\bddT}[1][\T]{\ensuremath{\bddOp_{#1}}\xspace}

\newcommand{\bddf}{\bddT[\!f]}

\newcommand{\AC}[0]{\operatorname{AC}}
\newcommand*{\ldb}{\{\mskip-5mu\{}
\newcommand*{\rdb}{\}\mskip-5mu\}}

\allowdisplaybreaks

\def\TITLE{Efficient and Generic Algorithms for\\ Quantitative Attack Tree Analysis}

\hypersetup{pdftitle={\TITLE}}

\begin{document}
\title{%
	\TITLE
	\thanks{%
		We thank Sebastiaan Joosten for his help with the proof of
		\Cref{theo:NP_hard}, and Lars Kuijpers for collaborations that
		led to \Cref{alg:shortest_path_BDD}.
	This work was partially supported by ERC Consolidator Grant 864075
	(\href{https://www.utwente.nl/en/eemcs/fmt/research/projects/caesar/}{\emph{CAESAR}}).
	Funded by the European Union under GA n.101067199-ProSVED. Views and
	opinions expressed are those of the author(s) only and do not
	necessarily reflect those of the European Union or The European Research
	Executive Agency. Neither the European Union nor the granting authority
	can be held responsible for them.
	}
}
\author{
	\IEEEauthorblockN{%
		Milan Lopuha\"a-Zwakenberg\IEEEauthorrefmark{2}\orcidID{0000-0001-5687-854X}
		\qquad
		Carlos E.\ Budde\IEEEauthorrefmark{1}\orcidID{0000-0001-8807-1548}
		\qquad
		Mari\"elle Stoelinga\IEEEauthorrefmark{2}\IEEEauthorrefmark{3}\orcidID{0000-0001-6793-8165}
	}
	
	\IEEEauthorblockA{%
		\IEEEauthorrefmark{1}{University of Trento,
			CyberSecurity, Trento, Italy.}\\
		\IEEEauthorrefmark{2}{University of Twente,
			Formal Methods and Tools, Enschede, the Netherlands.}\\
		\IEEEauthorrefmark{3}{Radboud University,
			Department of Software Science, Nijmegen, the Netherlands.}\\
		{\texttt{\small carlosesteban.budde@unitn.it~~\{m.a.lopuhaa,m.i.a.stoelinga\}@utwente.nl}}
	}
}
\maketitle
\thispagestyle{plain}  
\pagestyle{plain}

\begin{abstract}
Numerous analysis methods for quantitative attack tree analysis have been proposed.
These algorithms compute relevant security metrics, i.e.\ performance indicators that quantify how good the security of a system is;
typical metrics being the most likely attack, the cheapest, or the most damaging one. However, existing methods are only geared towards specific metrics or do not work on general attack trees. This paper classifies attack trees in two dimensions: proper trees vs.\ directed acyclic graphs (i.e.\ with shared subtrees); and static vs.\ dynamic gates.
For three out of these four classes, we propose novel algorithms that work over a generic attribute domain, encompassing a large number of concrete security metrics defined on the attack tree semantics; dynamic attack trees with directed acyclic graph structure are left as an open problem.
We also analyse the computational complexity of our methods.
\end{abstract}


\section{Introduction}
\label{sec:intro}

Attack trees (\ATs) are important tools to analyse the security of complex systems.
Their intuitive definition and general applicability makes them widely studied in academia and used in industry.
\ATs are part of many system engineering frameworks, e.g.\ \emph{UMLsec} \cite{Jur02,RA15} and \emph{SysMLsec} \cite{AR13}, and supported by industrial tools such as Isograph's \emph{AttackTree}~\cite{IsographAT}. 

An \AT is a hierarchical diagram that describes potential attacks on a system.
Its root represents the attacker's goal, and the leaves represent basic attack steps: indivisible actions of the attacker.
Intermediate nodes are labeled with gates, that determine how their children activate them.
The most basic notions of \ATs have \OR and \AND gates only; many extensions exist to model more elaborate attacks.

Attack trees are often studied via \emph{quantitative analysis}, where \ATs are assigned a wide range of security metrics.
Typical examples of such metrics are the minimal time \cite{KRS15,AHPS14,KSR+18,LS2021}, minimal cost \cite{AGKS15}, or maximal probability \cite{JKM+15} of a successful attack, as well as Pareto analyses that study trade offs among attributes \cite{KRS15,FW19}.
Calculating such metrics is essential when comparing alternatives or making trade offs.
This leads to the following research problem:

\begin{question}
	How can AT metrics be calculated efficiently?
\end{question}

\begin{figure}[t]
	\vspace{-.7ex}
	\hspace*{.7em}%
	\begin{minipage}[b]{.6\linewidth}
		\includegraphics[width=\linewidth]{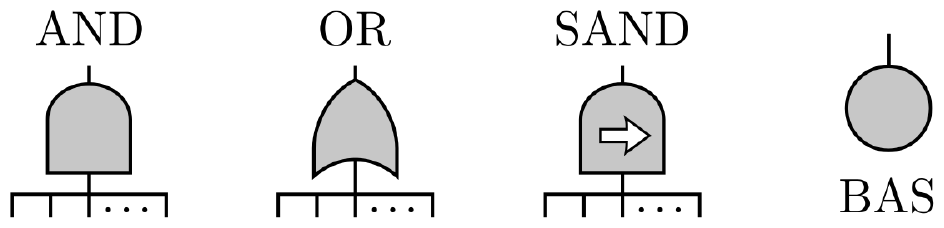}
	\end{minipage}
	\quad~
	\begin{minipage}[b]{.28\linewidth}
		\vfill~\\
		\captionof{figure}{\raggedright Nodes in an attack tree.}
		\label{fig:AT_nodes}
	\end{minipage}
	\vspace{-2ex}
\end{figure}

Numerous algorithms have been proposed to address this question. Such algorithms exploit a plethora of techniques, for instance Petri nets \cite{DMCR06}, model checking \cite{AGKS15}, and Bayesian networks \cite{GIM15}.
While these algorithms provide good ways to compute metrics, they also suffer from several drawbacks:
\begin{enumerate*}[label=(\arabic*)]
\item	Many of them are geared to specific attributes, such as attack time
		or probability, and the procedure may not extend to other metrics;
\item	Several algorithms do not exploit the acyclic structure of the \AT,
		especially approaches based on model checking;
\item	Since their application is mostly illustrated on small examples,
		it is unclear how these approaches scale to larger case studies.
\end{enumerate*}


\begin{table*}
	\centering


\begingroup
\arrayrulecolor{gray}
\colorlet{rowcolA1}{white}
\colorlet{rowcolB1}{gray!28}
\colorlet{rowcolA2}{rowcolA1}
\colorlet{rowcolB2}{rowcolB1}
\rowcolors{2}{rowcolA1}{rowcolB1}
\extrarowheight=\aboverulesep
\addtolength{\extrarowheight}{\belowrulesep}
\aboverulesep=0pt
\belowrulesep=0pt
\def\metric#1{\parbox[c][5ex]{7.7em}{\bfseries #1}}
\def\algo#1#2{\parbox[c]{#1}{\centering{#2}}}
\def\OPEN{\textsf{\scshape\bfseries open~problem}}
\def\BU{\acronym{bu}}
\def\BUKordy{$\pazocal{C}$-\BU\xspace}
\def\head#1{\larger[.5]\sffamily\color{white}{#1}}
\begin{tabular}{>{~}l|>{}c|>{}c|>{}c>{}c|>{}c}
	\rowcolor{black!75}
		  \head{Metric}
		& \head{Static tree}
		& \head{Dynamic tree}
		& \multicolumn{2}{>{}c|}{\head{Static \DAG}}
		& \head{Dynamic \DAG}
	\\[.5ex]\midrule
	\metric{min cost}
		& \algo{8em}{\BU \cite{MO06,Wei91,Sch99}}
		& \algo{6em}{\BU \cite{JKM+15}}
		& \algo{7em}{\acronym{mtbdd} \cite{BET13}}
		& \algo{5em}{\BUKordy \cite{KW18}}
		& \algo{5em}{\acronym{pta} \cite{KRS15}}
	\\
	\metric{min time}
		& \algo{7em}{~\BU \cite{MO06,HAF+09}}
		& \algo{10em}{\acronym{aph} \cite{AHPS14} ~~ \BU \cite{JKM+15}}
		& \multicolumn{2}{>{}c|}{%
		  \algo{7em}{Petri nets \cite{DMCR06}}}
		& \algo{5em}{\acronym{milp} \cite{LS2021}}
	\\
	\metric{min skill}
		& \algo{7em}{~\BU \cite{MO06,BFM04}}
		& \algo{6em}{\BU \cite{JKM+15}}
		& \multicolumn{2}{>{}c|}{%
		  \algo{8em}{\BUKordy \cite{KW18}}}
		& ---
	\\
	\metric{max damage}
		& \algo{8em}{~\BU \cite{MO06,HAF+09,BFM04}}
		& \algo{6em}{\BU \cite{JKM+15}}
		& \algo{7em}{\acronym{mtbdd} \cite{BET13}}
		& \algo{5em}{\acronym{dpll} \cite{JW08}}
		& \algo{5em}{\acronym{pta} \cite{KRS15}}
	\\
	\metric{probability}
		& \algo{7em}{\BU \cite{BLP+06,HAF+09}}
		& \algo{7em}{\acronym{aph} \cite{AHPS14}}
		& \algo{7em}{\BDD \cite{Rau93}}
		& \algo{5em}{\acronym{dpll} \cite{JW08}}
		& \algo{7em}{\acronym{i/o-imc} \cite{AGKS15}}
	\\
	\metric{Pareto fronts}
		& \algo{7em}{\BU \cite{AN15,HAF+09}}
		& \algo{7.1em}{\bfseries\Cref{lemma:antichain}}
		& 		  \algo{7em}{\BUKordy \cite{FW19}}
		& \algo{5em}{\bfseries\Cref{lemma:antichain}}
		& \algo{6em}{\acronym{pta} \cite{KRS15}}
	\\
	\rowcolor{rowcolB2}
	\metric{Any of the above}
		& \algo{7em}{\bfseries\Cref{alg:bottom_up_SAT}:~$\mathtt{BU_{SAT}}$}
		& \algo{7em}{\bfseries\Cref{alg:bottom_up_DAT}:~$\mathtt{BU_{DAT}}$}
		& \multicolumn{2}{>{}c|}{%
		  \algo{9em}{\bfseries\Cref{alg:bottom_up_BDD}:~$\mathtt{BDD_{DAG}}$}}
		& \algo{7.5em}{\OPEN\footnotemark[1]}
	\\
	\rowcolor{rowcolA2}
	\metric{$\boldsymbol{k}$-top metrics}
		& \algo{8.3em}{\BU-projection \cite{MO06}}
		& \algo{10em}{\BU \cite{GRK+16} ~~ \bfseries\Cref{lemma:topk}}
		& \multicolumn{2}{>{}c|}{%
		  \algo{13.5em}{\bfseries\Cref{alg:shortest_path_BDD}:~$\mathtt{BDD\_shortest\_paths}$}}
		& \algo{7.5em}{\OPEN\footnotemark[2]}
	\tikz[overlay]{\draw[black,very thick] (-15.4,-.25) rectangle ++(15.6,1.94);}
	\\
\end{tabular}
\endgroup

	\vspace{1ex}
	\caption{Efficient algorithms to compute security metrics on different \AT
		classes~(details and abbreviations are in \Cref{sec:related_work}).
		${}^1$~\Cref{alg:satmod} reduces runtime for any found method;
		${}^2$ \Cref{lemma:topk} reduces $k$-top calculation to metric calculation.}
	\label{tab:all_algos_intro}
	\vspace{-2ex}
\end{table*}

The aim of this work is to answer the question above for a general class of metrics. The answer hinges on two factors:
\begin{enumerate}[leftmargin=3ex]
\item	\emph{Static vs.\ dynamic ATs:}
Apart from the standard \ATs, which we call \emph{static}, an important extension are \emph{dynamic} \ATs. These allow for sequential-\AND (\SAND) gates, which require their children to succeed in left-to-right order \cite{JKM+15,AGKS15}. A formal approach to analyse dynamic \ATs requires more rich semantics than for static \ATs.
\item	\emph{Tree vs.\ DAG structure:}
Their name notwithstanding, \ATs can be directed acyclic graphs (\DAGs), in which a node may have multiple parents. The additional condition that an \AT is tree-structured allows for considerably faster computation of metrics.
\end{enumerate}
Thus, our contribution answers the question as follows:

\begin{answer}
	We provide efficient and generic algorithms to compute \AT metrics,
	by tailoring them to 
	our 2-dimensional categorisation: 
	static vs.\ dynamic \ATs, and tree-structured vs.\ \DAG-structured \ATs.
\end{answer}

On \cpageref{sec:intro:contributions} we give a precise description of our contributions, presented in accordance to the categorisation above.
Our algorithmic results are summarised in \Cref{tab:all_algos_intro}, and \Cref{sec:related_work} provides an elaborate comparison with related work.

There are various naming conventions in the \AT literature.
We follow the terminology of works like \cite{MO06,BET13,BS21,BKS21}, akin to fault tree standards on which attack trees were inspired~\cite{VSD+02}.
Two points to highlight in this respect are our use of \emph{attack tree} to refer also to \DAG-like structures, and our use of \emph{dynamic} to mean \ATs with sequential-\AND gates.
These and all other terms used in this work are formalised in \Cref{sec:AT:syntax,sec:SAT:semantics,sec:SAT:metrics}, and \Crefrange{sec:DAT:wellformed}{sec:DAT:metrics}.
In \Cref{sec:related_work} we compare our terminology to alternatives in the literature.

\paragraph{Static trees}
The simplest category of Attack Trees are tree-structured static \ATs.
As shown in a seminal paper by Mauw \& Oosdijk \cite{MO06}, metrics can be computed for these \ATs in a bottom-up fashion, using appropriate operators $\operOR$ and $\operAND$ on a set $V$, resp.\ for the \OR and \AND gates in the tree.
We show this as \Cref{alg:bottom_up_SAT}.
A key insight in \cite{MO06} is that this procedure works whenever the algebraic structure $(\Vdom,\operAND,\operOR)$ constitutes a semiring, i.e.\ $\operAND$ must distribute over $\operOR$.
\\[.3ex]
We provide an alternative proof of correctness for this result: while \cite{MO06} deploys rewriting rules for attack trees, we show in \Cref{sec:mod} that it directly follows from the validity of \emph{modular analysis} on DAG-structured ATs.
Furthermore, we propose new categories of attribute domains, which extend the application of the bottom-up algorithm to compute popular security metrics. Ordered semiring domains constitute an important category of metrics, for which we show that derived metrics such as Pareto fronts and $k$-top values also fall within the semiring attribute domain framework.

\paragraph{Static DAGs}
It is well-known that static \ATs with \DAG structure cannot be studied with bottom-up procedures \cite{Rau93,BK18}.
Many algorithms exist to tackle such \ATs, mostly geared to specific metrics \cite{AGKS15,BET13,DMCR06,JW08,KW18}.
\Cref{sec:SAT_DAGs} presents a generic algorithm that works for any semiring attribute domain $(\Vdom,\operOR,\operAND)$ that is absorbing, i.e.\ $x\operOR(x\operAND y)=x$, and has neutral elements \ntOR and \ntAND for operators $\operOR$ and $\operAND$ resp.

Concretely, we exploit a binary decision diagram representation (\BDD) of the attack tree.
Our algorithm visits each \BDD node once and is thus linear in its size.
The caveat is that \BDDs can be of exponential size in the number of basic attack steps (\BASes), but one cannot hope for faster algorithms: as we show, computing a minimal attack is an NP-hard problem. However, in practice the \BDD approach still yields computational gains: we show that we can split up the calculation according to the \emph{modules} of the \AT, i.e. sub\DAGs only connected to the rest of the \AT by their root.
This speeds up calculations considerably.
Moreover, \BDDs are known to be compact in practice \cite{Bry86}, and allow parallel traversals \cite{ODP17}, making them an overall efficient choice.
Furthermore, we show that the bottom-up algorithm---of linear runtime on any static tree---works also when applied to \DAG \ATs, if the operators $\operOR$ and $\operAND$ are idempotent, i.e.\ $x \operOR x = x \operAND x = x$.

\paragraph{Dynamic trees}
%
Metrics for dynamic attack trees (DATs) are usually decoupled from semantics, and defined either on the syntactic \AT structure, or ad hoc for the selected computation method \cite{JKM+15,KRS15,ANP16,KW18}.
The main obstacle to a semantics-based approach is to choose semantics for \DATs in a way that supports a proper definition of metric, i.e.\ that is compatible with the notion of metric of static \ATs, and that is as generic as the attribute domains from \cite{MO06}.
In particular, the interaction among multiple \SAND gates is nontrivial, because they may impose conflicting execution orders on the \BAS of the tree.
%
%

We follow \cite{BS21,LS2021} in giving semantics to \DATs via partially ordered sets (\emph{posets}).
Each poset \poset represents an attack scenario, where \attack collects all attacks steps to be performed, and $a\prec b$ indicates that step $a$ must be completed before step $b$ starts.
%
%
%
This set up enables us to define a notion of metric for \DATs based on their semantics.
\emph{A key contribution is defining general metrics for dynamic attack trees, and showing that tree-structured DATs are analysable by extending the bottom-up algorithm with an additional operator (see \Cref{alg:bottom_up_DAT})}.
Concretely, we use attribute domains with three operators: $\operOR$, $\operAND$, $\operSAND$, where $\operSAND$ distributes over $\operOR$ and $\operAND$, and $\operAND$ distributes over $\operOR$.
We prove this \namecref{alg:bottom_up_DAT} correct in our formal semantics.
This is relevant and non-trivial because
(a) earlier algorithms do not provide explicit correctness results in terms of semantics, and
(b) the metrics are formally defined on the poset semantics of a \DAT, while the algorithm works on its syntactic \AT structure. 

\paragraph{Dynamic DAGs}
Efficient computation of metrics for \DAG-structured \DATs is left as future research challenge.
A na\"ive, inefficient algorithm would enumerate all posets in the semantics.
Instead, one could extend \BDD-algorithms for static \DAGs to dynamic \ATs.
This is non-trivial, as \BDDs ignore the order of attack steps.
Thus, efficient analysis of \DAG-structured DATs is an important open problem.
Yet we do show how modular analysis works here as well, and can be used to speed up any algorithm.


\paragraph{Paper structure}
We list our contributions next, and give minimal background in \Cref{sec:AT}\textbf{.}\;
\Crefrange{sec:SAT}{sec:SAT_DAGs} study static  attack trees, and
\Crefrange{sec:DAT}{sec:DAT_DAGs} study dynamic attack trees.
\Cref{sec:mod} presents modular analysis, and \Cref{sec:order} studies Pareto fronts and $k$-top metrics.
The paper discuses related work and concludes in \Cref{sec:related_work,sec:conclu}.

\hlbox{%
\paragraph{Contributions of \cite{BS21}}
\label{sec:intro:contributions}
An earlier version of this paper was published in \cite{BS21}, whose key contributions are:
\begin{enumerate}[label=\textbf{\arabic*.},leftmargin=1.6em]
\item	An efficient and generic \BDD-based algorithm for \DAG-\SATs,
		working for absorbing semiring attribute domains with neutral elements (\Cref{sec:SAT_DAGs});
\item	A theorem proving that computing a minimal successful attack
		is NP-hard (\Cref{sec:SAT_DAGs:complexity});
\item	The adaptation of semiring attribute domains to \DATs,
		defining general metrics on \DATs (\Cref{sec:DAT});
\item	A bottom-up algorithm for tree-structured \DATs (\Cref{sec:DAT_trees});
\item	A \BDD-based algorithm to compute $k$-top metrics on \SATs
		(\Cref{sec:ktop});
\item	Future directions to analyse \DAG-\DATs efficiently
		(identified as an open problem in \Cref{sec:DAT_DAGs}).
\end{enumerate}}

\vspace{-1ex}

\hlbox{%
\paragraph{New contributions in this version}
Contributions of the current paper over \cite{BS21} are:
%
\begin{enumerate}[label=\textbf{\arabic*.},leftmargin=1.6em]
\setcounter{enumi}{6}
\item   The result that for idempotent semiring domains, bot\-tom-\-up algorithms---of linear runtime---work even for \DAG-structured \SATs (\Cref{sec:idempotent});
\item   A unified method to calculate
		multiple metrics sim\-ul\-tan\-eously---as needed by Pareto fronts, $k$-top metrics, and uncertainty sets---exploiting ordered attribute domains (\Cref{sec:order});
\item   General semantics for \DATs
		that cover the entire universe of models, also less restrictive than \cite{BS21}
		for the well-formedness criterion
		(\Cref{sec:DAT:semantics});
\item   A formalised modular analysis technique to improve the runtime of DAG algorithms, by combining analysis results for independent subtrees (\Cref{sec:mod});
\item   An improved description of the \BDD-based algorithm, illuminating the intuition behind the method (\Cref{alg:bottom_up_BDD} in \Cref{sec:SAT_DAGs:BDDs}).
\end{enumerate}
\smallskip
We place ourselves in the literature in \Cref{tab:all_algos_intro} and \Cref{sec:related_work}.
}

\section{Attack Trees}
\label{sec:AT}

\subsection{Attack tree models}
\label{sec:AT:models}

Syntactically, an attack tree is a rooted \DAG that models an undesired event caused by a malicious party, e.g.\ a theft.
\ATs show a top-down decomposition of a top-level attack---the root of the \DAG---into simpler steps.
The leaves are basic steps carried out by the attacker.
The nodes between the basic steps and the root are intermediate attacks, and are labelled with gates to indicate how its input nodes (children) combine to make the intermediate attack succeed.

\paragraph{Basic Attack Steps}
The leaves of the \AT represent indivisible actions carried out by the attacker, e.g.\ smash a window,
decrypt a file, 
etc.
These \BAS nodes can be enriched with attributes, such as its execution time, the cost incurred, and the probability with which the \BAS occurs.
We model attributes via an attribution function $\attrOp\from\BAS\to\Vdom$.

\paragraph{Gates}
Non-leaf nodes serve to model intermediate steps that lead to the top-level goal.
Each has a logical \emph{gate} that describes how its children combine to make it succeed:
an \OR gate means that the intermediate attack will succeed if any of its child nodes succeeds;
an \AND gate indicates that all children must succeed, in any order or possibly in parallel;
a \SAND gate (viz.\ sequential-\AND) needs all children to succeed sequentially in a left-to-right order.



\begin{figure}
	\centering
	\def\HEIGHT{23ex}
	\begin{subfigure}[b]{.48\linewidth}
		\centering\hspace*{-2em}%
		\includegraphics[height=\HEIGHT]{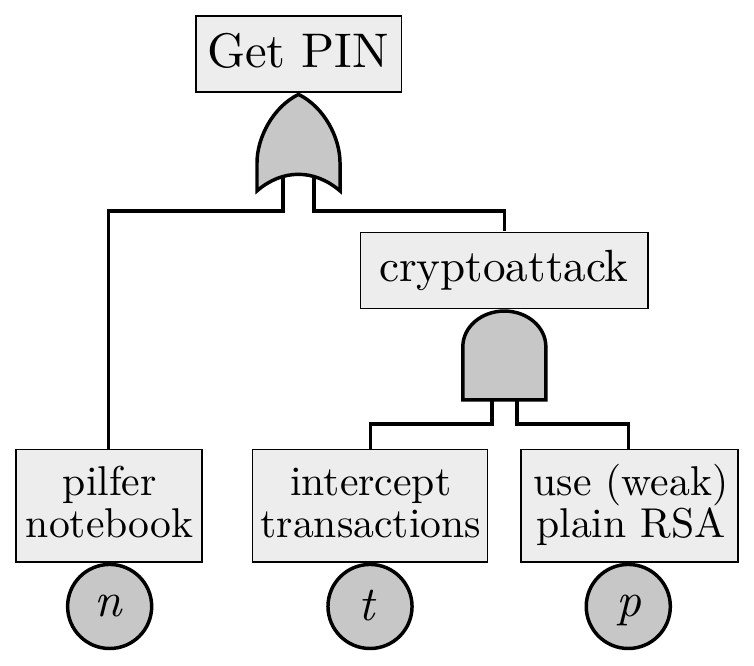}
		\caption{A static-tree \AT: \sampleTs\hspace*{2em}}
		\label{fig:AT:example:static}
	\end{subfigure}
	\hspace{-1em}
	\begin{subfigure}[b]{.48\linewidth}
		\centering
		\includegraphics[height=\HEIGHT]{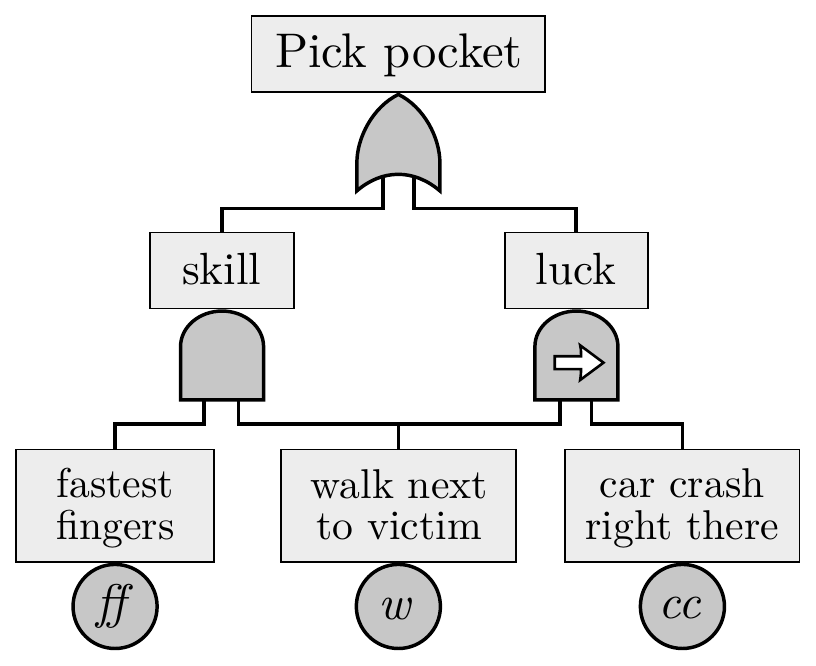}
		\caption{A dynamic-\DAG \AT: \sampleTd}
		\label{fig:AT:example:dynamic}
	\end{subfigure}
	\caption{Attack tree models}
	\label{fig:AT}
	\vspace{-2ex}
\end{figure}

\begin{example}
	\label{ex:running_examples}
	\Cref{fig:AT:example:static} shows a static attack tree, \sampleTs, that models how a \acronym{pin} code can be obtained by either pilfering a notebook, or via a cryptographic attack.
	The pilfering is considered atomic, while the cryptoattack consists of two steps which must both succeed: intercepting transactions, and abusing weak \acronym{rsa} encryption. 
	Note that \sampleTs has a plain tree structure.
	Instead, \Cref{fig:AT:example:dynamic} shows a dynamic attack tree, \sampleTd, with a \DAG structure.
	Its \TLA is to pick a pocket, which is achieved either by having ``skill'' or ``luck.''
	In both cases the attacker must walk next to the victim, so these gates share the \BAS child $w$, making \sampleTd not a tree.
	In the case of ``luck'' the order of events matters: if the attacker first walks next to the victim and then a traffic accident happens, the pick-pocket succeeds.
	Thus, this intermediate attack is modelled with a \SAND gate.
	Instead, ``fastest fingers'' is an inherent attacker flair that is always present.
	It is thus meaningless to speak of an order w.r.t.\ the attacker-victim encounter, so an \AND gate is used.
\end{example}

\paragraph{Security metrics}
A key goal in quantitative security analysis is to compute relevant \emph{security metrics}, which quantify how well a system performs in terms of security.
Typical examples are the cost of the cheapest attack, the probability of the most likely one, the damage produced by the most harmful one, and combinations thereof.
Security metrics for \ATs are typically obtained by combining the attribute values $\attr{a}\in\Vdom$ assigned to each $a\in\BAS$.
For instance, the cheapest attack is that for which the sum of the cost of its BASes is minimal.
\emph{The topic of this paper is how to compute large classes of security metrics in generic and efficient ways}.

\subsection{Attack tree syntax}
\label{sec:AT:syntax}

\ATs are rooted \DAG{s} with typed nodes: we consider types $\mathbb{T}=\{\tBAS,\tOR,\tAND,\tSAND\}$.
The edges of an \AT are given by a function $\chOp$ that assigns to each node its (possibly empty) sequence of children.
We use set notation for sequences, e.g.\ $e\in(e_1,\ldots,e_m)$ means $\exists i.\,e_i=e$, and we denote the empty sequence by $\varepsilon$.

\begin{definition}
	\label{def:AT:syntax}
	An \emph{attack tree} is a tuple
	\mbox{$\T=\left(\ATnodes,\typOp,\chOp\right)$} where:
	\begin{itemize}[topsep=.5ex,parsep=.1ex,itemsep=0pt]
	\item	$\ATnodes$ is a finite set of \emph{nodes};
	\item	$\typOp\from \ATnodes\to\mathbb{T}$
			gives the \emph{type} of each node;
	\item	$\chOp\from \ATnodes\to\ATnodes^\ast$
			gives the sequence of \emph{children} of a node.
	\end{itemize}
	Moreover, \T satisfies the following constraints:
	\begin{itemize}[topsep=.5ex,parsep=.1ex,itemsep=0pt]
	\item	$(\ATnodes,\mathobject{E})$ is a connected \DAG, where\\
			\phantom{.}\hfill%
			$\mathobject{E}=\left\{(v,u)\in\ATnodes^2\mid u\in\child{v}\right\}$;
	\item	\T has a unique root, denoted \ATroot:\\
			\phantom{.}\hfill%
			$\exists!\,\ATroot\in \ATnodes.~%
			\forall v\in \ATnodes.~\ATroot\not\in\child{v}$;
	\item	{$\normalfont\BAS_{\T}$} nodes are the leaves of \T:\\
			\phantom{.}\hfill%
			$\forall v\in \ATnodes.~%
			{\type{v}=\tBAS} \Leftrightarrow {\child{v}=\varepsilon}$.
	\end{itemize}
\end{definition} 

We omit the subindex \T if no ambiguity arises, e.g.\ an attack tree $\T=(\ATnodes,\typOp,\chOp)$ defines a set $\BAS\subseteq\ATnodes$ of basic attack steps.
If \mbox{$u\in\child{v}$} then $u$ is called a \emph{child} of $v$, and $v$ is a \emph{parent} of $u$.
We write $v=\AND(v_1,\ldots,v_n)$ if ${\type{v}=\tAND}$ and ${\child{v}=(v_1,\ldots,v_n)}$, and analogously for $\OR$ and $\SAND$, denoting $\child{v}_i=v_i$.
Moreover we denote the universe of \ATs by \allATs and call $\T\in\allATs$ \emph{tree-structured} if
${\forall v,u\in\ATnodes. \child{v}\cap\child{u}=\varepsilon}$;
else we say that \T is \emph{DAG-structured}.
If $\tSAND \notin \type{N}$ we say that \T is a \emph{static attack tree} (\SAT); else it is a \emph{dynamic attack tree} (\DAT).



\section{Static Attack Trees}
\label{sec:SAT}


In the absence of \SAND gates the order of execution of the BASes is irrelevant.
This allows for simple semantics given in terms of a Boolean function called \emph{structure function}.
The computation of security metrics, however, crucially depends on whether the AT is tree- or DAG-structured.

\subsection{Semantics for static attack trees}
\label{sec:SAT:semantics}

The semantics of a \SAT is defined by its successful attack scenarios, in turn given by its structure function.
First, we define the notions of attack and attack suite.

\begin{definition}
	\label{def:SAT:attack}
	An \emph{attack scenario}, or shortly an \emph{attack},
	of a static \AT \T is a subset of its basic attack steps:
	\mbox{$\attack\subseteq\BAS_T$}.
	An \emph{attack suite} is a set of attacks
	\mbox{$\suite\subseteq2^{\BAS_T}$}.
	We denote by 
	$\allAttacks_T={2^{\BAS_T}}$
	the universe of attacks of \T, and by 
	$\allSuites_T=2^{2^{\BAS_T}}$
	the universe of attack suites of \T.
	We omit the subscripts when there is no confusion.
\end{definition}

Intuitively, an attack suite $\suite\in\allSuites$ represents different ways in which the system can be compromised.
From those, one is interested in attacks $\attack\in\suite$ that actually represent a threat.
For instance for \sampleTs in \Cref{ex:running_examples} one such attack is $\{t,p\}$.
In contrast, $\{t\}$ is an attack that does not succeed, i.e.\ it does not reach the top-level goal.
The structure function $\sfunT(v,\attack)$ indicates whether the attack $\attack\in\allAttacks$ succeeds at node $v\in\ATnodes$ of \T.
For Booleans we use $\BB=\{\top,\bot\}$.

\begin{definition}
	\label{def:SAT:sfun}
	The \emph{structure function} $\sfunT\from\ATnodes\times\allAttacks\to\BB$
	of a static attack tree \T is given by:
	\begin{align*}
	  \sfunT(v,\attack) =&
	  \begin{cases}
		\top  & \parbox{55pt}{if~$\type{v}=\tOR$}~\text{and}~%
				\exists u\in\child{v}.\sfunT(u,\attack)=\top,\\
		\top  & \parbox{55pt}{if~$\type{v}=\tAND$}~\text{and}~%
				\forall u\in\child{v}.\sfunT(u,\attack)=\top,\\
		\top  & \parbox{55pt}{if~$\type{v}=\tBAS$}~\text{and}~%
				v\in\attack,\\
		\bot  & \text{otherwise}.
	  \end{cases}
	\end{align*}
\end{definition}



An attack \attack is said to \emph{reach} a node $v$ if $\sfunT(v,\attack) = \top$, i.e. it makes $v$ succeed. If no proper subset of \attack reaches $v$, then \attack is a \emph{minimal attack on $v$}. The set of attacks reaching $v$ is denoted $\Suc{v}$, and the set of minimal attacks on $v$ is denoted $\dsem{v}$. We define $\sfunT(\attack) \doteq \sfunT(\ATroot,\attack)$, and attacks that reach $\ATroot$ are called \emph{succesful}. We write $\dsem{\T} \doteq \dsem{\ATroot}$ and $\Suc{\T} \doteq \Suc{\ATroot}$. Furthermore, the minimal attacks on \ATroot (i.e.\ the minimal succesful attacks) are called \emph{minimal attacks}. The minimal attacks relate to the structure function in the sense that
\begin{equation*}
	\sfunT(v,\attack) =
		\bigvee_{\attack' \in \dsem{v}}
		\bigwedge_{a \in \attack'}
		\delta_{a}(\attack),
\end{equation*}
where $\delta_{a}(\attack) = \top$ iff $a \in \attack$. This can also be applied the other way around: we can represent \T by a propositional formula $L_T$ by replacing each $a\in\BAS$ with the atom $\delta_a$, and each node with its corresponding logical connector. 
Then the conjunctions in the minimal disjunctive normal form of $L_T$ correspond to the minimal attacks on \T.
 
\begin{definition}
	\label{def:propform}
	For $v \in N$, the propositional formula $L_{T,v}$ with atoms in $\{\delta_a \mid a \in \BAS_{\T}\}$ is given by
	\begin{equation*}
	L_{T,v} = \begin{cases}
		\bigvee_{w \in \child{v}} L_{T,w}  & \parbox{55pt}{if~$\type{v}=\tOR$},\\
		\bigwedge_{w \in \child{v}} L_{T,w}  & \parbox{55pt}{if~$\type{v}=\tAND$},\\
		\delta_v  & \parbox{55pt}{if~$\type{v}=\tBAS$}.
	  \end{cases}
	\end{equation*}
	Furthermore, we define $L_{T} := L_{\T,\ATroot}$\,.
\end{definition}


\SATs are \emph{coherent} \cite{BP75}, meaning that adding attack steps preserves success: if \attack is successful then so is $\attack\cup\{a\}$ for any $a\in\BAS$.
Thus, the suite of successful attacks of an \AT is characterised by its minimal attacks.
This was first formalised in \cite{MO06}, and is called multiset semantics in \cite{WAFP19}:
\begin{definition}
	\label{def:SAT:semantics}
	The \emph{semantics of a \SAT} \T is its suite of minimal attacks $\dsem{\T}$.
\end{definition}

\begin{example}
	\label{ex:SAT:semantics}
	The SAT in \Cref{ex:running_examples}, \sampleTs
	(\Cref{fig:AT:example:static}), has three successful attacks:
	$\{n\}$, $\{t,p\}$, and $\{n,t,p\}$. The first two are minimal,
	so we have: $\ssem{\sampleTs}=\{\{n\},\{t,p\}\}$.
\end{example}

An alternative characterisation of this semantics for tree-structured \SATs is shown as \Cref{lemma:ssem}, which also provides the key argument for correctness of the bottom-up procedure (\Cref{alg:bottom_up_SAT} in \Cref{sec:SAT_trees}). 
For tree-structured \SATs, \Cref{lemma:ssem} can be used to compute the semantics of \Cref{def:SAT:semantics} by recursively applying cases \textit{\ref{lemma:ssem:BAS})--\ref{lemma:ssem:AND})} to $\ssem{\ATroot} = \ssem{\T}$.
However, \BDD representations provide more compact encodings of this semantics (see \Cref{sec:SAT_DAGs}).

We formulate \Cref{lemma:ssem} for binary \ATs; its extension to arbitrary trees is straightforward but notationally cumbersome.
The proof of \Cref{lemma:ssem} is given in \Cref{app:semantics}.

\begin{lemma}
	\label{lemma:ssem}
	Consider a \SAT with nodes $a\in\BAS, v_1, v_2\in\ATnodes$. Then:
	\begin{enumerate}
	\def\REF#1{\textit{\ref{#1})}}
	\item	$\ssem{a} = \{ \{a\}\}$;
			\label{lemma:ssem:BAS}
	\item	$\ssem{\OR(v_1,v_2)} \subseteq \ssem{v_1} \cup \ssem{v_2}$;
			\label{lemma:ssem:OR}
	\item	$\ssem{\AND(v_1,v_2)} \subseteq \{ \attack_1\cup \attack_2 \mid
				\attack_1\in\ssem{v_1} \land \attack_2\in\ssem{v_2} \}$;
			\label{lemma:ssem:AND}
	\item All RHS of cases \ref{lemma:ssem:BAS}--\ref{lemma:ssem:AND} consist of succesful attacks.
	\end{enumerate}
	If $T$ is tree-structured then furthermore:
	\begin{enumerate}
	 \setcounter{enumi}{4}
	\item In cases \ref{lemma:dsem:OR} and \ref{lemma:dsem:AND}
			the \ssem{v_i} are disjoint, and in case \ref{lemma:dsem:AND}
			moreover the $A_i$ are disjoint;
			\label{lemma:ssem:disjoint}
	\item Equality holds in cases \ref{lemma:dsem:OR} and \ref{lemma:dsem:AND}.
	\end{enumerate}
\end{lemma}



\subsection{Security metrics for static attack trees}
\label{sec:SAT:metrics}

\Cref{lemma:ssem} allows for \emph{qualitative analyses}, i.e.\ finding the successful attacks of minimal size.
To enable \emph{quantitative analyses}, i.e.\ computing security metrics such as the minimal time and cost among all attacks, all \BASes are enriched with attributes.
We define security metrics in three steps:
first an attribution $\attrOp$ assigns a value to each \BAS;
then a security metric $\metrAOp$ assigns a value to each attack scenario;
and finally the metric $\metrOp$ assigns a value to each attack suite.


\begin{definition}
	\label{def:metric}
	Given an \AT \T and a set \Vdom of values:
	\begin{enumerate}
	\item	an \emph{attribution} $\attrOp\from\BAS_{\T}\to\Vdom$
			assigns an \emph{attribute value} \attr{a},
			or shortly an \emph{attribute}, to each basic attack step $a$;
	\item	a \emph{security metric} refers both to a function
			$\metrAOp\from\allAttacks_{\T}\to\Vdom$ that assigns a value
			$\metrA{\attack}$ to each attack $\attack$;
			\newline
			and to a function $\metrOp\from\allSuites_{\T}\to\Vdom$ that
			assigns a value $\metr{\suite}$ to each attack suite $\suite$.
	\end{enumerate}
	We write $\metr{\T}$ for $\metr{\ssem{\T}}$, setting the metric of an \AT to the metric of its suite of minimal attacks. 
\end{definition}

\begin{example}
	\label{ex:metric}
	Let $\Vdom=\NN$ denote time, so that \attr{a} gives the time required
	to perform the basic attack step $a$.
	Then the time needed to complete an attack \attack can be given by
	$\metrA{\attack} = \sum_{a\in\attack} \attr{a}$,
	and the time of the fastest attack in a suite \suite is
	$\metr{\suite} = \min_{\attack\in\suite} \metrA{\attack}$.
	If instead $\Vdom=[0,1]\subseteq\RR$ denotes probability,
	then the probability of an attack is given by
	$\metrA{\attack} = \prod_{a\in\attack} \attr{a}$,
	and the probability of the likeliest attack in a suite is
	$\metr{\suite} = \max_{\attack\in\suite} \metrA{\attack}$.
\end{example}

\Cref{def:metric} gives a very general notion of metric.
For a more operational definition---which enables computation for static \ATs, but does not depend on their tree/\DAG-structure---one must resort to the semantics.
For this we follow an approach along the lines of Mauw and Oostdijk \cite{MO06}.
Namely, we define a \emph{metric function} $\metrOp\from \allATs \to\Vdom$ that yields a value for each \SAT based on its semantics, an attribution, and two binary operators $\operOR$ and $\operAND$.


\begin{definition}
	\label{def:SAT:metric}
	\def\VH{\vphantom{\bigoperAND_{\ssem{\T}}}}
	Let \Vdom be a set:
	\begin{enumerate}
	\item an \emph{attribute domain} over \Vdom is a tuple $\domain=(\Vdom,\operOR,\operAND)$,
	whose \emph{disjunctive operator} \mbox{$\operOR\from\Vdom^2\to\Vdom$},
	and \emph{conjunctive operator} \mbox{$\operAND\from\Vdom^2\to\Vdom$},
	are associative and commutative;
	\item the attribute domain is a \emph{semiring}\footnote{Since we require $\operAND$ to be commutative, $\domain$ is in fact a commutative semiring. Rings often include a neutral element for disjunction and an absorbing element for conjunction, but these are not needed in \Cref{def:SAT:metric}.}
	if $\operAND$ distributes over $\operOR$, i.e.
	$\forall\,x,y,z\in\Vdom.\;%
		x\operAND(y\operOR z)=(x\operAND y)\operOR(x\operAND z)$;
	\item
	let \T be a static \AT and $\attrOp$ an attribution on \Vdom.
	The \emph{metric for \suite} associated to $\attrOp$ and $\domain$ is given by:\footnote{This expression motivates our notation $\widehat{\alpha}$ and $\widecheck{\alpha}$, which are similar to $\operAND$ and $\operOR$, respectively. These notations, in turn, were chosen in \cite{MO06} to be similar to $\wedge$ and $\vee$, respectively; the connection between these operators and the logical connectors is expressed in \Cref{theo:bottom_up_SAT}.}
	\begin{align*}
		\metr{\suite} &~=
			\underbrace{\VH\bigoperOR_{\attack\in\suite}\:}_{\mathlarger\metrOp}\,
			\underbrace{\VH~\bigoperAND_{a\in\attack}~}_{\mathlarger\metrAOp}
			\attr{a}.
	\end{align*}
	\item The metric for $T$ associated to $\attrOp$ and $\domain$ is given by $\metr{\T} \doteq \metr{\dsem{T}}$.
	\end{enumerate}
\end{definition}



\begin{example}
	\label{ex:SAT:metric}
	\setlength{\abovedisplayskip}{1.5ex}  
	The fastest attack time metric from \Cref{ex:metric} comes from the semiring attribute domain $(\mathbb{N},\min,+)$; indeed, the time for an attack is the sum of the attack times for all constituting BASes ($\operAND = +$), and the attack time for the \AT is the time of the fastest attack ($\operOR= \min$). Similarly, the highest attack probability comes from the semiring attribute domain $([0,1],\max,\cdot)$. Consider the SAT $\sampleTs=\OR\big(n,\AND(t,p)\big)$ from \Cref{fig:AT:example:static}. These two metrics can be calculated as follows.
	\begin{enumerate}[leftmargin=1.3em]
	\item \emph{Fastest attack}: Recall that $\ssem{\sampleTs}=\{\{n\},\{t,p\}\}$, and consider an attribution ${\attrOp}=\{{n\mapsto1}, {t\mapsto100}, {p\mapsto0}\}$. Then:
	\begin{align*}
	\metr{\sampleTs}
		&=~ \metrA{\{n\}} \operOR \metrA{\{t,p\}} \\
		&=~ \attr{n}\operOR\big(\attr{t}\operAND\attr{p}\big) \\
		&=~ \min(1,100+0) ~=~ 1.
	\end{align*}
	\item \emph{Most probable attack:} Now consider an attribution ${\attrOp'}=\{{n\mapsto0.07},{p\mapsto0.01},$ $ {t\mapsto0.95}\}$ for the same tree:
	\begin{align*}
	\metrOp'(\sampleTs)
		&= \attrOp'(n) \operOR' \big(\attrOp'(t) \operAND' \attrOp'(p)\big)\\
		&= \max(0.07,0.95 \cdot 0.01) ~=~ 0.07.
	\end{align*}
	\end{enumerate}

\end{example}


\section{Computations for tree-structured SATs}
\label{sec:SAT_trees}

\Cref{ex:SAT:metric} illustrates how to compute metrics for SATs using \Cref{def:SAT:metric}.
However, this method requires to first compute the semantics of the attack tree, which is \emph{exponential} in the number of nodes \card{\ATnodes}---see \Cref{theo:NP_hard} or \cite{KW18}.

A key result in  \cite{MO06} is that metrics defined on attribute domains $(\Vdom,\operOR,\operAND)$ that are semirings, can be computed via a bottom-up algorithm that is \emph{linear} in $\card{\ATnodes}+\card{E}$ (assuming constant time complexity of $\operOR$ and $\operAND$) as long as the static \AT has a proper tree structure.
We repeat this result here, giving a more direct proof of correctness. We extend the result to dynamic attack trees in \Cref{sec:SAT_DAGs}.

\subsection{Bottom-up algorithm}
\label{sec:SAT_trees:algorithm}

First we formulate the procedure as \Cref{alg:bottom_up_SAT}, which propagates the attribute values from the leaves of the \SAT to its root, interpreting \OR gates as $\operOR$ and \AND{s} as $\operAND$.
This algorithm is linear in $|N|+|E|$ since each node $v$ in the tree \T is visited once, and at $v$ we have $\deg(v)-1$ computation steps.
%
%
\Cref{alg:bottom_up_SAT} can be called on any node of \T: to compute the metric \metr{\T} it must be called on its root node \ATroot.


\begin{algorithm}
	\KwIn{Static attack tree $\T=(\ATnodes,\typOp,\chOp)$,\newline
	      node $v\in\ATnodes$,\newline
	      attribution $\attrOp$,\newline
	      semiring attribute domain $\domain=(\Vdom,\operOR,\operAND)$.}
	\KwOut{Metric value $\metr{\dsem{v}}\in\Vdom$.}
	\BlankLine
	\uIf{$\type{v}=\tOR$}{%
		\Return{$\bigoperOR_{u\in\child{v}}
		         \BUSAT(\T,u,\attrOp,\domain)$}
	} \uElseIf{$\type{v}=\tAND$}{%
		\Return{$\bigoperAND_{u\in\child{v}}
		         \BUSAT(\T,u,\attrOp,\domain)$}
	} \Else(\tcp*[h]{$\type{v}=\tBAS$}) {%
		\Return{\attr{v}}
	}
	\caption{\BUSAT for a tree-structured \SAT \T}
	\label{alg:bottom_up_SAT}
\end{algorithm}

We state the correctness of \Cref{alg:bottom_up_SAT} in \Cref{theo:bottom_up_SAT}, which we prove in \Cref{app:metrics}.
This result was proven in \cite{MO06} via rewriting rules for \ATs with a slightly different structure denoted ``bundles''.
Our result concerns attack trees in the syntax from \Cref{def:AT:syntax}, which is more conforming to the broad literature \cite{Wei91,Sch99,BLP+06,JW08,KMRS11,BET13,KRS15}.

\begin{theorem}
	\label{theo:bottom_up_SAT}
	Let \T be a static \AT with tree structure,
	$\attrOp$ an attribution on \Vdom,
	and $\domain=(\Vdom,\operOR,\operAND)$ a semiring attribute domain.
	Then $\metr{\T} = \BUSAT(\T,\ATroot,\attrOp,\domain)$.
\end{theorem}
\subsection{Metrics as semiring attribute domains}
\label{sec:SAT_trees:metrics}

Many relevant metrics for security analyses on \SATs can be formulated as semiring attribute domains.
\Cref{tab:SAT:metric} shows examples, where $\NN_\infty=\NN\cup\{\infty\}$ includes $0$ and $\infty$.
For instance ``min cost'' can be formulated in terms of ${(\NN_\infty,\min,+)}$, which is a semiring attribute domain because $+$ distributes over $\min$, i.e.\ $a+\min(b,c) =\min(a+b,a+c)$ for all ${a,b,c\in\NN_\infty}$.
Also, attribute domains can handle \SAND gates if the execution order is irrelevant for the metric, which happens e.g.\ for min skill and max damage.

Besides the metrics of \Cref{tab:SAT:metric}, one can also be interested in derived concepts, such as the Pareto front of multiple metrics, or uncertainty sets when the attribute values of \BASes are unknown. We show that such concepts fit into the semiring attribute domain framework in \Cref{sec:order}.

\paragraph{Non-semiring metrics}
However, some meaningful metrics do fall outside this category. 
For instance and as noted in~\cite{MO06}, the cost to defend against all attacks is represented by $(\NN_\infty,+,\min)$; but \Cref{alg:bottom_up_SAT} cannot compute this metric because $\min$ does not distribute over $+$.
Less well-known is that total attack probability---i.e.\ $\metr{\T}=\sum_{\attack\in\Suc{\T}} \metrA{\attack}$ where $\metrA{\attack}=\big(\prod_{a\in \attack} \attr{a}\big) \cdot \big(\prod_{a\not\in \attack} (1-\attr{a})\big)$---can neither be formulated as an attribute domain.
The problem is that $\metrA{\attack}$ does not have the shape $\bigoperAND_{a\in\attack} \attr{a}$, and that the sum is taken over all succesful attacks rather than just the minimal ones.
Interestingly though, this probability can still be computed via a bottom-up procedure by taking $\metr{\AND(v_1,v_2)} = \metr{\Suc{v_1}} \cdot \metr{\Suc{v_2}}$ and $\metr{\OR(v_1,v_2)} = \metr{\Suc{v_1}} + \metr{\Suc{v_2}} - \metr{\Suc{v_1}\cap\Suc{v_2}}$.


\begin{table}
  \centering
  \begin{tabular}{l>{}c>{\quad}c>{\quad}c}
	\toprule
	\scshape Metric       & $\Vdom$           & $\operOR$ & $\operAND$ \\
	\midrule
	min cost              & $\NN_\infty$      & $\min$ & $+$    \\
	min time (sequential) & $\NN_\infty$      & $\min$ & $+$    \\
	min time (parallel)   & $\NN_\infty$      & $\min$ & $\max$ \\
	min skill             & $\NN_\infty$      & $\min$ & $\max$ \\
	max challenge         & $\NN_\infty$      & $\max$ & $\max$ \\
	max damage            & $\NN_\infty$      & $\max$ & $+$    \\ 
	discrete prob.        & $[0,1]$       & $\max$ & $\cdot$ \\
	continuous prob.      & $\RR\to[0,1]$ & $\max$ & $\cdot$ \\
	\bottomrule
  \end{tabular}
  \caption{\SAT metrics with semiring attribute domains}
  \label{tab:SAT:metric}
\end{table}

\paragraph{Stochastic analyses}
Semirings form a bicomplete category, so finite and infinite products exist \cite{Mac71}.
This allows to propagate not only tuples of attribute values, but also functions over them. 
In particular, \emph{cumulative density functions} that assign a probability $t\mapsto P[X\leqslant t]$ constitute a semiring \cite{AHPS14}.
Such functions serve e.g.\ to consider attack probabilities (and cost, and damage) as functions that evolve on time.

\paragraph{Absorbing semirings}
Although in this paper we calculate metrics by considering all \emph{minimal} attacks, one could also simply consider all attacks.
For many metrics this does not make a difference: for example, the successful attack with minimal cost will always be a minimal attack, since adding \BASes can only increase the cost.
Therefore, in the calculation of min cost we may as well take the minimum over all successful attacks, rather than just minimal attacks.
On the other hand, when calculating max damage one will get a different answer when taking all attacks into account, as the full attack of all \BASes will do more damage than a smaller minimal attack.
The difference between these metrics can be described mathematically as follows.
We call a semiring attribute domain $\domain = (\Vdom,\operOR,\operAND)$ \emph{absorbing} if $x \operOR (x \operAND y) = x$ for all $x,y \in \Vdom$.
If \domain is absorbing and $\attrOp$ is an attribution into \domain, then $\metrA{A} \operOR \metrA{A'} = \metrA{A}$ for any two attacks with $A \subseteq A'$. It follows that for absorbing semiring attribute domains one has $\metr{\T} = \metr{\Suc{T}}$.
Note that all metrics in \Cref{tab:SAT:metric} are absorbing except for max challenge and max damage.


\section{Computations for DAG-structured SATs}
\label{sec:SAT_DAGs}

Attack trees with shared subtrees cannot be analysed via a bottom-up procedure on its (\DAG) structure, as we illustrate next in \Cref{ex:bottom_up_DAG}.
This is a classical result from fault tree analysis \cite{LGTL85}, rediscovered for attack trees e.g.\ in \cite{KW18}.

There are many methods to analyse \DAG-structured \ATs: see \Cref{tab:all_algos_intro} for contributions over the last 15 years, including \cite{AGKS15,BET13,DMCR06,JW08,KW18}.
These methods are often geared to specific metrics, e.g.\ cost, time, or probability \cite{BLP+06,JW08,AHPS14}.
Others use general-purpose techniques of high complexity and low efficiency, such as model checking \cite{DMCR06,KRS15}. 

We present a novel algorithm based on a binary decision diagram (\BDD) representation of the structure function of the \SAT.
\BDDs offer a very compact encoding of Boolean functions \cite{Bry86}, and are heavily used in model checking \cite{HS99,BK08,KP12}, as well as for probabilistic fault tree analysis \cite{Rau93,RS15b}.

Our \BDD-based approach works for absorbing semiring attribute domains, with neutral elements for operators $\operOR$ and $\operAND$, regardless of the \AT structure.
It thus extends the generic and efficient result of \cite{MO06}---that works for tree-structure \SATs only---to \DAG-structured \SATs as well.

Our algorithm traverses the \BDD bottom-up and it is linear in its size.
However, \BDDs can be exponential in the tree size \cite{RD97}: but no asymptotically-faster algorithms exist, since the problem of computing metrics is NP-hard, as we show below.
Moreover, \BDDs are among the most efficient approaches in terms of performance of practical computation on Boolean formulae \cite{Bry86,BET13}.

\subsection{Computational complexity}
\label{sec:SAT_DAGs:complexity}

We first show why the bottom-up procedure cannot compute metrics for \ATs that have shared subtrees. 


\begin{example}
	\label{ex:bottom_up_DAG}
	\Cref{fig:bottom_up_DAG} shows how the bottom-up approach can fail when applied to \DAG-structured attack trees.
	Intuitively, the problem is that a visit to node $v$ in \Cref{alg:bottom_up_SAT}---or any bottom-up procedure that operates on the \AT structure---can only aggregate information on its descendants.
	So, the recursive call for $v$ cannot determine whether a sibling node (i.e.\ any node $v'$ which is not an ancestor nor a descendant of $v$) shares a \BAS descendant with $v$.
	As a result, recursive computations for both $v$ and $v'$ may select a shared descendant $b\in\BAS$, and use \attr{b} in (both) their local computations.
	This causes the miscomputation in \Cref{fig:bottom_up_DAG}.
\end{example}

\begin{figure}
  \vspace{-2ex}
  \centering
  \def\caja#1{\parbox[b]{.5em}{\centering$#1$}}

  \begin{subfigure}[b]{.27\linewidth}
	\includegraphics[width=.95\linewidth]{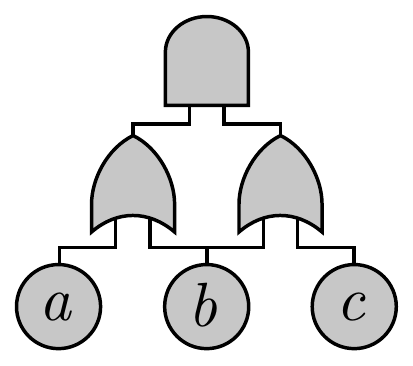}
  \end{subfigure}
  ~
  \begin{subfigure}[b]{.66\linewidth}
	\def\mp#1#2{\begin{minipage}[b]{#1\linewidth}{{#2}}\end{minipage}}
	\mp{.08}{%
		Let:\vspace{6.5ex}
	}
	\mp{.80}{%
	  \begin{align*}
		\attr{\caja{a}} &= 3  & \Vdom    &= \NN_\infty  \\[-.5ex]
		\attr{\caja{b}} &= 2  & \operOR  &= \min        \\[-.5ex]
		\attr{\caja{c}} &= 4  & \operAND &= {+}
	  \end{align*}
	}
	\mp{.08}{\hspace*{-2em}
	\tikz[baseline=-.4ex]{\draw[x=1ex,y=1ex,decorate,decoration={%
		                        brace,mirror,amplitude=5pt}] (0,0)--(0,8.1)
	                      node [midway,xshift=1.3em] {\domain};}}
	
	The cheapest attack is $\{b\}$: $\metrA{\{b\}}=2$.
	\vspace{0.3ex}
  \end{subfigure}
  
  \caption{Metrics cannot be computed bottom-up on \ATs with \DAG structure.
           For min cost in this \AT, \Cref{alg:bottom_up_SAT} yields:
           $\BUSAT(\T,\ATroot,\attrOp,\domain) = \min(3,2) + \min(2,4) = 4 \neq 2 = \metr{\T}$.
           The miscomputation stems from counting \attr{b} twice.}
  \label{fig:bottom_up_DAG}
\end{figure}


Known workarounds to this issue are keeping track of the \BAS selected by the metric at each step \cite{KW18,BK18}, and operating on the \AT semantics \cite{MO06}.
In all cases the worst-case scenario has exponential complexity on the number of nodes of the attack tree:
for \cite{KW18} this is in the algorithm input, i.e.\ determining the sets of necessary and optional clones;
instead for \cite{MO06} and our \Cref{def:SAT:metric} the complexity lies in the computation of the semantics.

In general, one cannot hope for faster algorithms: \Cref{theo:NP_hard} shows that the core problem---computing minimal attacks of \DAG-structured attack trees---is NP-hard even in the simplest structure: plain attack trees with \AND/\OR gates.
The proof (in \Cref{app:hard}) reduces the satisfiability of logic formulae in conjunctive normal form, to the computation of minimal attacks in general \SATs.

\begin{theorem}
	\label{theo:NP_hard}
	Given a DAG-structured static AT, the problem of computing any successful attack
	of minimal size is NP-hard.
\end{theorem}

From this we can show that calculating metrics on \DAG-structured \SATs is NP-hard. To do this, we define the following attribute domain: for a \SAT $T$, take $\Vdom = \mathscr{M}(\BAS_T)$, the set of multisets on $\BAS_T$. Let $\operAND = \uplus$, i.e.\ multiset union. We identify $\mathscr{M}(\BAS_T) \cong \NN^{\BAS_T}$, and $\uplus$ becomes $+$ under this identification. We furthermore choose an enumeration $\BAS_T = \{a_1,\ldots,a_n\}$, so that we may identify $\Vdom \cong \NN^n$. We then define a map $\mu\colon \Vdom \rightarrow \NN$ by
\begin{equation*}
\mu(c) = \prod_{i=1}^n p_i^{c_i},
\end{equation*}
where $p_i$ is the $i$-th prime; here an element $c \in \Vdom = \NN^n$ is determined by its coefficients $c_i \in \NN$. Define a linear order $\preceq$ on $V$ by $c \preceq c'$ iff either $\sum_i c_i < \sum_i c'_i, \textrm{ or } \sum_i c_i = \sum_i c'_i \textrm{ and } \mu(c) \leq \mu(c')$. Let $\operOR = \min$ w.r.t. $\preceq$. One can then prove (see Appendix 
\ref{app:metrcomp}) the following:

\begin{lemma} \label{lem:metrcomp}
$(V,\operOR,\operAND)$ is a semiring attribute domain. Furthermore, let $\alpha\colon \BAS_T \rightarrow \mathscr{M}(\BAS_T)$ be given by $\alpha(a) = \ldb a \rdb$. Then the multiset $\metr{T}$ is a set, and it is the succesful attack of minimal size.
\end{lemma}

\smallskip\noindent%
With \Cref{theo:NP_hard} this yields the following corollary:


\begin{corollary}
	\label{coro:NP_hard}
	Computing a metric for a semiring attribute domain
	in a \DAG-structured \SAT is NP-hard.
\end{corollary}

\subsection{Idempotent semiring attribute domains} \label{sec:idempotent}

While in general computing metrics on DAG-structured \SATs is hard, for some metrics the bottom-up algorithm still works. In \Cref{ex:bottom_up_DAG} it is seen that \BUSAT fails because some \BASes may be counted twice. However, when the operators $\operOR$, $\operAND$ are such that multiple occurences of a BAS in a formula does not impact the calculation, then this is not a problem. One can express this formally by the following definition:

\begin{definition}[Idempotent domain]
	\label{def:idempotent}
	A binary operator $\star$ on a set $X$ is called \emph{idempotent}
	if $x \star x = x$ for all $x \in X$.
	A semiring attribute domain $\domain = (\Vdom,\operOR,\operAND)$ is called
	idempotent if both operators $\operOR$ and $\operAND$ are idempotent.
\end{definition}

Idempotency of the domain is not enough for \BUSAT to work for DAG-structured \SATs; we also need \domain to be absorbing.
The reason for this is that in a \DAG, there might be \BASes that are not an element of any minimal attack, and hence are not present in the expression of $\metr{\T}$.
Nevertheless, these \BASes may still be used in the calculation of \BUSAT, as in the following example.
\begin{wrapfigure}[6]{r}{.18\linewidth}
	\vspace{1.5ex}
	\includegraphics[width=.85\linewidth]{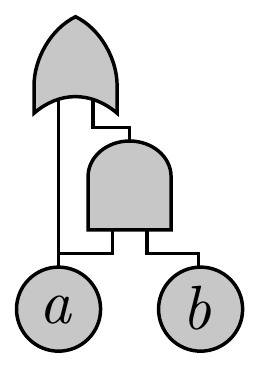}
\end{wrapfigure}

\begin{example}
    Consider the \SAT to the right, $\T = \OR(a,\AND(a,b))$, and consider the domain $\domain = (\NN_{\infty},\max,\max)$ for the max challenge metric; \domain is idempotent.
    Take $\attr{a} = 1$, $\attr{b} = 2$: since $\dsem{\T} = \{\{a\}\}$ then $\metr{\T} = \attr{a} = 1$.
    However, \Cref{alg:bottom_up_SAT} calculates $\BUSAT(\T,\ATroot,\alpha,\domain) = \max(\attr{a},\max(\attr{a},\attr{b})) = 2$.
    The miscalculation comes from the fact that $b$ is not an element of any minimal attack.
\end{example}

Only when \domain is both idempotent and absorbing, then \BUSAT calculates $\metr{\T}$ correctly.
A motivating example is the domain $\domain = (\NN_{\infty},\min,\max)$, which represents the min time and min skill metrics; another example is $(\BB,\vee,\wedge)$ from \cite{KW18}.
The fact that \BUSAT works for idempotent absorbing domains has been proven for the parallel min time metric in \cite{LS2021}; the \namecref{theo:idempotent} below extends this result to the general case.
This result is similar to \cite{KW18}, where it is proven for Attack--\-Defense trees under mildly stronger assumptions on \domain, namely the existence of identity and absorbing elements in \Vdom.

\begin{theorem}
	\label{theo:idempotent}
	%
	%
	Let \T be a static \AT,
	$\attrOp$ an attribution on \Vdom,
	and $\domain=(\Vdom,\operOR,\operAND)$ a semiring attribute domain.
	If \domain is idempotent and absorbing then
	$\metr{\T} = \BUSAT(\T,R_T,\alpha,D)$.
\end{theorem}

\subsection{Binary decision diagrams}
\label{sec:SAT_DAGs:BDDs}

\BDDs offer a representation of Boolean functions that is often extremely compact. A \BDD is a rooted \DAG \bddf that, intuitively, represents a Boolean function $f\from \BB^n\to\BB$ over variables $\Vars=\{x_i\}_{i=1}^n$.
The terminal nodes of \bddf represent the outcomes of $f$: $\bot$ or $\top$.
A nonterminal node $w\in\BDDnodes$ represents a subfunction $f_w$ of $f$ via its Shannon expansion.
That means that $w$ is equipped with a variable $\BDDlab(w)\in\Vars$ and two children:
$\low(w)\in\BDDnodes$, representing $f_w$ in case that the variable $\BDDlab(w)$ is set to $\bot$;
and $\high(w)$, representing $f_w$ if $\BDDlab(w)$ is set to $\top$.


\begin{definition}
	\label{def:BDD}
	A \emph{BDD} is a tuple $\bddT[]=(\BDDnodes,\low,\high,\BDDlab)$
	over a set \Vars where:
	\begin{itemize}[topsep=.5ex,parsep=.1ex,itemsep=0pt]
    \item	The \emph{set of nodes} \BDDnodes is partitioned into
			terminal nodes (\BDDnodesT) and
			nonterminal nodes (\BDDnodesN);
    \item	$\low    \from \BDDnodesN \to \BDDnodes$
			maps each node to its \emph{low child};
    \item	$\high   \from \BDDnodesN \to \BDDnodes$
			maps each node to its \emph{high child};
    \item	$\BDDlab \from \BDDnodes   \to \{\bot,\top\}\cup\Vars$
			maps terminal nodes to Booleans,
			and nonterminal nodes to variables:\\[.5ex]
			${\BDDlab(w)\in\begin{cases}
			\{\bot,\top\} & \text{if}~w\in\BDDnodesT,\\
			\Vars         & \text{if}~w\in\BDDnodesN.
			\end{cases}}$
	\end{itemize}
	Moreover, \bddT[] satisfies the following constraints:
	\begin{itemize}[topsep=.5ex,parsep=.1ex,itemsep=0pt]
	\item	$(\BDDnodes,\mathobject{E})$ is a connected \DAG, where\\
			\phantom{.}\hfill%
			$\mathobject{E}=\{(w,w')\in\BDDnodes^2 \mid w'\in\{\low(w),\high(w)\}\}$;
	\item	\bddT[] has a unique root, denoted \BDDroot:\\
			\phantom{.}\hfill%
			$\exists!\,\BDDroot\in \BDDnodes.~%
			\forall w\in\BDDnodesN.~\BDDroot\not\in\{\low(w),\high(w)\}$.
	\end{itemize}
\end{definition}

Given a \BDD representing $f$, and $\boldsymbol{x}=(x_1,\ldots,x_n)\in\BB^n$, one calculates $f(\boldsymbol{x})$ by starting from the root of the \BDD, and at every node $w\in\BDDnodesN$ with $\BDDlab(w)=x_i$, proceed to $\high(w)$ if $x_i = \top$, and to $\low(w)$ if $x_i = \bot$. The terminal node one ends up in, is the value $f(\boldsymbol{x})$.

\paragraph{Reduced ordered BDDs}
We operate with \emph{reduced ordered \BDDs}, simply denoted \BDDs.
This requires a total order ${<}$ over the variables.
For \Cref{def:BDD} this means that:
\begin{itemize}[label=\textbullet]
\item	\Vars comes equipped with a total order,
		so \bddf is actually defined over a pair \poset[\Vars][{<}];
\item	the variable of a node is of lower order than its children:
		$\forall\,w\in\BDDnodesN.\,%
			\BDDlab(w)<\BDDlab(\low(w)),\BDDlab(\high(w))$;
\item	the children of nonterminal nodes are distinct nodes;
\item	all terminal nodes are distinctly labelled;
\item   nonterminal nodes are uniquely determined by their label and children:
$\forall w,w' \in W_n.\, ({\BDDlab(w) = \BDDlab(w')} \land {\low(w) = \low(w')} \land {\high(w) = \high(w')}) \Rightarrow w=w'$.
\end{itemize}
This has the following consequences in the \BDD:
\begin{itemize}[label=\textbullet]
\item	there are exactly two terminal nodes:
		$\BDDnodesT=\{\oldbot,\!\oldtop\}$,
		with ${\BDDlab(\oldbot)=\bot}$ and ${\BDDlab(\oldtop)=\top}$;
\item	the label of the root node \BDDroot has the lowest order;
\item	in any two paths from \BDDroot to $\oldbot$ or $\oldtop$,
		the order of the variables visited is (increasing and) the same.
\end{itemize}
Given the ordering $<$ on $\Vars$, there is a unique reduced ordered \BDD that represents $f$.

\paragraph{Encoding static \ATs as \BDDs}
The semantics of an \AT $T$ can be encoded by its \BDD. This is done for fault trees in \cite{Rau93}, and the method works identically for \ATs. One assumes an arbitrary order on $\BAS_T$, and creates the reduced ordered \BDD $B_T$ of the propositional formula $L_T$ of \Cref{sec:SAT:semantics}.

A Boolean vector $\boldsymbol{x}$ evaluates to $\top$ if its corresponding path in the \BDD ends up in the terminal node $\oldtop$ that is labelled $\top$.\!%
\footnote{%
To ease graphical interpretation, we identify the terminal node $\oldtop$ with its label $\top$ and thus speak of paths $\BDDroot[\bddT] \to \top$.
}
For \bddT, this means that an attack \attack is successful if and only if there is a path $p$ from \BDDroot[\bddT] to $\top$ such that the $\high$-edges traversed form a subset of \attack. This can be phrased as follows:

\begin{theorem}[\cite{Rau93}]
	\label{thm:BDDpaths}
	Let $P$ be the set of paths $\BDDroot[\bddT] \to \top$ in \bddT. Then the map
	\begin{align*}
		\pi\colon P &\rightarrow \allAttacks \\
		p &\mapsto \{\BDDlab(w) \in \BAS_T \mid (w,\high(w)) \in p\}
	\end{align*}
 satisfies $\ssem{T} \subseteq \textrm{im}(\pi)$, and $\pi(p)$ is succesful for all $p \in P$.
\end{theorem}

\smallskip

\begin{example}
	\label{ex:SAT:BDD}
	\opencutright  
	\renewcommand*{\windowpagestuff}{%
		\centering\includegraphics[width=.7\linewidth]{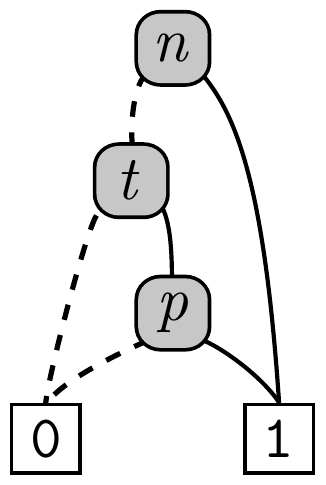}
	}
	\begin{cutout}{0}{.73\linewidth}{0pt}{8}
	Let $n<t<p$ in \sampleTs from \Cref{ex:running_examples}: the resulting \BDD $(\bddT[\sampleTs])$ is illustrated to the right.
	The children of a node appear below it (so the root is on top), and a dashed line from $w \rightarrow w'$ means $w'=\low(w)$, and a solid line means $w'=\high(w)$.
	The two paths $n \to \top$ correspond to the minimal attacks $\{n\},\{t,p\}$.
	\end{cutout}
\end{example}

An algorithm to find the (reduced ordered) \BDD of an \AT is given in \cite{Rau93}.
In the worst case, the size of the \BDD is exponential in the number of variables, i.e.\ the \BASes of the \AT \cite{Bry86}.
However, \DAG-structures that represent Boolean functions---such as fault trees and \ATs---often have small \BDD encodings \cite{BET13,RD97}.
The choice of the linear order $<$ on $\BAS_T$ impacts the size of the \BDD: finding the order that yields the smallest \BDD is NP-hard, but several heuristics exists to find a good order \cite{Val79,LWW14}.



\subsection{BDD-based algorithm for DAG-structured SATs}
\label{sec:SAT_DAGs:algorithm}

\Cref{alg:bottom_up_BDD} (on \cpageref{alg:bottom_up_BDD}) computes metrics for \DAG-structured attack trees --- a similar algorithm for fault trees and their failure probability metric was introduced in \cite{Rau93}.
Here, just like \BUSAT, \Cref{alg:bottom_up_BDD} requires $\domain = (\Vdom,\operOR,\operAND)$ to be a semiring attribute domain.
However, \Cref{alg:bottom_up_BDD} also requires the definition of neutral elements \ntOR and \ntAND for $\operOR$ and $\operAND$, i.e.\
${\forall x\in\Vdom.~x \operOR \ntOR = x \operAND \ntAND = x}$;
we write $\ndomain = (\Vdom,\operOR,\operAND,\ntOR,\ntAND)$ and call this a \emph{unital semiring}.
We furthermore require the domain to be absorbing --- see \Cref{sec:SAT_trees:metrics}.
These conditions are mild: neutral elements are common (a semiring without them can always be extended to have them) \cite{Mac71}.
Examples of neutral elements in \Cref{tab:SAT:metric} are $\ntOR=\infty$ and $\ntAND=0$ for min cost, and $\ntOR=0$ and $\ntAND=1$ for (max) discrete probability.
Moreover, most semiring metrics are absorbing, e.g.\ all in \Cref{tab:SAT:metric} are, except max challenge and max damage. 

\paragraph{The algorithm}
To explain the algorithm in more detail, we first introduce some notation.
For $w \in W_n$, let $P(w)$ be the set of paths $w \to \top$.
Furthermore, for $p \in P(w)$, define
\begin{align}
	\metrA{p} &= \bigoperAND_{\substack{w \in W_n}: (w,\high(w)) \in p} \alpha(\BDDlab(w)), \label{eq:bddalg1}\\
	\metr{w} &= \bigoperOR_{p \in P(w)} \metrA{p}. \label{eq:bddalg2}
\end{align}
By \Cref{thm:BDDpaths} each path in $P(R_{B_T})$ corresponds to a succesful attack, and $\metr{R_{B_T}}$ is calculated by performing $\operOR$ over the metric values of these attacks. These attacks include the minimal attacks by \Cref{thm:BDDpaths}, and because $D$ is absorbing all other attacks are irrelevant. Hence one has $\metr{R_{B_T}} = \metr{T}$. This value can be calculated in three steps:
\begin{enumerate*}[label=\textit{\arabic*.}]
\item	determine $P(\BDDroot[\bddT])$;
\item	for $p \in P(\BDDroot[\bddT])$, calculate $\metrA{p}$;
\item	compute $\metr{w}$.
\end{enumerate*}

Even though this approach works, it has some inefficiency built into it: there will typically be paths that share sections, and on these sections the above method calculates $\operAND$ twice. Therefore we instead use a bottom-up algorithm on the \BDD, where at every node we use $\operOR$ on the paths up to that node. The key observation is that paths $p \in P(w)$ either pass through $\low(w)$ or $\high(w)$, and from that point onwards it is an element of $P(\low(w))$ or $P(\high(w))$. Also, in the latter case, we need to add $\alpha(\BDDlab(w))$ to the \mbox{$\operAND$-ation} when calculating $\metrA{p}$, because that \BAS is included. By exploiting the distributivity of $\operOR$ over $\operAND$, one can show that
\begin{equation}
	\label{eq:BDDstep}
	\metr{w} = \metr{\low(w)} \operOR \big(\metr{\high(w)} \operAND \alpha(\BDDlab(w))\big).
\end{equation}
In \Cref{alg:bottom_up_BDD} we use a bottom-up method to calculate $\metr{\BDDroot[\bddT]}=\metr{\T}$ by repeatedly applying \eqref{eq:BDDstep} from \BDDroot[\bddT]. This is done until we reach the bottom nodes $\bot$ and $\top$, to which we assign the values $1_{\operOR}$ and $1_{\operAND}$, respectively.

\begin{algorithm}
	\KwIn{\BDD $\bddT=(\BDDnodes,\low,\high,\BDDlab)$,\newline
	      node $w\in\BDDnodes$,\newline
	      attribution $\attrOp$,\newline
	      \mbox{semiring attribute domain 
	      $\ndomain=(\Vdom,\operOR,\operAND,\ntOR,\ntAND)$.}}
	\KwOut{Metric value $\metr{\T}\in\Vdom$.}
	\BlankLine
	\uIf{$\BDDlab(w)=\bot$}{%
		\Return{\ntOR}
	} \uElseIf{$\BDDlab(w)=\top$}{%
		\Return{\ntAND}
	} \Else(\tcp*[h]{$w\in\BDDnodesN$}) {%
		\Return{${\BUBDD(\bddT,\low(w),\attrOp,\ndomain)} \operOR
			{\big(\BUBDD(\bddT,\high(w),\attrOp,\ndomain) \operAND
				\attr{\BDDlab(w)}\big)}$}
	}
	\caption{\BUBDD for a \DAG-structured \SAT \T}
	\label{alg:bottom_up_BDD}
\end{algorithm}

\medskip

\begin{example}
	\label{ex:SAT:bottom_up_BDD}
	\opencutright  
	\renewcommand*{\windowpagestuff}{%
		\centering\includegraphics[width=.7\linewidth]{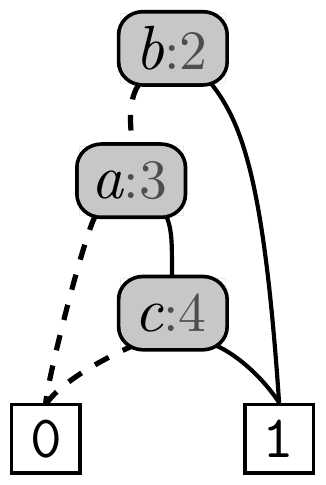}
	}
	\begin{cutout}{0}{.73\linewidth}{0pt}{7}
	For the \DAG-structured \SAT shown in \Cref{fig:bottom_up_DAG}, the order ${b<a<c}$ of its \BASes yields the \BDD to the right.
	Call it \bddT and let us use it to compute the min cost of \T like in \Cref{fig:bottom_up_DAG}, via the attribution $\attrOp=\{{a\mapsto3},{b\mapsto2},{c\mapsto4}\}$ and the domain $(\NN_\infty,\min,+)$.
	Moreover, to use \Cref{alg:bottom_up_BDD}, we use the neutral elements $\:\ntOR=\infty\:$ for $\,\min$ and $\ntAND=0$ for ${+}$, viz.\ we use the attribute domain $\,\ndomain=(\NN_\infty,\min,{+},\infty,0)\,$.
	\end{cutout}
	\noindent%
	Let the nonterminal nodes of \bddT be $\BDDnodesN=\{w_a,w_b,w_c\}$, and for $w\in\BDDnodes$ let $\BU(w)$ be shorthand for $\BUBDD(\bddT,w,\attrOp,\ndomain)$; this is equal to $\metr{w}$.
	We compute the metric:
	\begingroup
	\setlength{\abovedisplayskip}{1ex}
	\setlength{\belowdisplayskip}{1.5ex}
	\begin{align*}
	  \BU(\BDDroot[\bddT])
		&= \min\!\big(\BU(w_a) ,\, \BU(\oldtop)+\attr{b} \big)\\
		&= \min\!\big(\BU(w_a) ,\, \ntAND+2 \big)\\
		&= \min\!\big(\BU(w_a) ,\, 2 \big)\\
		&= \min\!\big(\min(\BU(\oldbot) , \BU(w_c)+\attr{a}) ,\, 2 \big)\\
		&= \min\!\big(\min(\ntOR , \BU(w_c)+3) ,\, 2 \big)\\
		&= \min\!\big(\BU(w_c)+3 ,\, 2 \big)\\
		&= \min\!\big(\min(\BU(\oldbot), \BU(\oldtop)+\attr{c})+3 ,\, 2 \big)\\
		&= \min\!\big(\min(\ntOR, \ntAND+4)+3 ,\, 2 \big)\\
		&= \min(4+3 ,\, 2) ~=~ 2.
	\end{align*}
	\endgroup
	To compute instead the (max) discrete probability we use the attribution $\attrOp'=\{{a\mapsto0.1},{b\mapsto0.05},{c\mapsto0.6}\}$ and the attribute domain $\ndomain'=({[0,1]_\QQ},\max,\cdot,0,1)$.
	Then computations are as before until the last line, which here becomes: ${\max\!\big(\attrOp'(c)\cdot\attrOp'(a),\,\attrOp'(b)\big)} = {\max(0.6\cdot0.1,0.05)} = 0.06$.
\end{example}

\Cref{theo:bottom_up_BDD} states the correctness of \Cref{alg:bottom_up_BDD}, i.e.\ that it yields the metric for a static \AT given in \Cref{def:SAT:metric} regardless of its structure. We prove \Cref{theo:bottom_up_BDD} in \Cref{app:metrics}.

\begin{theorem}
	\label{theo:bottom_up_BDD}
	\def\root{\ensuremath{\BDDroot[{\bddT}]}\xspace}
	\def\IC#1{\ensuremath{\mathrm{IC}_{#1}}\xspace}
	Let \T be a static \AT, \bddT its \BDD encoding over \poset[\BAS][{<}],
	$\attrOp$ an attribution on \Vdom,
	and $\ndomain=(\Vdom,\operOR,\operAND,\ntOR,\ntAND)$ an absorbing unital semiring attribute domain.
	Then $\metr{\T} = \BUBDD(\bddT,\root,\attrOp,\ndomain)$.
\end{theorem}

We also note that, when actually implementing \Cref{alg:bottom_up_BDD}, one can further optimize its efficiency via dynamic programming.
That is, by storing the calculated values of nodes so that this calculation is not repeated unnecessarily for nodes with multiple parents.
We leave such considerations out of \Cref{alg:bottom_up_BDD} to highlight the simple structure of the method and its relation to \cref{eq:BDDstep}.

\Cref{alg:bottom_up_BDD} has linear complexity in the size of \bddT, so the overall complexity of calculating \AT metrics via its \BDD is mainly determined by the size of \bddT.
As mentioned in \Cref{sec:SAT_DAGs:BDDs}, this is worst-case exponential, but in practice it is usually a lot faster.
This makes the \BDD approach a suitable heuristic for calculating \AT metrics.
In \Cref{sec:mod} we show how the performance can be further improved by incorporating modular analysis.


\subsection{Computing the \texorpdfstring{$\boldsymbol{k}$}{k}-top metric values}
\label{sec:ktop}

The approach described above can be extended to efficiently compute the $k$-top values for a given metric.
This problem asks not only the min/max value of the metrics from \Cref{tab:SAT:metric}, but also the next $k-1$ min/max values, e.g.\ the cost of the $k$ cheapest attacks, or the probability of the $k$ most likely ones. 


Formally, we can express $k$-top values in the language of multisets.
For a linearly ordered set $X$ and a multiset $M \in \NN^X$, we let $\operatorname{min}^k(M)$ be the multiset of the $k$ smallest elements of $M$, with $\operatorname{min}^k(M) \doteq M$ when $\card{M} \leqslant k$.
Given an attribution $\attrOp \from \BAS \to X$, the \emph{top-$k$ metric values of \T} are defined as $\operatorname{Top}_k(\T,\attrOp) = \operatorname{min}^k(\ldb \metrA{\attack} \mid \attack \in \dsem{\T}\rdb)$ with $\ldb\cdot\rdb$ denoting a multiset.
That is, $\operatorname{Top}_k(\T,\attrOp)$ is an element of $\mathscr{M}(V)$, viz.\ a multiset of (the $k$ smallest) values: it describes the $k$-top values when $\operOR = \min$ with respect to the order on $X$.
When $\operOR = \max$, the $k$-top values are defined analogously.

The multiset $\operatorname{Top}_k(\T,\alpha)$ can be computed from the \BDD using the $k$-shortest-paths algorithm for \DAGs; this is a well-known extension of the Dijkstra (or Thorup) algorithm \cite{Dij59,Tho99}. For a \DAG $G$ with edges weighted by the matrix $Q$, $\kshortest(G,Q,s,t,k,\circ)$ returns the multiset of weights of the $k$-shortest paths from a source node $s$ of $G$, to a target node $t$, using operator $\circ$ to accumulate weight.

By \Cref{thm:BDDpaths} one has that $\operatorname{Top}_k(\T,\attrOp) = \ldb \metrA{p} \mid p \in P(\BDDroot[\bddT])\rdb$. To apply the $k$-shortest path algorithm, we interpret $\metrA{p}$ as the (weighed) length of path $p$, which is done by assigning weight $\attr{\BDDlab(w)}$ to each edge $(w,\high(w))$ in $p$, and weight $1_{\operAND}$ to each edge $(w,\low(w))$. We accumulate these values---to compute the length of $p$---with operator~$\operAND$.
This way, the $k$ shortest paths are the $k$ paths with the smallest values $\metrA{p}$; when $\operOR = \max$, we invert the weights.\!%
\footnote{%
Technically, we assume that $V$ has an order-inverting bijection $-1$. If this is not the case, then one can just invert the linear order on $V$ when $\operOR = \max$.}
This is encapsulated as \Cref{alg:shortest_path_BDD}, whose correctness is given by \Cref{thm:BDDpaths} and the $\kshortest$ algorithm.




\begin{algorithm}
	\def\matrix{\mathobject{Q}}
	\def\sgn{\ensuremath{\mathit{sgn}}}
	\KwIn{\BDD $\bddT=(\BDDnodes,\low,\high,\BDDlab)$,\newline
	      number of values to compute $k\in\NN$,\newline
	      attribute domain $\domain=(\Vdom,\operOR,\operAND)$,\newline
	      attribution $\attrOp$.}
	\KwOut{$k$-top metric values of \T for $\attrOp$ and $\domain$.}
	\BlankLine
	\DontPrintSemicolon
	$\mathobject{E} := \bigcup_{w \in \BDDnodesN}
		\{(w,\low(w)),(w,\high(w))\}$ \\[.2ex]
	\matrix\ :=  $\card{\BDDnodes}\times\card{\BDDnodes}$ matrix
		filled with $\ntAND$  \\ 
	\leIf(\tcp*[f]{$\operOR=\max$})
		{$\operOR=\min$}  
		{\sgn\ := $1$}    
		{\sgn\ := $-1$}   
	\ForEach{nonterminal node $w\in\BDDnodesN$}{%
		$\matrix[w][\high(w)]$ := $\sgn\cdot\attr{\BDDlab(w)}$\\
	}
	\Return{$\sgn\cdot\kshortest((\BDDnodes,E),\matrix,\BDDroot[\bddT],\oldtop,k,\operAND)\!\!\!\!$}
	\caption{$\mathtt{k\_top}$ metric values for a \SAT \T}
	\label{alg:shortest_path_BDD}
\end{algorithm}


\begin{example}
	\label{ex:k-top_values}
	Consider the \DAG-structured \SAT from \Cref{fig:bottom_up_DAG}, $\T={\AND\big(\OR(a,b),\OR(b,c)\big)}$.
	To compute its two cheapest attacks under the attribution $\attrOp = \{{a\mapsto3}, {b\mapsto1}, {c\mapsto4}\}$, let $b<a<c$ s.t.\ \bddT is as in \Cref{ex:SAT:bottom_up_BDD}.
	The $\low$ edge of the root $b$ (that encodes ``not performing $b$'') is labelled with cost $\ntAND=0$, and the $\high$ edge with cost $\attr{b}=1$; the same is done for $a$ and $c$.
    Then the shortest-weight path from the root of \bddT to its $\top$-labelled leaf is $p_1=(b,\top)$, which yields the cheapest attack $\attack_1=\{b\}$ with cost $\metrA{\attack_1}=\metrA{p_1}=\attr{b}=1$.
	Second to that we find the path $p_2=(b,a,c,\top)$, which yields the second-cheapest attack $\attack_2=\{a,c\}$ with cost $\metrA{\attack_2}=\metrA{p_2}=\ntAND\operAND\attr{a}\operAND\attr{c}=0+3+4=7$.
\end{example}

In \Cref{sec:order:ktop} we introduce another method to calculate $\operatorname{Top}_{k}(T,\attrOp)$, by expressing it as a semiring attribute domain itself. This method also extends to the dynamic case.


\section{Dynamic Attack Trees}
\label{sec:DAT}

In \ATs with \SAND gates the execution order of the \BASs becomes relevant.
This affects primarily the semantics, i.e.\ what it means to perform a successful attack, but there are also metrics sensitive to the sequentiality of events.

\subsection{Partially-ordered attacks}
\label{sec:DAT:wellformed}

As for the static case, the semantics of a dynamic attack tree (\DAT) is given by its successful attack scenarios.
However, \DATs necessitate a formal notion of order, because a sequential gate $\SAND(v_1,\ldots,v_n)$ succeeds only if every $v_i$ child is completely executed before $v_{i+1}$ starts.
This models dependencies in the order of events.
For example, in H\aa{}stad's broadcast attack, $n$ messages must first be intercepted, from which an $n$-th root (the secret key) may be computed.
This standard interpretation is \emph{ordered}, and an activated \BAS is uninterruptedly completed.
This rules out constructs that introduce circular dependencies such as $\SAND(a,b,a)$.\!%
\footnote{\,Cf.\ Kumar et al. (2015), who separate activation from execution of a \BAS and can therefore operate with $\SAND(a,b,a)$ \cite{KRS15}.}

Thus, an attack scenario that operates with \SAND gates is not just a set $\attack\subseteq\BAS$, but rather a partially-ordered set (poset) \poset: here $\prec$ is a strict partial order, where $a\prec b$ indicates that $a\in\attack$ must be carried out strictly before $b\in\attack$. 
Incomparable basic attack steps can be executed in any order, or in parallel.
Thus, the attack \poset indicates that all \BAS in \attack must be executed, and their execution order will respect ${\prec}$\,.
This succinct construct can represent combinatorially many execution orders of \BAS.
For instance \poset[\{a,b\}][\varnothing] allows three executions: the sequence $(a,b)$, and $(b,a)$, and the parallel execution $a\|b$.
Instead, \poset[\{a,b\}][\{(a,b)\}] only allows the execution sequence $(a,b)$.
%
%
Note that strict partial orders are irreflexive and transitive, so e.g.\ $\SAND(a,b,c)$ gives rise to ${\prec}=\{(a,b),(b,c),(a,c)\}$.

\subsection{Semantics for dynamic attack trees}
\label{sec:DAT:semantics}

As for SATs we need to define the notions of attacks, suites, and the structure function. These are given below and are mostly straightforward analoga of the definitions in \Cref{sec:SAT:semantics} in the realm of posets. The semantics as presented here were first defined in \cite{LS2021}.

\begin{definition}
	\label{def:attack}
	Let $\T$ be a DAT.
	\begin{enumerate}[leftmargin=1.3em]
	\item The set $\allAttacks_T$ of \emph{attacks} on a DAT \T is the set of strictly partially ordered sets  $O = \poset$, where $\attack \subseteq \BAS_{\T}$.
	\item $\allAttacks_T$ has a partial order $\leqslant$ given by $O \leqslant O'$, for ${O=\poset}$ and ${O'=\poset[A'][\prec']}$, if and only if $A \subseteq A'$ and ${\prec} \subseteq {\prec'}$.
	\item An \emph{attack suite} is a set of attacks $\suite \subseteq \allAttacks_T$. The set of all attack suites is denoted $\allSuites_T$. 
	\end{enumerate}
\end{definition}

%
\begin{definition}
	\label{def:success}
	The \emph{structure function} $\sfunT\from\ATnodes\times\allAttacks\to\BB$
	of a dynamic \AT \T for $O = \poset$ is given by
	\begin{align*}
	 f_T(v,O)=&
	  \begin{cases}
		\top  & \parbox{61pt}{if~$\type{v}=\tBAS$}~\text{and}~%
				v\in\attack,\\
		\top  & \parbox{61pt}{if~$\type{v}=\tOR$}~\text{and}~%
				\exists u\in\child{v}.\sfunT(u,O)=\top,\\
		\top  & \parbox{61pt}{if~$\type{v}=\tAND$}~\text{and}~%
				\forall u\in\child{v}.\sfunT(u,O)=\top,\\
		\top  & \parbox{61pt}{if~$\type{v}=\tSAND$}~\text{and}~%
				\forall u\in\child{v}.\sfunT(u,O)=\top\\
			  & 
			    \,\text{and}\;\forall\, 1 \leqslant i < \card{\child{v}} .\, 
				\big(a \in \attack\cap\desc_{ch(v)_i} \land\\
			  & 
				\hfill a' \in \attack\cap\desc_{ch(v)_{i+1}}\big)
				\Rightarrow a\prec a',\\
		\bot  & \text{otherwise}.
	  \end{cases}
	\end{align*}
	where $\desc_v = \BAS \cap \chOp^\text{+}(v)$
	are the \BAS descendants of~$v$; and recall that \child{v} is a sequence:
	$\child{v}=(v_1,\ldots,v_n)$.
\end{definition}

\Cref{def:success} resembles \Cref{def:SAT:sfun} and adds \SAND{s}: a gate $\SAND(v_1,\ldots,v_n)$ succeeds iff all the \BAS descendants of each $v_i$ are completed before any \BAS descendant of $v_{i+1}$.

Again we say that \emph{attack $O$ reaches $v$} if $\sfunT(v,O) = \top$, and write $\Suc{v}$ for the suite of (``successful'') attacks reaching $v$.
Its minimal elements w.r.t.\ the partial order $\leqslant$ are called \emph{minimal attacks}, and its set of minimal attacks is denoted \dsemin{v}.
We let $\Suc{\T} \doteq \Suc{R_T}$ and $\dsemin{\T} \doteq \dsemin{R_T}$. Note that the set of minimal attacks has a different notation than we used for \SATs; this is because, as we explain below, for \DATs the semantics are not given by the set of minimal attacks.

\begin{example}
	\label{ex:DAT:attack}
	Three successful attacks for  the \DAT \sampleTd of \Cref{fig:AT:example:dynamic} are: \poset[\{\ww,\cc\}][\{(\ww,\cc)\}], \poset[\{\ff,\ww\}][\emptyset], and \poset[\{\ff,\ww,\cc\}][\{(\ww,\cc)\}].
	The first two are minimal.
	Attacks \poset[\{\ff,\cc\}][\emptyset] and \poset[\{\ww,\cc\}][\{(\cc,\ww)\}] are not successful.
\end{example}

\paragraph{Satisfiability of DATs}
Contrary to the static case, a dynamic \AT may have no successful attacks, for instance $\Suc{\SAND(a,a)}=\emptyset$.
We call \T \emph{satisfiable} if $\Suc{\T} \neq \emptyset$.
Satisfiability is enforced with a notion of well-formedness in \cite{BS21}, where only well-formed \DATs are given semantics.
This enforces coherence by discarding trees that allow conflicting execution orders of \BASes.
However, \cite{BS21} notes that this is overly restrictive, since it also discards satisfiable \DATs, e.g.\ $\OR(\SAND(a,b),\SAND(b,a))$ where \poset[\{a,b\}][\{(a,b)\}] is a valid attack.
Our definition of semantics avoids this issue:

\begin{definition}
	\label{def:DAT:semantics}
	The \emph{semantics of a \DAT} \T is its set of successful attacks $\Suc{\T}$, denoted \dsem{\T}.
	%
\end{definition}

\begin{wrapfigure}[5]{r}{.23\linewidth}
	\vspace{-2ex}\hspace*{-5pt}%
	\includegraphics[width=\linewidth]{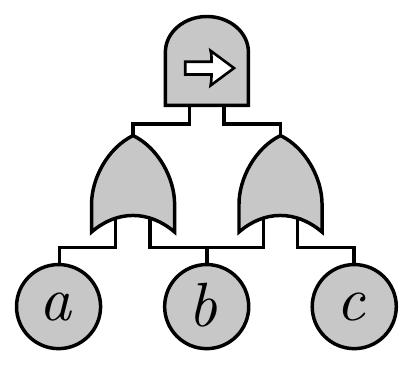}%
\end{wrapfigure}

Cf.\ to \Cref{def:SAT:semantics}, where the semantics of a \SAT \T is given by its minimal attacks $\therefore {\ssem{\T}=\dsemin{\T}}$.
Instead \Cref{def:DAT:semantics} indicates that, in general, $\dsem{\T'}\supseteq\dsemin{\T'}$ for \DAT $\T'$. We use this definition because \DATs are not coherent: it is possible that attack $O$ is successful while $O' \geqslant O$ is not. For instance in the \DAT of the picture to the right, \poset[\{a,c\}][\{(a,c)\}] is a successful attack, but \poset[\{a,b,c\}][\{(a,c)\}] is not: the single \SAND-gate has two children that share a \BAS, so not all \BAS of the first child precede those of the second child.

Hence and contrary to \SATs, the full semantics of a \DAT cannot be recovered from just its minimal attacks.
While the noncoherence of \DAT semantics is a drawback compared to \SAT semantics, it is needed in order to define semantics for all \DAG-structured \DATs, and not just the well-formed ones like in \cite{BS21}.
For well-formed \DATs our semantics are coherent; in fact this is true for all \DATs in which no two children of a \SAND gate share \BASes.

Note also that in spite of being laxer than \cite[Def.\:10]{BS21}, \Cref{def:DAT:semantics} still rules out some feasible interleavings in the execution of high-level \SAND gates.
Consider e.g.\ ${\T'=\SAND\big(a,\AND(b,c)\big)}$, where \Cref{def:success} forces $a$ to occur before any of $\{b,c\}$. However, one might argue that $\AND(b,c)$ is only executed once both $b$ and $c$ are complete, and an attack should also be considered successful as long as $a$ is executed before $\AND(b,c)$. Under this interpretation $(b,a,c)$ would be a valid execution sequence in $\T'$.
To allow this kind of sequences we need a more complex notion of strict partial order, as we discuss in \Cref{sec:conclu}.
However, this work is about the efficient computation of metrics, and as we show next these metrics are invariant for the different valid orders of execution of \BAS.
Therefore, here we use the stricter but simpler semantics that stems from \Cref{def:success}.

Finally, \Cref{lemma:dsem} characterises the minimal attacks resulting from \Cref{def:success}, analogously to how \Cref{lemma:ssem} does it for static tree-structured \ATs.
This is key to prove the correctness of linear-time algorithms that compute metrics on tree-structured \DATs.
We prove this \namecref{lemma:dsem} in \Cref{app:semantics}. We note that a generalization to \DAG-structured \DATs, similar to \Cref{lemma:ssem}, also exists, although its formulation is considerably more complicated \cite{LS2021}; we give the full statement in \Cref{app:semantics}.

\begin{lemma}
	\label{lemma:dsem}
	\def\REF#1{\textit{\ref{#1})}}
	\def\REFS#1#2{\textit{\ref{#1})--\ref{#2})}}
	Consider a tree-structured \DAT with nodes ${a\in\BAS}$, ${v_1,v_2\in\ATnodes}$. Then:
	\begin{enumerate}
	\item	$\dsemin{a} = \{ \poset[\{a\}][\emptyset] \}$;%
			\label{lemma:dsem:BAS}
	\item	$\dsemin{\OR(v_1,v_2)} = \dsemin{v_1} \cup \dsemin{v_2}$;%
			\label{lemma:dsem:OR}
	\item	
			$\dsemin{\AND(v_1,v_2)} = \mbox{$\left\{
				\poset[\attack_1{\cup}\attack_2][{\prec_1}{\cup}{\prec_2}]
				\,|\, \poset[\attack_i][\prec_i]\in\dsemin{v_i}
			\right\}$}$;%
			\label{lemma:dsem:AND}
	\item	$\dsemin{\SAND(v_1,v_2)} = \{
				\poset[\attack_1\cup\attack_2 \,]%
				      [~{\prec_1}\cup{\prec_2}\cup{\attack_1\times\attack_2}]
				\cdots$ \\
				\hspace*{\stretch{1}} $\cdots
					\mid \mbox{$\poset[\attack_i][\prec_i]\in\dsemin{v_i}$}
			\}$;%
			\label{lemma:dsem:SAND}
	\item	In cases \REFS{lemma:dsem:OR}{lemma:dsem:SAND} above the
			\dsemin{v_i} are disjoint, and in cases \REF{lemma:dsem:AND} and
			\REF{lemma:dsem:SAND} moreover the $A_i$ are disjoint.
			\label{lemma:dsem:disjoint}
	\end{enumerate}
\end{lemma}

\paragraph{Comparison with literature}
\label{par:comparison_series_parallel_graphs}
The semantics for dynamic \ATs resulting from \Cref{def:DAT:semantics} resembles the so-called \emph{series-parallel graphs} from \cite{JKM+15}.
We adhere to \cite{BS21,LS2021} and define dynamic attacks as posets for a number of reasons:
\begin{itemize}[label=\textbullet,leftmargin=1em]
\item	they are a succinct, natural lifting of the \SAT concepts,
		that facilitate the extension of earlier results such as the
		characterisation of \dsemin{\cdot} in \Cref{lemma:dsem};
\item	metrics can be formally defined on this semantics,
		decoupling specific algorithms from a notion of correctness;
\item	this allows us to define algorithms to compute metrics
		regardless of the tree- or \DAG-structure of the \DAT.
\end{itemize}
The latter is different for \cite{JKM+15}, which does not work for \DAG-structured \DATs as noted in \cite{KW18}.
In the series-parallel graph semantics, \BASes will occur multiple times when they have multiple parents. When calculating metrics, this leads to double counting as in \Cref{ex:bottom_up_DAG}.
%
%
In contrast, posets entail a formal definition of metric over \DAT semantics---given now in \Cref{sec:DAT:metrics}---which in particular yields the expected result even for \DAG-structured \DATs.

\subsection{Security metrics for dynamic attack trees}
\label{sec:DAT:metrics}

The same fundamental concepts of metric for static \ATs work for dynamic \ATs: from the attributes of every \BAS, obtain a metric for each attack in \dsemin{\T}, and from these values compute the metric for \T.
Thus, the generic notion of metric given by \Cref{def:metric} (\Cref{sec:SAT:metrics}) carries on to this \namecref{sec:DAT:metrics}.

However, attribute domains do not suffice for \DATs: metrics such as min attack time are sensitive to sequential execution of \BASs.
This calls for an additional \emph{sequential operator ${\operSAND \from \Vdom^2 \to \Vdom}$}, to compute values of sequential parts in an attack.
Therefore, metrics computation has an extra~step:
\begin{enumerate}
\setcounter{enumi}{-1}
\item	first, an attribution $\attrOp$ assigns a value to each \BAS;
\item	then, a sequential metric $\metrSOp$ uses the operator $\operSAND$ to assign a value to each sequential part of an attack;
\item	then, a parallel metric $\metrAOp$ uses $\operAND$ to assign a value to each attack, as the parallel execution of its sequential parts;
\item	finally, the metric $\metrOp$ uses $\operOR$ to assign a value to the whole attack suite, considering all its parallel attacks.
\end{enumerate}

To formalise this operational intuition consider an attack ${O = \poset}$: a \emph{maximal chain} is a sequence $(a_1,\ldots,a_n)$ in \attack s.t.\ $a_1$ is minimal under $\prec$, $a_n$ is maximal, and $a_{i+1}$ is a direct successor of $a_i$ for each $i < n$.
The set of maximal chains in $O$ is denoted $\MC_O$.
Thus, the 4-steps computation described above can be reinterpreted as follows:
\begin{enumerate}
\setcounter{enumi}{0}  
\item	$\metrSOp$ uses $\operSAND$ on each maximal chain,
		yielding one value $s_C\in\Vdom$ for each $C \in \MC_O$;
\item	$\metrAOp$ uses $\operAND$ on the values $\{s_C\}_{C \in \MC_O}$,
		yielding a metric for the attack $O$;
\item	$\metrOp$ uses $\operOR$ on the metrics of all attacks in a suite $\suite$,
		yielding the metric for $\suite$.
\end{enumerate}

We use these concepts to define the metric \metr{\dsemin{\T}} of a dynamic \AT \T, which we denote \metr{\T} in order to map \Cref{def:DAT:metric} to the generic notion of metric given in \Cref{def:metric}.

\begin{definition}
	\label{def:DAT:metric}
	\def\VH{\vphantom{\bigoperAND_{\Hasse}}}
	Let $\operOR,\operAND,\operSAND$ be three associative and commutative
	operators over a set \Vdom:
	we call $\domain=(\Vdom,\operOR,\operAND,\operSAND)$
	a \emph{dynamic attribute domain}.
	Let \T be a satisfiable dynamic \AT and $\attrOp$ an attribution on \Vdom. Let $\suite$ be an attack suite on $\T$.
 	The \emph{metric for \suite} associated to \domain and $\attrOp$ is given by:
	\begin{align} \label{eq:metrdat}
		\metr{\suite} &=
			\underbrace{\VH\bigoperOR_{O\in\suite}}_{\mathlarger\metrOp}
			\underbrace{\VH\;\bigoperAND_{C \in \MC_O}\;}_{\mathlarger\metrAOp}
			\underbrace{\VH~\bigoperSAND_{a\in\ccomp}~}_{\mathlarger\metrSOp}
			\attr{a}
	\end{align}
	where $\MC_O$ is the set of maximal chains in the poset $O$. The metric for $\T$ is defined as $\metr{\T} \doteq \metr{\dsemin{\T}}$.
\end{definition}

As for \SATs, metrics are defined in terms of the minimal attacks, rather than all successful attacks.
This causes a slight mismatch between the definition of metrics of \DATs (based on minimal attacks) and their semantics (based on succesful attacks).
We have chosen to define \DAT metrics as in \Cref{def:DAT:metric} for consistency with \SATs: in particular, when \SATs are interpreted as \DATs without \SAND-gates, metrics such as max damage have the same value as \SAT-metric and as \DAT-metric.
Note that when a \DAT metric satisfies a suitable absorption axiom, similar to \Cref{sec:SAT_trees:metrics}, it does not matter whether the metric is calculated from $\dsem{T}$ or $\dsemin{T}$.

\begin{example}
	\label{ex:DAT:metric}
	\def\wrap#1{\Bigg(\raisebox{1ex}{$\displaystyle #1$}\Bigg)}
	The minimal attacks for the dynamic \AT from \Cref{ex:running_examples} are
	$\dsemin{\sampleTd}= {\{O_1,O_2\}} = {\{
	\poset[\{\ff,\ww\}][\emptyset]\,,
	\poset[\{\ww,\cc\}][\{\ww\prec\cc\}]\}}$.
	The Hasse diagrams of these attacks---which resp.\ have one and two
	$\MC$s---are shown in \Cref{fig:Hasse:fw,fig:Hasse:wc}.
	To compute the \emph{min time} metric of \sampleTd consider the attribution
	${\attrOp}=\{{\ff\mapsto3}, {\ww\mapsto15}, {\cc\mapsto1}\}$
	and the dynamic attribute domain $\domain=(\NN,\min,\max,+)$.
	Then the time of the fastest attack for \domain and $\attrOp$ is:
	\begin{align*}
	\metr{\sampleTd}
		&= \bigoperOR_{O\in\dsemin{\T}}
		   \bigoperAND_{~\ccomp\in\MC_O~}
		   \bigoperSAND_{a\in\ccomp}~\attr{a}\\
		&= \wrap{\bigoperAND_{\ccomp\in\MC_{O_1}}
		   \bigoperSAND_{a\in\ccomp}~\raisebox{-.2ex}{$\attr{a}$}}
		   \mathbin{\text{\raisebox{.3ex}{$\operOR$}}}
		   \wrap{\bigoperAND_{\ccomp\in\MC_{O_2}}
		   \bigoperSAND_{a\in\ccomp}~\raisebox{-.2ex}{$\attr{a}$}}\\
		&= \big(\attr{\ff}\operAND\attr{\ww}\big)
		   \operOR
		   \big(\attr{\ww}\operSAND\attr{\cc}\big)\\
		&= \min(\max(3,15),15+1) ~=~ 15.
	\end{align*}
	Note that attack $O_1=\poset[\{\ff,\ww\}][\emptyset]$
	has two parallel steps: two
	maximal chains with one node each---see \Cref{fig:Hasse:fw}---so
	operator $\operAND$ has two operands with one node each:
	$\ccomp=\{\ff\}$ and $\ccomp'=\{\ww\}$.
	In contrast, $O_2=\poset[\{\ww,\cc\}][\{\ww\prec\cc\}]$ has one
	maximal chain with two nodes, so operator $\operAND$ has one operand
	but $\operSAND$ has two: \attr{\ww} and \attr{\cc}.
	Finally, the \emph{min time} of \sampleTd is the ${\operOR}=\min$ of these
	two metrics: the one for~$O_1$.
	\\ \vphantom{\raisebox{.3ex}{$\Big)$}}%
	Now consider the \emph{min skill} metric with the attributes
	${\attrOp'}=\{{\ff\mapsto42},{\ww\mapsto10},{\cc\mapsto0}\}\!$.
	This metric is oblivious of sequential order:
	the skill needed to perform a task is independent of whether
	it must wait for the completion of other tasks.
	So, to compute the min skill metric of \sampleTd we use the
	dynamic attribute domain $\domain'=(\NN,\min,\max,\max)$, where the
	operators $\operAND$ and $\operSAND$ are the same.
	This results in:
	\begin{align*}
	\metrOp'(\sampleTd)
		&= \big(\attrOp'(\ff)\operAND'\attrOp'(\ww)\big)
		   \operOR'
		   \big(\attrOp'(\ww)\operSAND\!'\attrOp'(\cc)\big)\\
		&= \min(\max(42,10),\max(10,0)) ~=~ 10.
	\end{align*}
\end{example}

Note that \emph{the order of execution} of the \BASes in the $\MC$ of an attack \emph{is irrelevant for the value of a metric}.
This is a consequence of the commutativity of the $\operSAND$ operator.

As for \SATs, in order to be able to actually compute metrics, we need additional structure on $(V,\operOR,\operAND,\operSAND)$.
More precisely, for a bottom-up algorithm to work on tree-structured \DATs (see \Cref{theo:bottom_up_DAT}), we again need distributivity:

\begin{definition}
	\label{def:semiring_dynamic_attribute_domain_OMG_howlong_isthislabel}
	A \emph{semiring dynamic attribute domain} is a dynamic attribute domain
	$\domain=(\Vdom,\operOR,\operAND,\operSAND)$ where
	operator $\operSAND$ distributes over $\operAND$ and $\operOR$,
	and also $\operAND$ distributes over~$\operOR$.
	%
	%
\end{definition}

Note that min time and min skill---used in \Cref{ex:DAT:metric} above---are both semiring dynamic attribute domains.

\paragraph{Relation to SAT metrics}
Many metrics are like min skill in \Cref{ex:DAT:metric}: insensitive to the sequentiality of events. We can calculate these metrics for a \DAT \T by changing all \SAND-gates into \AND-gates, and applying our theory on (static) attribute domains. Alternatively, if $(\Vdom,\operOR,\operAND)$ is a semiring attribute domain with idempotent $\operAND$ (such as min skill), then the metric can be calculated via \Cref{def:DAT:metric} for the dynamic attribute domain $(\Vdom,\operOR,\operAND,\operAND)$.

On the other hand, any \SAT \T is also a \DAT without \SAND-gates: we disambiguate by writing \T[d] for the \DAT interpretation.
Note then that $\attack \in \ssem{\T}$ iff $\poset[A][\varnothing] \in \dsemin{\T[d]}$.
Let then $\attrOp$ be an attribution on a semiring attribute domain $\domain = (\Vdom,\operAND,\operOR)$.
One can extend \domain into a semiring \emph{dynamic} attribute domain $(\tilde{\Vdom},\tilde{\operOR},\tilde{\operAND},\tilde{\operSAND})$ with $\Vdom \subseteq \tilde{\Vdom}$, s.t.\ ${\metr{\T} = \metr{\T[d]}}$; for details see \Cref{app:semantics}.
As a consequence, results proved for metrics on \DATs also hold for \SATs.




\begin{figure}
  \centering
  \def\WIDTH{.16\linewidth}
  \begin{subfigure}{\WIDTH}
	\centering
	\includegraphics[width=\linewidth]{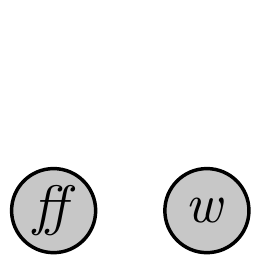}
	\caption{\Hasse[O_1][]}
	\label{fig:Hasse:wc}
  \end{subfigure}
  \qquad
  \begin{subfigure}{\WIDTH}
	\centering
	\includegraphics[height=1.02\linewidth]{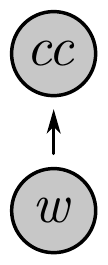}
	\caption{\Hasse[O_2][]}
	\label{fig:Hasse:fw}
  \end{subfigure}
  \qquad~
  \begin{subfigure}{\WIDTH}
	\centering
	\includegraphics[width=\linewidth]{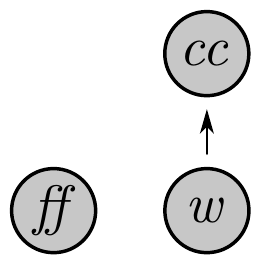}
	\caption{\Hasse[O_3][]}
	\label{fig:Hasse:fwc}
  \end{subfigure}
  \caption{Hasse diagrams of attacks of \sampleTd:\
	$O_1=\poset[\{\ff,\ww\}][\emptyset]$,\, \
	$O_2=\poset[\{\ww,\cc\}][\{(\ww,\cc)\}]$,\, \
	$O_3=\poset[\{\ff,\ww,\cc\}][\{(\ww,\cc)\}]$.}
  \label{fig:Hasse}
\end{figure}


\section{Computations for tree-structured DATs}
\label{sec:DAT_trees}

Earlier in \Cref{ex:DAT:metric}, the computation of metrics for dynamic \ATs was illustrated using \Cref{def:DAT:metric}, which is worst-case exponential in the number of nodes.
However and as for \SATs, there is a bottom-up algorithm to compute metrics for tree-structured \DATs, that is linear in the number of nodes of the attack tree.
We present a recursive version in \Cref{alg:bottom_up_DAT} (\cpageref{alg:bottom_up_DAT}), and state its correctness in \Cref{theo:bottom_up_DAT}.

The proof of \Cref{theo:bottom_up_DAT} (in \Cref{app:metrics}) relies on \Cref{thm:mod}, whose proof in turn uses the distributivity of operator $\operSAND$ over $\operOR$ and $\operAND$.
Thus the fact that $(\Vdom,\operOR,\operAND,\operSAND)$ is a \emph{semiring} dynamic attribute domain is crucial.


\begin{algorithm}
	\KwIn{Dynamic attack tree $\T=(\ATnodes,\typOp,\chOp)$,\newline
	      node $v\in\ATnodes$,\newline
	      attribution $\attrOp$,\newline
	      semiring dynamic attr.\ dom.\ ${\domain=(\Vdom,\operOR,\operAND,\operSAND)}$.\!\!}
	\KwOut{Metric value $\metr{\T}\in\Vdom$.}
	\BlankLine
	\uIf{$\type{v}=\tOR$}{%
		\Return{$\bigoperOR_{u\in\child{v}}
		         \BUDAT(\T,u,\attrOp,\domain)$}
	} \uElseIf{$\type{v}=\tAND$}{%
		\Return{$\bigoperAND_{u\in\child{v}}
		         \BUDAT(\T,u,\attrOp,\domain)$}
	} \uElseIf{$\type{v}=\tSAND$}{%
		\Return{$\text{\raisebox{.3ex}{%
		        $\bigoperSAND_{u\in\child{v}}
		         \BUDAT(\T,u,\attrOp,\domain)$}}$}
	} \Else(\tcp*[h]{$\type{v}=\tBAS$}) {%
		\Return{\attr{v}}
	}
	\caption{\BUDAT for a tree-structured \DAT \T}
	\label{alg:bottom_up_DAT}
\end{algorithm}

\begin{theorem}
	\label{theo:bottom_up_DAT}
	Let \T be a dynamic \AT with tree structure,
	$\attrOp$ an attribution on \Vdom,
	and $\domain=(\Vdom,\operOR,\operAND,\operSAND)$ a
	semiring dynamic attribute domain.
	Then $\metr{\T} = \BUDAT(\T,\ATroot,\attrOp,\domain)$.
\end{theorem}


\section{Computations for DAG-structured DATs}
\label{sec:DAT_DAGs}

\Cref{alg:bottom_up_DAT} does not work for dynamic \ATs with a \DAG-structure, for the same reasons exposed for \SATs in \Cref{sec:SAT_DAGs}.
Neither is it possible to propose algorithms based on standard \BDD theory, because the computation of metrics for \DATs needs a notion of order among their \BASs, that is not present in standard \BDD-based data types.

As discussed in \Cref{sec:intro}, some general approaches do exist to compute metrics on \DAG-structured \DATs \cite{KRS15,AGKS15}.
However, these often overshoot in terms of computation complexity.
For \SATs and from a procedural (rather than semantic) angle, \cite{KW18} proposes a more efficient, ingenious approach that computes and then corrects a metric value by traversing the \AT bottom-up repeatedly.
It may be possible to extend this algorithm to cover \SAND gates as well~\cite{WAFP19}.

Alternatively, \Cref{def:DAT:metric} of \DAT metric could be encoded into a na\"ive algorithm.
This would enumerate all posets from \dsemin{\T}, and compute the value \metr{\T} using three nested loops to traverse the corresponding Hasse diagrams.
We do not expect such approach to be computationally efficient.

Instead and as in the static case, we expect that \BDD encodings of the \DAT offer better solutions.
This requires \BDD-like structures also sensitive to variable orderings.
In that sense, the so-called sequential-\BDDs recently presented for dynamic fault trees seem promising \cite{YW20}.
A first challenge would be to extend them to attributes other than failure (viz.\ attack) probability.
Another---harder---challenge to apply this approach efficiently is the combinatorial explosion, that stems for the different possible orderings of \BAS descendants of \SAND gates.

In view of these considerations, we regard the algorithmic analysis for \DAG-structured dynamic attack trees as an important open problem for future research.
Instead, we now discuss modular analysis: a simplification strategy that can be used in any algorithm that calculates metrics.

\section{Modular analysis}
\label{sec:mod}

\newcommand{\metrAlgo}[1][A]{\ensuremath{\mathscr{#1}}\xspace}


In this \namecref{sec:mod} we show that the calculation of metrics can be split up according to the modules of an \AT.
The resulting \emph{modular analysis} is a well-established method for quantitative analysis of fault trees and \ATs \cite{DR96,reay2002fault,yevkin2011improved,LS2021}. We exploit modular analysis in the general semiring (dynamic) attribute domain setting, leading to improved performance in calculating these metrics.

For $v \in \ATnodes$ we let \T[v] be the sub\DAG of \T consisting of all descendants of $v$, with $v$ as the root.
Intuitively, a module is an inner node $v$ such that all paths from $\T \setminus \T[v]$ to \T[v] pass through $v$.
This is formalised in the following definition.

\begin{definition}
Let $v \in \ATnodes \setminus \BAS$. We call node $v$ a \emph{module} if $\T[v] \cap \T[w] \in \{\T[v],\T[w],\emptyset\}$ for all $w \in \ATnodes$.
\end{definition}

Note that the root of \T is always a module.
The modules of an \AT \T can be found in linear time \cite{DR96}.
These modules aid calculation in the following manner: Let $v$ be a module, then $v$ is the only node within \T[v] with parents outside of \T[v].
This means that we can create a tree $\T^v$ by replacing \T[v] within \T by a new single \BAS $\tilde{v}$.
Then the parents of $\tilde{v}$ in $\T^v$ are the parents of $v$ in \T, see \Cref{fig:modular}.
This allows to calculate a metric \metr{\T} by first calculating the metric on $\T[v]$, and then on $\T^v$.
This is formalized in \Cref{thm:mod}.

\medskip\noindent%
\begin{minipage}{.55\linewidth}
	\centering
	\includegraphics[width=.9\linewidth]{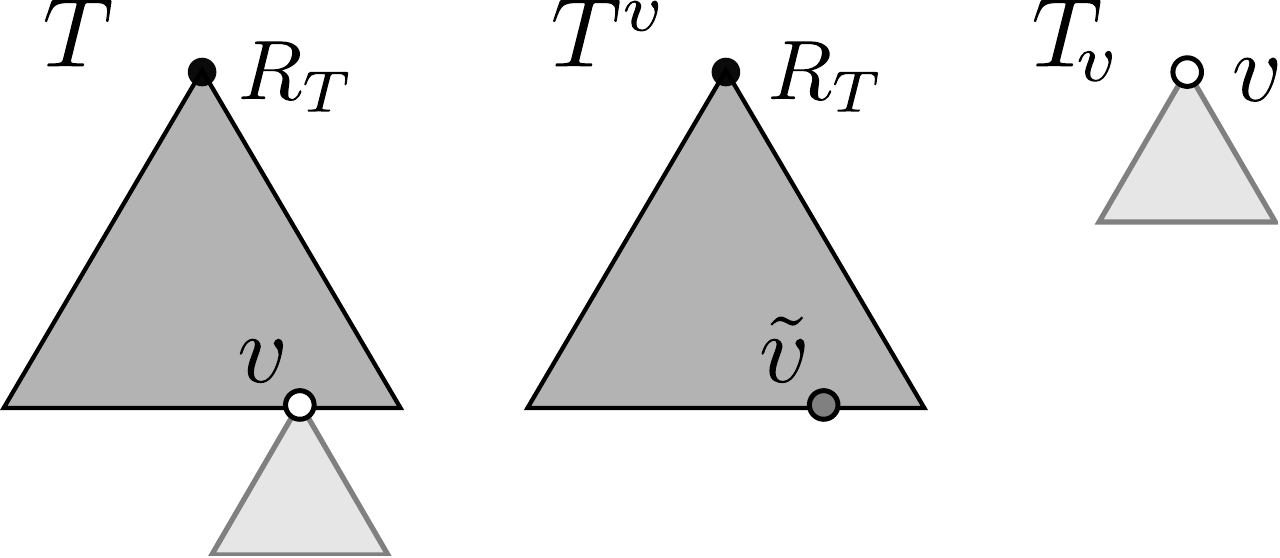}
\end{minipage}
\hfill
\begin{minipage}{.4\linewidth}
	\vspace{1ex}
	\captionof{figure}{To compute \metr{\T} first compute the metric
		\metr{\T[v]} for module $v$, then compute $\metrOp^v(\T^v)$
		as in \Cref{thm:mod}.}
	\label{fig:modular}
\end{minipage}
\medskip

\begin{theorem}
	\label{thm:mod}
	\def\ATedges{\mathobject{E}}
	Let $v$ be a module in an \AT \T, and $\attrOp$ an attribution into
	a (dynamic) attribute domain \Vdom.
	Let $\T^v = (\ATnodes^v,\ATedges^v)$ be the \AT obtained by replacing \T[v]
	by a new single \BAS $\tilde{v}$. Let $\attrOp^v\from \ATnodes^v \to \Vdom$
	be an attribution for $\T^v$ given by
	\begin{equation*}
		\attrOp^v(v') =
		\begin{cases}
			\attrOp(v')  & \text{if}~ v' \neq v,\\
			\metr{\T[v]} & \text{if}~ v' = \tilde{v}.
		\end{cases}
	\end{equation*}
	Then $\metr{\T} = \metrOp^v(\T^v)$.
\end{theorem}

\Cref{thm:mod} allows us to split up metric computation by splitting up the attack tree \T.
This can be used to par\-al\-lel\-ise---at least partially---any algorithm \metrAlgo that calculates a metric.
More generally, when the time complexity of \metrAlgo is high, e.g.\ exponential in the number of nodes of \T, then splitting the calculation in the modules of \T will result in a lower computation time.
This result is non-trivial, as it requires to express the minimal attacks of \T (and their maximal chains) in terms of the minimal attacks of $\T^v$ and \T[v]: this is encoded in \Cref{thm:mod} (proved in \Cref{app:modular}).

This result can be implemented by identifying all modules $v$ of \T as in \cite{DR96}, and then calculating $\metrOp(\T[v])$ for them using any algorithm \metrAlgo.
By doing this bottom-up, we can use the result for lower modules in the calculation of higher ones.
The resulting method is presented as \Cref{alg:satmod}; its correctness is stated in the following \namecref{coro:satmod_correct} of \Cref{thm:mod}:

\begin{corollary}
	\label{coro:satmod_correct}
	\Cref{alg:satmod} correctly calculates $\metrOp(\T)$.
\end{corollary}


\begin{algorithm}
	\def\ModUlez{\ensuremath{U}}
	\KwIn{Static or dynamic attack tree \T,\newline
	      attribution $\attrOp$,\newline
	      algorithm \metrAlgo to calculate metric $\metrOp$}
	\KwOut{$\metrOp(\T)$.}
	\BlankLine
	\DontPrintSemicolon
	$\ModUlez := \mathtt{Modules}(\T)$\;
	\While{$\ModUlez \neq \varnothing$}{
		Extract minimal $v$ from $\ModUlez$\;
		$\metrOp(\T[v]) := \mathscr{A}(\T[v],\attrOp)$
			\tcp*[h]{Compute metric for module $v$}\;
		$(\T,\attrOp) := (\T^v,\attrOp^v)$
			\tcp*[h]{Update $\attrOp$ and remove module $v$}\;
	}
	\Return{$\attr{\ATroot}$} \tcp*[h]{\ATroot is a \BAS now}
	\caption{Modular analysis.
		Notation $(\T^v,\attrOp^v)$ is from \Cref{thm:mod};
		$\mathtt{Modules}$ is an algorithm that lists all modules,
		e.g.\ the one from \cite{DR96}.}
	\label{alg:satmod}
\end{algorithm}

\begin{example}
	\label{ex:modular}
	\opencutright	
	\renewcommand\windowpagestuff{%
		\centering\includegraphics[width=.8\linewidth]{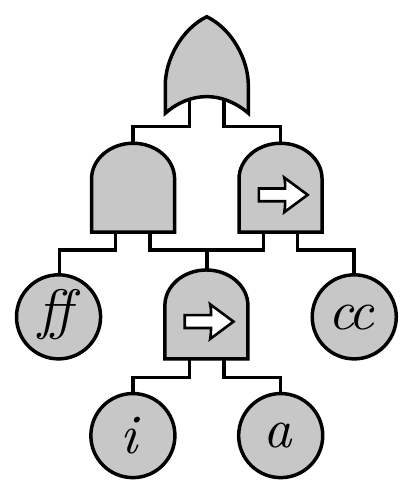}
	}
	\begin{cutout}{0}{.7\linewidth}{0pt}{7}
	Recall the \DAT from \Cref{fig:AT:example:dynamic}, \sampleTd, and refine the \BAS \ww (``walk next to victim'') to consist of two steps: ``identify possible target'' ($i$) and ``approach victim'' ($a$).
	Then instead of $\ww\in\BAS$, the sub-tree shared by gates ``skill'' and ``luck'' becomes $v=\SAND(i,a)$, yielding the \DAT $\T'=\OR\big(\AND(\ff,v),\SAND(v,\cc)\big)$.
	Consider the \emph{min time} metric given by the semiring dynamic attribute domain $(\mathbb{N},\min,\max,+)$, and attribution $\attrOp = \{{\ff \mapsto 3},\allowbreak {i \mapsto 10},\allowbreak {a \mapsto 5},\allowbreak {\cc \mapsto 1}\}$.
	Then $\,{\metr{\T[v]'} = \attr{i}+\attr{a} = 15}\,$, and since
	\end{cutout}
	\noindent%
	the truncated \DAT $\T'^v$ is isomorphic to \T[d], by \Cref{thm:mod} we get $\metr{\T'}$ = $\metr{\T[d]} = 15$ (see \Cref{ex:DAT:metric}).
	In contrast, applying \Cref{eq:metrdat} directly to $\T'$ would have computed \metr{v} twice.
	The gain of this approach is proportional to the amount of repetitions of each module, times their size.
\end{example}

Note that if \T has tree structure, each node is a module. In this case \Cref{alg:satmod} becomes a bottom-up method, and so \Cref{thm:mod} is an important result for \Cref{theo:bottom_up_SAT,theo:bottom_up_DAT}.

\Cref{alg:satmod} does not solve the problem of calculating metrics for \DAG-structured \DATs, since it assumes the existence of an algorithm \metrAlgo that computes $\metrOp$. However, it can speed up any found algorithm, including the enumeration of all minimal attacks to calculate the metric from these.

\section{Multiple metrics simultaneously}
\label{sec:order}

An important class of metrics are those which assign to an \AT not a single metric value, but a set of metric values. Typical examples are the following (stated below for \SATs but applicable to \DATs as well):

\begin{enumerate}
    \item \textbf{Uncertainty sets}: Suppose that the $\attr{a}$ are not known exactly, but instead we only have bounds $\attr{a} \in [L_a,U_a]$ for all $a \in \BAS$. For instance, we might only have a confidence interval for $\attr{a}$. In this case, we are interested in finding $\widecheck{L},\widecheck{U}$ such that $\metr{\T} \in [\widecheck{L},\widecheck{U}]$.
    \item \textbf{$\boldsymbol{k}$-top metrics}: The $k$-lowest (or -highest) values of a given metric, see \Cref{sec:ktop}.
    \item \textbf{Pareto front}: 
    Attributes can be opposed, e.g.\ attack $\attack_1$ may be less expensive than $\attack_2$, but take more time. To understand the tradeoffs between different metrics one studies its \emph{Pareto front}: the set of metric values of attacks that are not dominated in all metrics by another attack.
\end{enumerate}

Such metrics relate to partial orders on the domain:
We now give a framework to express the examples above as semiring attribute domain metrics, via their ordering.
\emph{%
This allows us to run once any algorithm that calculates semiring metrics, e.g.\
\Cref{alg:bottom_up_SAT,alg:bottom_up_BDD,alg:bottom_up_DAT}, to compute all elements of these (multi) sets simulaneously%
} --- proofs are in \Cref{app:order}.

\begin{definition}
\begin{enumerate}
    \item A \emph{partially ordered semigroup (POSG)} is a tuple $(X,\preceq,\operAND)$ such that:
\begin{enumerate}
    \item $(X,\preceq)$ is a poset;
    \item $\operAND$ is a commutative associative operation on $X$;
    \item If $x,y \in X$ are such that $x \preceq y$, then $x \operAND z \preceq y \operAND z$ for all $z \in X$.
\end{enumerate}
If $\preceq$ is a linear order, we call $(X,\preceq,\operAND)$ a \emph{linearly ordered semigroup (LOSG)}.
\item A \emph{dynamic partially ordered semigroup (DPOSG)} is a tuple $(X,\preceq,\operAND,\operSAND)$ such that $(X,\preceq,\operAND)$ and $(X,\preceq,\operSAND)$ are POSGs and $\operSAND$ distributes over $\operAND$. If $\preceq$ is a linear order, we call $(X,\preceq,\operAND,\operSAND)$ a \emph{dynamic linearly ordered semigroup (DLOSG)}.
\end{enumerate}
\end{definition}

There are different ways to create semiring attribute domains out of POSGs, and semiring dynamic attribute domains out of DPOSG. In the case that the partial order is linear this can be done directly:

\begin{lemma} \label{lemma:LOSG}
\begin{enumerate}
\item If $(X,\preceq,\operAND)$ is an LOSG, then $(X,\min,\operAND)$ is a semiring attribute domain.
\item If ($X,\preceq,\operAND,\operSAND)$ is a DLOSG, then $(X,\min,\operAND,\operSAND)$ is a semiring dynamic attribute domain.
\end{enumerate}
\end{lemma}

Semiring attribute domains from LOSGs are ubiquitous: all examples in \Cref{tab:SAT:metric} come from the construction in \Cref{lemma:LOSG} (reverting the natural order to change $\min$ into $\max$ if necessary). In the following we explain how examples 1)--3) mentioned above can be calculated using semiring (dynamic) attribute domains derived from (D)POSGs. We only give the dynamic statements below, but the static cases are completely analogous.


\begin{table*}
	\centering
	
	\vspace{1ex}
	\caption{%
	  Algorithms for metrics on different \AT classes%
	  ~~(replica of \Cref{tab:all_algos_intro})
	}
	\label{tab:all_algos_conclu}
	\vspace{2ex}
	\parbox{.73\linewidth}{
		\def\bfacro#1{\acronym{\bfseries{#1}}}%
		\bfacro{bu}:      bottom-up on the \AT structure.
		\bfacro{aph}:     acyclic phase-type (time distribution).
		\bfacro{bdd}:     binary decision diagram.
		\bfacro{mtbdd}:   multi-terminal \acronym{bdd}.
		$\pmb{\pazocal{C}}$-\bfacro{bu}:
		                  repeated \acronym{bu}, identifying clones.
		\bfacro{dpll}:    \acronym{dppl sat}-solving in the \AT formula.
		\bfacro{pta}:     priced time automata (semantics).
		\bfacro{milp}:    mixed-integer linear programming.
		\bfacro{i/o-imc}: input/output interactive Markov chains (semantics). ${}^1$ \Cref{alg:satmod} reduces runtime for any found method; ${}^2$ \Cref{lemma:topk} reduces $k$-top calculation to metric calculation.
	}
	\vspace{-2ex}
\end{table*}

\subsection{Uncertainty sets}

Let $(X,\preceq,\operAND,\operSAND)$ be a DLOSG. Suppose that $\alpha$ is not known exactly; instead for every $a \in \BAS$ we have $L_a,U_a \in X$ for which we know $L_a \preceq \alpha(a) \preceq U_a$. In this case, we are interested in
\begin{align*}
\widecheck{L}_T &\doteq \inf\{\metr{\T} \mid \forall a. L_a \preceq \alpha(a) \preceq U_a \}\\
\widecheck{U}_T &\doteq \sup\{\metr{\T} \mid \forall a. L_a \preceq \alpha(a) \preceq U_a \}.
\end{align*}

We find this as follows: let $\domain = (X,\min,\operAND,\operSAND)$ be the semiring dynamic attribute domain from \Cref{lemma:LOSG}. Consider the semiring dynamic attribute domain $\domain^2$, which has underlying set $X^2$ and on which every operator acts componentwise. Define an attribution $\beta$ with values in $X^2$ by $\beta(a) = (L_a,U_a)$; then $\widecheck{\beta}(\T) = (\widecheck{L}_T,\widecheck{U}_T)$. The key observation to prove this is that for every \T, the map $\alpha \mapsto \metr{\T}$ is monotonous in each $\alpha(a)$.

\subsection{\texorpdfstring{$\boldsymbol{k}$}{k}-top metrics}
\label{sec:order:ktop}

Let $(X,\preceq,\operAND,\operSAND)$ be a DLOSG: we want to find the $k$-top metric $\operatorname{Top}_{k}(T,\alpha) \in \mathscr{M}(X)$ from \Cref{sec:ktop}.
This can be done via a semiring metric as follows. Let $\mathscr{M}^k(X)$ be the set of multisets in $X$ of cardinality at most $k$. Define three operations $\operOR^k$, $\operAND^k$ and $\operSAND^k$ on $\mathscr{M}^k(X)$ by
\begin{align*}
M_1 \operOR^k M_2 &= \operatorname{min}^k(M_1 \uplus M_2)\\
M_1 \operAND^k M_2 &= \operatorname{min}^k\ldb x_1 \operAND x_2 \mid x_1 \in M_1,x_2 \in M_2 \rdb\\
M_1 \operSAND^k M_2 &= \operatorname{min}^k\ldb x_1 \operSAND x_2 \mid x_1 \in M_1,x_2 \in M_2 \rdb.
\end{align*}
Here $\uplus$ denotes multiset union. Furthermore, define a map $\beta\colon \BAS \rightarrow \mathscr{M}(X)$ by $\beta(a) = \ldb \alpha(a) \rdb$. Then $\operatorname{Top}_{k}(\T,\alpha)$ can be found as follows:

\begin{lemma}  \label{lemma:topk}
The tuple $(\mathscr{M}^k(X),\operOR^k,\operAND^k,\operSAND^k)$ is a semiring dynamic attribute domain, and $\widecheck{\beta}(\T) = \operatorname{Top}_k(\T,\alpha)$.
\end{lemma}

Compared to \Cref{alg:shortest_path_BDD}, this method is more general in the sense that it also works for \DATs, but it comes at a complexity cost for \SATs: once the \BDD $(W,Low,High,Lab)$ corresponding to \T has been constructed, \Cref{alg:shortest_path_BDD} has complexity $O(|W|+k)$, while applying \Cref{alg:bottom_up_BDD} to \Cref{lemma:topk} has complexity $O(k^2|W|)$.

\subsection{The antichain semiring}

If $(X,\preceq,\operAND,\operSAND)$ is a DLOSG, then we can interpret it as a semiring attribute domain $(X,\min,\operAND,\operSAND)$. We cannot do the same when it is a DPOSG, because $\min(x,y)$ may not exist. To create a semiring attribute domain out of a DPOSG we need a more elaborate construction. Specifically, let $\AC_X$ be the set of \emph{antichains} in $(X,\preceq)$, i.e.\ sets of pairwise incomparable elements:
\[
\AC_X = \{S \in 2^X \mid \forall x,x' \in S\colon x' \not \prec x\}.
\]
Furthermore, define a map $m\colon 2^X \rightarrow \AC_X$ that sends a set to the antichain of its minimal elements:
\[
m(S) = \{x \in S \mid \forall x' \in S \colon x' \not \prec x\}.
\]
We also define three operations $\operOR_{\AC},\operAND_{\AC}, \operSAND_{\AC}$ on $\AC_X$ by
\begin{align*}
S_1 \operOR_{\AC} S_2 &= m(S_1 \cup S_2)\\
S_1 \operAND_{\AC} S_2 &= m(\{x_1 \operAND x_2 \mid x_1 \in S_1, x_2 \in S_2\})\\
S_1 \operSAND_{\AC} S_2 &= m(\{x_1 \operSAND x_2 \mid x_1 \in S_1, x_2 \in S_2\}).
\end{align*}
The following \namecref{lemma:antichain} is an extension of \cite{FW19}, where it is shown for the static case under mild assumptions on $X$.

\begin{lemma} \label{lemma:antichain}
The tuple $(\AC_X,\operOR_{\AC},\operAND_{\AC},\operSAND_{\AC})$ is a semiring dynamic attribute domain.
\end{lemma}

\smallskip\noindent%
This has a number of applications:

\smallskip
\begin{enumerate}[leftmargin=1.5em]
    \item If $(X,\preceq,\operAND,\operSAND)$ is a DLOSG, then every antichain is a singleton, and the map $X \rightarrow \AC_X$ given by $x \mapsto \{x\}$ is an isomorphism of semiring dynamic attribute domains.
	\item For a \SAT $\T$, consider the static POSG $(2^{\BAS},\subseteq, \cup)$ and the attribute $\beta\from \BAS \to \AC_{2^{\BAS}}$ given by $\beta(a) = \{\{a\}\}$; then $\widecheck{\beta}(\T) = \dsem{\T}$. The elements of $\AC_{2^{\BAS}}$ are suites, and for any semiring attribute domain $(\Vdom,\operOR,\operAND)$, each attribution $\alpha\colon \BAS \rightarrow \Vdom$ induces a morphism of semiring attribute domains ${\widecheck{\alpha}\from \AC_{2^\BAS} \to \Vdom}$.
	\item Let $(X_1,\preceq_1,\operAND_1,\operSAND_1),\ldots,(X_n,\preceq_n,\operAND_n,\operSAND_n)$ be a collection of LOSGs. Let $X = \prod_i X_i$, on which we have a partial order $\preceq$ and binary operations $\operAND,\operSAND$ defined componentwise. Then $(X,\preceq,\operAND,\operSAND)$ is a DPOSG and so $(\AC_X,\operOR_{\AC},\operAND_{\AC},\operSAND_{\AC})$ is a semiring dynamic attribute domain. Let \T be a \SAT, and for each $i$ let $\alpha_i\colon \BAS \rightarrow X_i$ be an attribution; let $\alpha\from \BAS \to X$ be the product map. Define $\metrA{\attack} = (\widehat{\alpha}_1(\attack),\ldots,\widehat{\alpha}_n(\attack))$ for $\attack \in \ssem{\T}$. Then the Pareto front of \T w.r.t.\ the $\alpha_i$ is the subset of $X$ given by:
    \[
    \operatorname{PF}(\T,\alpha) = \{ \metrA{\attack} \mid \attack \in \dsemin{\T}, \forall \attack' \in \dsemin{\T}\colon \metrA{\attack'} \not\prec \metrA{\attack}\}.
    \]
    Finally, consider the map $\beta\from \BAS \to \AC_X$ given by $\beta(a) = \{\alpha(a)\}$. Then $\widecheck{\beta}(\T) = \operatorname{PF}(\T,\alpha)$.
\end{enumerate}


\section{Related work}
\label{sec:related_work}

Surveys on attack trees are \cite{KPS14,WAFP19}:
the latter covers \AT analysis via formal methods, from which we are close to quantitative model checking---cf.\ simulation studies such as \cite{DMCR06,WAN18}.
Concrete case studies have been reported in \cite{FGK+16}.

Terminology in the \AT literature is not uniform.
In particular, some works study \DAG-like structures but preserve the term ``attack tree'' \cite{MO06,BET13,BS21}.
Others restrict the syntactic structures to be actual trees, replicating parts of the tree---e.g.\ via so-called cloned nodes and repeated labels---to model the use of the same resource in several parts of an attack \cite{GIM15,BK18,KW18}.
We follow the former convention, which is akin to the treatment of common cause failures in fault tree analysis \cite[Sec.~8]{VSD+02}.
Thus, we write ``attack tree'' to refer to both tree- and \DAG-like structures.

Similarly, the term \emph{dynamic attack tree} has recently been used to refer to a set of \ATs that share the main attacker's goal \cite{AD21,AD22}.
These resemble the \emph{attack-tree series} from \cite{GM19}, where ``dynamic attack tree analysis'' refers to the study of attack-tree series.
Instead, in this work we follow \cite{BS21,BKS21} and call an \AT \emph{dynamic} when its structure includes a sequential-\AND gate---so its semantics must distinguish among different execution orders of the basic attack steps.
This is akin to the notions used in fault tree analysis, where dynamic gates like priority-\AND in dynamic fault trees have similar semantics to sequential-\AND in \ATs \cite{VSD+02,MBD20}.


Regarding \AT metrics, \Cref{tab:all_algos_conclu} condenses literature references on quantitative analyses of \ATs, classified by the structure and (dynamic) gates of the \ATs where they operate.
For each metric and \AT class, the \namecref{tab:all_algos_conclu} cites the earliest relevant contributions that include some computation procedure.

Works \cite{BLP+06,JW08} are among the first to model and compute the cost and probability of attacks: their algorithms have \EXPTIME complexity regardless of the \AT structure.
In \cite{KRS15,KSR+18} an attack is moreover characterised by the time it takes.
This allows for richer Pareto analyses but introduces one clock per \BAS in the Priced Time Automata semantics: algorithms have thus \EXPTIME \& \PSPACE complexity \cite{AD94,BLR05}.
The current work improves these bounds via specialised procedures tailored for the specific \AT class, e.g.\ \Cref{alg:bottom_up_SAT,alg:bottom_up_DAT} resp.\ for tree-structured \SATs and \DATs have \LINTIME complexity.

Indeed, all algorithms specialised on tree-structured \ATs implement a bottom-up traversal on its syntactic structure: we denote these \acronym{bu} in \Cref{tab:all_algos_conclu}.
Pareto analyses are polynomial, where the exponent is the number of parameters being optimised.
Most works are on static \ATs, with the relevant exception of \cite{AHPS14,JKM+15,GRK+16} which include sequential-\AND gates.

For \DAG-structured static \ATs the algorithmic spectrum is broader, owing to the NP-hardness of the problem (see \Cref{sec:SAT_DAGs:complexity}).
Such algorithms range from classical \BDD encodings for probabilities, and extensions to multi-terminal \BDDs, to logic-based semantics that exploit \acronym{dpll}, including an encoding of \SATs as generalised stochastic Petri nets.
Prominent contributions are \cite{BK18} and \cite[Alg.~1]{KW18}: after computing so-called optional and necessary clones, computations are exponential on the number of shared \BAS (only).
As discussed in \Cref{sec:SAT_DAGs:complexity}, in this case the exponential complexity---on the number of nodes of the complete \AT---lies in the input of the algorithm, i.e.\ clone computation.
This approach is used in \cite{FW19} to calculate Pareto fronts; in \Cref{lemma:antichain} we use a similar strategy but instead apply \Cref{alg:bottom_up_BDD}, whose exponential explosion lies in computing the \BDD that encodes the \AT.

The computation of security metrics for dynamic attack trees is more recent than for \SATs: here we find open problems in the literature, indicated in two cells of \Cref{tab:all_algos_conclu}.
These open problems are not easy to overcome, although efficient solutions have been presented for specific cases, such as the series-parallel graph semantics of \cite{JKM+15} which works for tree-structured \DATs.
However, as we discuss in \Cref{par:comparison_series_parallel_graphs} (\cpageref{par:comparison_series_parallel_graphs}), this does not extend to \DAG-structured dynamic attack trees.
Another example is \cite{AGKS15}, which encodes a \DAT as a (variant of a) Markov chain to compute attack probability.
Min time is phrased in \cite{LS2021} as a mixed-integer linear programming problem.
For other metrics, \cite{KRS15} encodes the \AT as a network of \acronym{pta} and solves the resulting cost-optimal reachability problem.
As earlier stated, these very powerful and general approaches are in detriment of computational efficiency.
Alternatively and as shown in \Cref{sec:DAT_trees,sec:mod}, efficient (linear) bottom-up algorithms can correctly compute metrics in tree-structured \DATs.
This is implemented for instance in ADTool 2.0, which can also create a ranking of attacks---e.g.\ to find the $k$-top values---under the expected conditions for the operators $\operOR, \operAND, \operSAND$ \cite{GRK+16}.

Regarding \DAG-strutured \DATs, where the open problems of \Cref{tab:all_algos_conclu} lie, recent related results encode dynamic fault trees as so-called sequential-\BDDs, to compute the probability of system failure \cite{YW20}.
However, such safety-oriented works are hard to map to security analysis such as \AT metrics because:
\begin{enumerate*}
\item	they can compute probability---and possibly parallel time---only;
\item	the dynamic gates are not the same than those in dynamic \ATs;
\item	the standard logical gates are interpreted differently.
\end{enumerate*}
Still, it might be feasible to adapt \cite{YW20} to compute \AT metrics, e.g.\ to compare it against the algorithms here presented.
Probably the main detriment is that sequential-\BDDs expand sequence dependencies of every pair of events, adding a combinatorial blow-up on top of the already exponential explosion incurred by \BDD representations of \DAGs.
This leads us to believe that even the \EXPTIME complexity of our \Cref{alg:bottom_up_BDD} can be more efficient.

\section{Conclusions}
\label{sec:conclu}


This paper presents algorithms to compute quantitative security metrics on attack trees.
This is done in two steps: first, we revise and consolidate semantics in line with the literature, and we define metrics on these semantics, providing formal grounds on which to demonstrate the correctness of any devised computation method.
A key contribution here is the adaptation of non-restrictive poset semantics for dynamic attack trees (\Cref{sec:DAT:semantics}), which allows for a formal definition of general metrics on a wide range of \DATs.

Second, we introduce efficient and unifying algorithms that can compute many popular metrics, including sets of metrics (i.e.\ several metrics simultaneously as in $k$-top and Pareto analyses).
Here, the \BDD-based approach for general metrics of \Cref{alg:bottom_up_BDD} is a prominent result, together with \Cref{lemma:topk,lemma:antichain} that show how to use single-value algorithms for the computation of set metrics.

We noted --- in \Cref{sec:DAT:semantics} --- that our \DAT semantics rules out some interleavings in the execution of \SAND gates, e.g.\ $(b,a,c)$ for ${\SAND\big(a,\AND(b,c)\big)}$, even when these sequences would arguably result in a succesful attack.
To allow such sequences it is necessary to use formulae --- rather than individual \BASs --- as elements of the partial order.
For the \DAT above, this would yield the relation $a \prec (b\land c)$, which allows $(b,a,c)$ because the formulae in that sequence are satisfied in the order ``first $a$, then $b\land c$.'' Such ordering graphs are a promising research direction.

Further lines for future work also include:
developing efficient algorithms to compute metrics
on \DAG-structured dynamic \ATs;
extending our \AT syntax to include \emph{sequential-\OR gates}
\cite{KRS15,HKKS16};
and extending our general metrics to Attack--Defense Trees \cite{KMRS11,BKMS12}.
Other important future work is to implement the methods and algorithms from this paper in real-life case studies. Interesting future work in the opposite direction would be to frame attack tree metrics in a wider, category-theoretical framework. Operad algebras may form a useful tool for research in this direction, as attribute domains can be regarded as algebras of the operad of (dynamic) attack trees.

\bibliographystyle{IEEEtran}
\bibliography{main.bib}


\begingroup
\def\PIC#1{\includegraphics[width=1in,height=1.25in,clip,keepaspectratio]{#1}\vfill\phantom{}}

\begin{IEEEbiography}  
	[\PIC{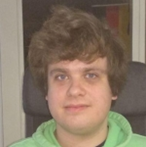}]{Milan Lopuha\"a-Zwakenberg}
	is a postdoc at University of Twente (NL), studying safety and security metrics and their interplay. Before, he was a postdoc at Eindhoven University of Technology (NL) on information-theoretic privacy metrics, and he received his PhD from Radboud University (NL) on arithmetic geometry.
\end{IEEEbiography}

\begin{IEEEbiography}  
	[\PIC{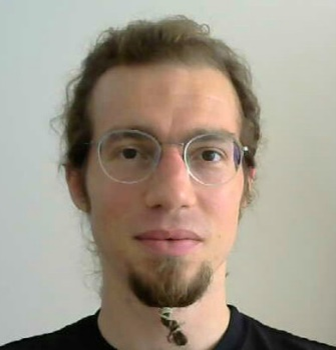}]{Carlos E.\ Budde}
	received his PhD in Computer Science in 2017 from the Universidad Nacional de C\'ordoba (AR), specialising in rare event simulation for formal methods. In 2017--2021 he was a postdoc at the Universiteit Twente (NL), also in collaboration with Dutch Railways. Since 2021 Carlos is assistant professor at the Universit\`a di Trento (IT), studying cybersecurity resilience via model simulation.
\end{IEEEbiography}

\begin{IEEEbiography}  
	[\PIC{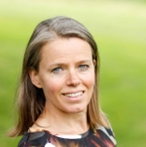}]{Mari\"elle Stoelinga}
	is a professor of risk management, both at the Radboud University and the University of Twente, in the Netherlands. Stoelinga is the project coordinator on PrimaVera, a large collaborative project on Predictive Maintenance in the Dutch National Science Agenda NWA. She also received a prestigious ERC consolidator grant. Stoelinga holds an MSc and a PhD degree from Radboud University, and has spent several years as a postdoc at the University of California at Santa Cruz, USA.
\end{IEEEbiography}

\endgroup

\vfill\newpage\clearpage\appendices


\begingroup

\allowdisplaybreaks

\def\REF#1{\textit{\ref{#1})}}
\def\REFS#1#2{\textit{\ref{#1})--\ref{#2})}}
\def\hop{\\[.5ex]}
\def\by#1{&&\hspace*{-\linewidth}\text{\smaller\color{black!70}~by #1}}


\section{Proofs of Lemmas \ref{lemma:ssem} and \ref{lemma:dsem}} \label[appendix]{app:semantics}

In order to prove the results from the paper we first prove two auxiliary lemmas. \Cref{lemma:satisdat} shows that SAT semantics and metrics can be interpreted as a special case of DAT semantics and metrics. By using this lemma we do not have to prove certain statements for SATs and DATs separately, and instead only prove the DAT case. \Cref{lemma:dsem:app} is an extension of \Cref{lemma:dsem} to DAG-type DATs.

\begin{lemma} \label{lemma:satisdat}
Let $T$ be a \SAT. Let $T_{\textrm{d}}$ be the \AT $T$ interpreted as a \DAT. Then:
\begin{enumerate}
\item $\dsemin{T_d} = \{\poset[A][\varnothing] \mid A \in \dsem{T}\}$.
\item Let $\alpha$ be an attribution of $T$ into a semiring attribute domain $D = (V,\operOR,\operAND)$. Define a semiring dynamic attribute domain $\tilde{D} = (\tilde{V},\tilde{\operOR},\tilde{\operAND},\tilde{\operSAND})$ by
\begin{align*}
\tilde{V} &= V \sqcup \{\omega\}\\
v \mathbin{\tilde{\operOR}} v' &= \begin{cases}
v \mathbin{\tilde{\operOR}} v', & \textrm{ if $v,v' \in V$}\\
v, & \textrm{ if $v' = \omega$}\\
v', & \textrm{ if $v = \omega$}
\end{cases}\\
v \mathbin{\tilde{\operAND}} v' &= \begin{cases}
v \operAND v', & \textrm{ if $v,v' \in V$}\\
\omega, & \textrm{ otherwise}\\
\end{cases}\\
v \mathbin{\tilde{\operSAND}} v' &= \omega.
\end{align*}
Let $\beta\colon N \rightarrow \tilde{V}$ by the dynamic attribution given by $\beta(a) = \alpha(a)$. Then $\metr{T} = \widecheck{\beta}(T_d)$.
\end{enumerate}
\end{lemma}

\begin{proof}
\begin{enumerate}
    \item Since $T_d$ has no SAND-gates, \Cref{def:success} does not put any restrictions on the strict partial order $\prec$ of an attack $\poset$. Hence the success of an attack does not depend on $\prec$, and each minimal attack will be of the form $\poset[A][\varnothing]$. Since this is successful iff $A$ is successful on $T$, we find that $\dsemin{T_d}$ consists of those $\poset[A][\varnothing]$ for which $A \in \dsem{T}$.
    \item By the previous point all minimal attacks of $T_d$ have a trivial poset structure; therefore the maximal chains in this poset are exactly all singletons. It follows that
\begin{align*}
\widecheck{\beta}(T_d) &= \widetilde{\bigoperOR_{O \in \dsemin{T_d}}} \widetilde{\bigoperAND_{C \in \operatorname{MC}_{O}}} \widetilde{\bigoperSAND_{a \in C}} \beta(a) \\
&= \widetilde{\bigoperOR_{A \in \dsem{T}}} \widetilde{\bigoperAND_{a \in A}} \beta(a) \\
&= \bigoperOR_{A \in \dsem{T}} \bigoperAND_{a \in A} \alpha(a) \\
&=\metr{T}. \qedhere
\end{align*}
\end{enumerate}    
\end{proof}

To formulate the generalised version of \Cref{lem:dat1} we need the following notation. For a relation $R$ on a set $X$ denote by $\operatorname{tr}(R)$ its transitive closure. For two attacks $O_1 = \poset[A_1][\prec_1]$ and $O_2 = \poset[A_2][\prec_2]$ define an attack $P(O_1,O_2) = \poset$ given by
\begin{align}
A &= A_1 \cup A_2\\
{\prec} &= \operatorname{tr}({\prec_1} \cup {\prec_2})
\end{align}
if $\prec$ is a strict partial order; otherwise $P(O_1,O_2)$ is undefined. For two suites $\suite_1$ and $\suite_2$ define
\[
P(\suite_1,\suite_2) = \{P(O_1,O_2) \mid O_i \in \suite_i, P(O_1,O_2) \textrm{ exists}\}.
\]
Furthermore, for subsets $X,Y \subset \BAS$ define an attack $S_{X,Y}(O_1,O_2) = 
\poset$ by
\begin{align*}
A &= A_1 \cup A_2\\
{\prec} &= \operatorname{tr}({\prec_1} \cup {\prec_2} \cup ((X \cap A) \times (Y \cap A))
\end{align*}
if $\prec$ is a strict partial order. For two suites we likewise define
\[
S_{X,Y}(\suite_1,\suite_2) = \{S_{X,Y}(O_1,O_2) \mid O_i \in \suite_i, S_{X,Y}(O_1,O_2) \textrm{ exists}\}.
\]

Recall that $\BAS_{v}$ is the set of BAS descendants of $v$.

\begin{lemma} \label{lemma:dsem:app}
Consider a DAT with nodes $a \in \BAS$ and $v_1,v_2 \in N$. Then:
\begin{enumerate}
    \item $\dsemin{a} = \{ \poset[\{a\}][\emptyset] \}$; \label{lemma:dsem:app:BAS}
    \item $\dsemin{\OR(v_1,v_2)} \subseteq \dsemin{v_1} \cup \dsemin{v_2}$;%
			\label{lemma:dsem:app:OR}
	\item $\dsemin{\AND(v_1,v_2)} \subseteq P(\dsemin{v_1},\dsemin{v_2})$; \label{lemma:dsem:app:AND}
	\item $\dsemin{\SAND(v_1,v_2)} \subseteq S_{\BAS_{v_1},\BAS_{v_2}}(\dsemin{v_1},\dsemin{v_2})$; \label{lemma:dsem:app:SAND}
	\item all attacks on the right-hand side of 1--4 are succesful.
\end{enumerate}
If $T$ is of tree type then furthermore
\begin{enumerate}
\setcounter{enumi}{5}
    \item $A_1$ and $A_2$ are disjoint for each $\poset[A_1][\prec_1] \in \dsemin{v_1}$ and $\poset[A_2][\prec_2] \in \dsemin{v_2}$. In particular $\dsemin{v_1}$ and $\dsemin{v_2}$ are disjoint.
    \item The inclusions in 2--4 are equalities.
\end{enumerate}
\end{lemma}

\begin{proof}
\begin{enumerate}
    \item By definition an attack $\poset$ reaches $a$ iff $a \in A$; hence the unique minimal attack on $a$ is $\poset[\{a\}][\emptyset]$.
    \item Suppose $O \in \dsemin{v}$ where $v = \OR(v_1,v_2)$. Since $O$ reaches $v$, it must reach either $v_1$ or $v_2$. Assume the former WLOG, and suppose that $O \notin \dsemin{v_1}$. Then there exists an $O' < O$ that reaches $\dsemin{v_1}$. But then $O'$ also reaches $v$, which contradicts the assumption $O \in \dsemin{v}$. We conclude that $O \in \dsemin{v_1}$.
    \item Suppose $O \in \dsemin{v}$ where $v = \AND(v_1,v_2)$. Since $O$ reaches $v$ it must reach both $v_1$ and $v_2$. Let $O_i \in \dsemin{v_i}$ be such that $O_i \leq O$ for $i = 1,2$. Then $O':= P(O_1,O_2)$ is the minimal attack satisfying $O_1,O_2 \leq O'$ if it exists; if not, then $O_1,O_2$ have no shared upper bound. It follows that $O'$ exists and $O' \leq O$. Since $O$ is minimal this implies $O = O'$.
    \item Suppose $O = \poset \in \dsemin{v}$ where $v = \SAND(v_1,v_2)$. Since $O$ activates $v$ there must exist $O_1 \in \dsemin{v_1}$ and $O_2 \in \dsemin{v_2}$ such that $O_1,O_2 \leq O$ and $a_1 \prec a_2$ for each $a_1 \in \BAS_{v_1} \cap A$ and $a_2 \in \BAS_{v_2} \cap A$. The minimal succesful attack with these properties is $O' = S_{\BAS_{v_1},\BAS_{v_2}}(O_1,O_2)$; hence $O' \leq O$. Since $O$ is minimal one has $O' = O$.
    \item This follows directly from the definition of the structure function $f_T$.
    \item If $T$ is of tree type then $\BAS_{v_1} \cap \BAS_{v_2} = \varnothing$, and so $A_1 \cap A_2 = \varnothing$ for $A_i \in \dsemin{v_i}$. In particular $A_1 \neq A_2$, and hence also $\dsemin{v_1} \cap \dsemin{v_2} = \varnothing$.
    \item Consider the case 2 first, and let $O = \poset \in \dsemin{v_1}$. Then $O$ reaches $v$, and we need to prove that it is minimal on $v$. Suppose $O' < O$ reaches $v$. Since $A \subset \BAS_{v_1}$ and $\BAS_{v_2} \cap \BAS_{v_1} = \varnothing$, we see that $O'$ does not reach $v_2$. Hence $O'$ must reach $v_1$; but then $O' < O$ contradicts the fact that $O \in \dsemin{v_1}$. It follows that $O$ is minimal.
    
    Now consider case 4, let $O_i = \poset[A_i][\prec_i] \in \dsemin{v_i}$ and consider $O = \poset =  S_{\BAS_{v_1},\BAS_{v_2}}(O_1,O_2)$. Then $O$ reaches $v$, and we need to prove that $O$ is minimal on $v$. Since $\BAS_{v_1} \cap \BAS_{v_2} = \varnothing$ and $A_i \subset \BAS_{v_i}$, the relation ${\prec_1} \cup {\prec_2} \cup (A_1 \times A_2)$ is transitively closed, and so ${\prec} = {\prec_1} \cup {\prec_2} \cup (A_1 \times A_2)$. Suppose $O' = \poset[A'][\prec'] < O$ reaches $v$. For $i = 1,2$ define $O_i' = \poset[A_i'][\prec'_i]$ by
    \begin{align*}
    A_i' &= A_i \cap A'\\
    \prec'_i &= \prec'|_{A_i'}.
    \end{align*}
    Since $O'$ reaches $v$ we know that $O'_i$ reaches $v_i$. Since $O' < O$ there are three options:
    \begin{enumerate}
        \item $A' \subsetneq A$: then either $A_1' \subsetneq A_1$ or $A_2' \subsetneq A_2$. Assume the former WLOG; then $O'_1 < O_1$. Since $O_1 \in \dsemin{v_1}$ this yields a contradiction. \label{pf:dsem:app:sand1}
        \item $A' = A$ and $(A_1 \times A_2) \not \subset {\prec'}$: then $O'$ is not succesful, which is a contradiction. \label{pf:dsem:app:sand2}
        \item $A' = A$ and $(A_1 \times A_2) \subset {\prec'}$ and ${\prec'} \subsetneq {\prec}$: then either ${\prec_1'}  \subsetneq {\prec_1}$ or ${\prec_2'}  \subsetneq {\prec_2}$. This is analogous to case \ref{pf:dsem:app:sand1}.
    \end{enumerate}
    In each case we find a contradiction, so we conclude that such $O'$ does not exist, and $O \in \dsemin{v}$.
    
    Equality in case 3 is proven analogously to case \ref{lemma:dsem:app:SAND}, except that case b) does not occur. \qedhere
\end{enumerate}
\end{proof}

Lemmas \ref{lemma:ssem} and \ref{lemma:dsem} are direct consequences of \Cref{lemma:dsem:app}:

\begin{replemma}{lemma:ssem}
	Consider a \SAT $T$ with nodes $a\in\BAS, v_1, v_2\in\ATnodes$. Then:
	\begin{enumerate}
	\def\REF#1{\textit{\ref{#1})}}
	\item	$\ssem{a} = \{ \{a\}\}$;
	\item	$\ssem{\OR(v_1,v_2)} \subseteq \ssem{v_1} \cup \ssem{v_2}$;
	\item	$\ssem{\AND(v_1,v_2)} \subseteq \{ \attack_1\cup \attack_2 \mid
				\attack_1\in\ssem{v_1} \land \attack_2\in\ssem{v_2} \}$;
	\item All RHS of cases \ref{lemma:ssem:BAS}--\ref{lemma:ssem:AND} consist of succesful attacks.
	\end{enumerate}
	If $T$ is of tree type then furthermore:
	\begin{enumerate}
	 \setcounter{enumi}{4}
	\item In cases \ref{lemma:dsem:OR} and \ref{lemma:dsem:AND}
			the \ssem{v_i} are disjoint, and in case \ref{lemma:dsem:AND}
			moreover the $A_i$ are pairwise disjoint;
	\item Equality holds in cases \ref{lemma:ssem:OR} and \ref{lemma:ssem:AND}.
	\end{enumerate}
\end{replemma}

\begin{proof} \label{lemma:ssem:proof}
Let $T_d$ be $T$ interpreted as a DAT. By \Cref{lemma:satisdat} we know that $\dsem{T} = \{A \mid \poset[A][\varnothing] \in \dsemin{T_d}\}$. Since
\[
P(\poset[A_1][\varnothing],\poset[A_2][\varnothing]) = \poset[A_1 \cup A_2][\varnothing]
\]
case \ref{lemma:ssem:AND} follows from case \ref{lemma:dsem:app:AND} of \Cref{lemma:dsem:app}. The other cases follow directly from \Cref{lemma:dsem:app}.
\end{proof}

\begin{replemma}{lemma:dsem}
	Consider a tree-structured \DAT with nodes ${a\in\BAS}$, ${v_1,v_2\in\ATnodes}$. Then:
	\begin{enumerate}
	\item	$\dsemin{a} = \{ \poset[\{a\}][\emptyset] \}$;%
			\label{lemma:dsem:BAS}
	\item	$\dsemin{\OR(v_1,v_2)} = \dsemin{v_1} \cup \dsemin{v_2}$;%
			\label{lemma:dsem:OR}
	\item	$\dsemin{\AND(v_1,v_2)} = \mbox{$\left\{
				\poset[\attack_1{\cup}\attack_2][{\prec_1}{\cup}{\prec_2}]
				\,|\, \poset[\attack_i][\prec_i]\in\dsemin{v_i}
			\right\}$}$;%
			\label{lemma:dsem:AND}
	\item	$\dsemin{\SAND(v_1,v_2)} = \{
				\poset[\attack_1\cup\attack_2 \,]%
				      [~{\prec_1}\cup{\prec_2}\cup{\attack_1\times\attack_2}]
				\cdots$ \\
				\hspace*{\stretch{1}} $\cdots
					\mid \mbox{$\poset[\attack_i][\prec_i]\in\dsemin{v_i}$}
			\}$;%
			\label{lemma:dsem:SAND}
	\item	In cases \REFS{lemma:dsem:OR}{lemma:dsem:SAND} above the
			\dsemin{v_i} are disjoint, and in cases \REF{lemma:dsem:AND} and
			\REF{lemma:dsem:SAND} moreover the $A_i$ are pairwise disjoint.
			\label{lemma:dsem:disjoint}
	\end{enumerate}
\end{replemma}
\begin{proof}	\label{lemma:dsem:proof}
We need to prove that if $T$ is of tree type, then $P(\poset[A_1][\prec_1],\poset[A_2][\prec_2])$ and $S_{\BAS_{v_1},\BAS_{v_2}}(\poset[A_1][\prec_1],\poset[A_2][\prec_2])$ exist and are defined by
\begin{align*}
&P(\poset[A_1][\prec_1],\poset[A_2][\prec_2]) \\
&= \poset[A_1 \cup A_2][{\prec_1}\cup{\prec_2}]\\
&S_{\BAS_{v_1},\BAS_{v_2}}(\poset[A_1][\prec_1],\poset[A_2][\prec_2]) \\
&= \poset[A_1 \cup A_2][{\prec_1}\cup{\prec_2}\cup(A_1 \times A_2)].
\end{align*}
We prove this for $S_{\BAS_{v_1},\BAS_{v_2}}(\poset[A_1][\prec_1],\poset[A_2][\prec_2]) =: \poset$, as the proof for $P(\poset[A_1][\prec_1],\poset[A_2][\prec_2])$ is analogous. Since $\BAS_{v_1} \cap \BAS_{v_2} = \varnothing$ we know that $A_i = A \cap \BAS_{v_i}$ for $i = 1,2$. Hence
\[
A_1 \times A_2 = (A \cap \BAS_{v_1}) \times (A \cap \BAS_{v_2}),
\]
and from the definition of $S_{X,Y}$ it suffices to show that  the relation ${\prec_1} \cup {\prec_2} \cup {(A_1 \times A_2)}$ is a strict partial order. But as $A_1$ and $A_2$ are disjoint this is a standard result.
\end{proof}




\section{Proofs for modular analysis} \label[appendix]{app:modular}

Because of \Cref{lemma:satisdat} it suffices to prove the theorem for DATs. The proof requires a bit of preparation in setting up the notation, but the overall structure is as follows:
\begin{enumerate}
    \item First, we show how the semantics of $T$ relate to that of $T^v$ and $T_v$ (\Cref{cor:dat}).
    \item Next, we discuss how this relates the maximal chains in attacks in $\dsemin{T}$ to those of attacks in $\dsemin{T^v}$ and $\dsemin{T_v}$ (\Cref{lemma:tchains}).
    \item Finally, we plug this into \eqref{eq:metrdat} to prove the theorem.
\end{enumerate}

We start off with a definition that expresses a useful way of combining two posets into one.

\begin{definition}
Let $O_1 = \poset[X_1][\prec_1]$ and $O_2 = \poset[X_2][\prec_2]$ be two nonempty posets, and let $x_1 \in X_1$. Define the \emph{insertion} $O_1[x_1/O_2]$ to be the poset $\poset[X][\prec]$ given by
\begin{align*}
X &= X_1 \setminus \{x_1\} \cup X_2\\
x \prec x' &\Leftrightarrow \begin{cases}
x \prec_1 x' & \textrm{if $x,x' \in X_1$}\\
x \prec_2 x' & \textrm{if $x,x' \in X_2$}\\
x \prec_1 x_1 & \textrm{if $x \in X_1$ and $x' \in X_2$}\\
x_1 \prec_1 x' & \textrm{ if $x \in X_2$ and $x' \in X_1$}.
\end{cases}
\end{align*}
\end{definition}

As the notation suggests, $O_1[x_1/O_2]$ is the poset obtained by replacing $x_1$ in $O_1$ by the entire poset $O_2$.

For a poset $O$, let $\operatorname{MC}_{O}$ be its set of maximal chains (which are posets themselves). We can find the maximal chains of an insertion as follows:

\begin{lemma} \label{lemma:inschains}
One has
\begin{align*}
&\operatorname{MC}_{O_1[x_1/O_2]}\\ 
& = \{C_1 \in \operatorname{MC}_{O_1} \mid x_1 \notin C_1\} \\
& \ \ \cup \{C_1[x_1/C_2] \mid x_1 \in C_1 \in \operatorname{MC}_{O_1}, C_2 \in \operatorname{MC}_{O_2}\}.
\end{align*}
\end{lemma}

\begin{proof}
Let $C = (c_1,\ldots,c_n) \in \operatorname{MC}_{O_1[x_1/O_2]}$. If $X_2 \cap C = \varnothing$, then $C$ is also a maximal chain of $O_1$ not containing $x_1$. Now suppose $C 
\cap X_2 \neq \varnothing$. By the nature of the poset $O_1[x_1/O_2]$, if $c_i \in X_1$, and $j < i < k$, then it is not possible that both $c_j,c_k \in X_2$. It follows that all the elements of $X_2 \cap C$ form a single block $C_2 = (c_i,\ldots,c_j)$; then $C_2$ is a chain in $X_2$. Furthermore, $C_1 = (c_1,\ldots,c_{i-1},x_1,c_{j+1},\ldots,c_n)$ is a chain in $X_1$, and $C = C_1[x_1/C_2]$. Furthermore, if $C_1$ and $C_2$ were not maximal, then at least one of them could be extended. This would also extend $C_1[x_1/C_2]$, so $C$ would not be maximal. This shows ``$\subset$''.

Conversely, let $C_1$ be a maximal chain of $O_1$ not containing $x_1$; this is also a chain of $O_1$. Suppose $C_1$ can be extended by an $x \in O_1[x_1/O_2]$. If $x \in X_1$, then $C_1$ could be extended by $x$; if not, then $C_1$ could be extended by $x_1$. Either way, $C_1$ is not maximal, which is a contradiction. Similarly, if $C_1$ is a maximal chain of $O_1$ containing $x_1$, and $C_2$ is a maximal chain of $O_2$, then one can show that $C_1[x_1/C_2]$ is maximal. This shows ``$\supset$''.
\end{proof}

Next, we define a map between attacks on $T_v$ and $T^v$ on one side, and attacks on $T$ on the other side. To describe the image of this map, we first need some more notation. Define $B := \BAS_{v}$, and define a map $r\colon \allAttacks_T \rightarrow \allAttacks_T$, where $\poset[A'][\prec'] = r(\poset[A][\prec])$ is given by $A' = A$ and
\begin{align*}
a_1 \prec' a_2 &\Leftrightarrow \begin{cases}
a_1 \prec a_2, & \textrm{if $a_1,a_2 \in B$ or $a_1,a_2 \notin B$}\\
\forall w \in B\colon w \prec a_2 &\textrm{if $a_1 \in B$ and $a_2 \notin B$}\\
\forall w \in B\colon a_1 \prec w &\textrm{if $a_1 \notin B$ and $a_2 \in B$}.
\end{cases}
\end{align*}

The intuition is that we remove all relations between elements of $B$ and $A\setminus B$ that are not shared with all elements of $B$.

\begin{lemma}
Let $O \in \allAttacks_T$.
\begin{enumerate}
    \item $r(O) \leq O$;
    \item $r(r(O)) = r(O)$;
    \item If $O$ is succesful then so is $r(O)$.
\end{enumerate}
\end{lemma}

\begin{proof}
Points 1 and 2 are immediate. For point 3 we note that since $v$ is a module, the constraints that \Cref{def:success} puts on the relation $\prec$ between elements of $B$ and $A \setminus B$, are shared between all elements of $A \setminus B$; hence $O$ satisfies such constraints if and only if $r(O)$ does.
\end{proof}

Define the sets
\begin{align*}
\mathscr{R}_T &= r(\allAttacks_T)\\
\allAttacks_{T^v}^+ &= \{\poset[A^v][\prec^v] \in \allAttacks_{T^v}: \tilde{v} \in A^v\}\\
\allAttacks_{T^v}^- &= \allAttacks_{T^v} \setminus \allAttacks_{T^v}^+\\
\allAttacks_{T_v}^{>0} &= \allAttacks_{T_v} \setminus \{\poset[\varnothing][\varnothing]\}.
\end{align*}
Then we define a map
\begin{align}
h\colon \allAttacks_{T^v}^- \cup (\allAttacks_{T^v}^+\times \allAttacks_{T_v}^{>0}) &\rightarrow \mathscr{R}_T \nonumber\\
\allAttacks_{T^v}^- \ni O^v &\mapsto O^v \nonumber\\
\allAttacks_{T^v}^+\times \allAttacks_{T_v}^{>0} \ni (O^v,O_v) &\mapsto O^v[\tilde{v}/O_v]. \nonumber
\end{align}
Furthermore, define a partial order $\leq$ on $\allAttacks_{T^v}^- \cup (\allAttacks_{T^v}^+ \times \allAttacks_{T_v}^{>0})$ by embedding it into $\allAttacks_{T^v} \times \allAttacks_{T_v}$, identifying $\allAttacks_{T^v}^-$ with $\allAttacks_{T^v}^- \times \{\poset[\varnothing][\varnothing]\}$. The following two lemmas show that $h$ is an order-preserving, success-preserving bijection:

\begin{lemma} \label{lem:dat1}
The map $h$ is an isomorphism of posets.
\end{lemma}

\begin{proof}
First, we note that for $(O^v,O_v) \in \allAttacks_{T^v}^+ \times \allAttacks_{T_v}^{>0}$ one has $r(h(O^v,O_v)) = h(O^v,O_v)$, and for $O^v \in \allAttacks_{T^v}^-$ one has $r(h(O^v)) = h(O^v)$. Therefore the image of $h$ lies in $\mathscr{R}_T$, and $h$ is well-defined.

Next we prove injectivity. If $O = \poset[A][\prec]$ is in the image of $h$, then $O \in h(\allAttacks_{T^v}^-)$ if and only if $A \cap B = \varnothing$. It is clear that $h$ is injective on $\allAttacks_{T^v}^-$, therefore it suffices to show that $h$ is injective on $\allAttacks_{T^v}^+ \times \allAttacks_{T_v}^{>0}$. If $O = \poset[A][\prec] = h(O^v,O_v)$, then one can recover $O^v = \poset[A^v][\prec^v]$ and $O_v = \poset[A_v][\prec_v]$ from $O$ by
\begin{align}
A^v &= A \setminus B \cup \{\tilde{v}\}, \label{eq:datmodproof1}\\
a \prec^v a' &\Leftrightarrow \begin{cases}
a \prec a',& \textrm{ if $a,a' \neq \tilde{v}$}\\
\forall w \in A \cap B\colon a \prec w,& \textrm{ if $a \neq \tilde{v} = a'$}\\
\forall w \in A \cap B\colon w \prec a',& \textrm{ if $a \neq \tilde{v} = a'$}
\end{cases} \label{eq:datmodproof2}\\
A_v &= A \cap B \label{eq:datmodproof3}\\
\prec_v &= \prec|_{A_v}. \label{eq:datmodproof4}
\end{align}
This shows that $h$ is injective. For surjectivity, let $\poset[A][\prec] \in \mathscr{R}_T$ and define $O^v = \poset[A^v][\prec^v]$ and $O_v = \poset[A_v][\prec_v]$ as in \eqref{eq:datmodproof1}--\eqref{eq:datmodproof4}; then $h(O^v,O_v) = \poset[A][\prec]$.

It is clear that $h$ is order-preserving on both $\allAttacks_{T^v}^-$ and $\allAttacks_{T^v}^+\times \allAttacks_{T_v}^{>0}$ individually. Now let $O^v_1 \in \allAttacks_{T^v}^-$, $O^v_2 \in \allAttacks_{T^v}^+$ and $O_v \in \allAttacks_{T_v}^{>0}$ be such that $O^v_1 \leq (O^v_2,O_v)$. Then $O^v_1 \leq O^v_2$ in the poset $\allAttacks_{T^v}$. Since $h(O^v_1) = O^v_1$ and $h(O^v_2,O_v)$ contains $O^v_2$ as a suborder, one has $h(O^v_1) \leq h(O^v_2,O_v)$. This proves that $h$ preserves the partial order.
\end{proof}

The following lemma follows directly from the definition of the structure function (see \Cref{def:success}).

\begin{lemma} \label{lem:dat2}
\begin{enumerate}   
    \item Let $O^v \in \allAttacks_{T^v}^-$. Then $h(O^v)$ is succesful on $T$ if and only if $O^v$ is succesful on $T^v$.
    \item Let $O^v \in \allAttacks_{T^v}^+$ and $O_v \in \allAttacks_{T_v}$. Then $h(O^v,O_v)$ is succesful if and only if either $O^v \setminus \{\tilde{v}\}$ is succesful, or if both $O^v$ and $O_v$ are succesful.
\end{enumerate}
\end{lemma}

From these two lemmas we get the following result:

\begin{corollary} \label{cor:dat}
Define $\dsemin{T^v}^{+} = \dsemin{T^v} \cap \allAttacks_{T_v}^{+}$ and likewise $\dsemin{T^v}^-$. Then $h$ induces a bijection 
\begin{equation*}
h\colon \dsemin{T^v}^- \cup (\dsemin{T^v}^+ \times \dsemin{T_v}) \stackrel{\sim}{\longrightarrow} \dsemin{T}.
\end{equation*}
\end{corollary}

\begin{proof}
By Lemmas \ref{lem:dat1} and \ref{lem:dat2}, the preimage of $\dsemin{T}$ consists of the minimal elements $O^v_1 \in \allAttacks_{T^v}^-$ such that $O^v_1$ is succesful, and $(O^v_2,O_v) \in \allAttacks_{T^v}^+ \times \allAttacks_{T_v}^{>0}$ such that both $O^v_2$ and $O_v$ are succesful (if $O^v_2 \setminus \{\tilde{v}\}$ were succesful then $h(O^v_2,O_v)$ would not be minimal). This is equivalent to $O^v_1 \in \dsemin{T^v}^-$ and $(O^v_2,O_v) \in \dsemin{T^v}^+\times \dsemin{T_v}$. 
\end{proof}

We can now use this characterisation of $\dsemin{T}$, and the characterisation of maximal chains of \Cref{lemma:inschains}, to describe the maximal chains of the elements of $\dsemin{T}$. The proof is straightforward and therefore omitted.

\begin{lemma} \label{lemma:tchains}
Let $O \in \dsemin{T}$.
\begin{enumerate}
    \item If $O = h(O^v)$ for $O^v \in \dsemin{T^v}^-$, then $\operatorname{MC}_O = \operatorname{MC}_{O^v}$.
    \item If $O = h(O^v,O_v)$ for $(O^v,O_v) \in \dsemin{T^v}^+ \times \dsemin{T_v}$, then
    \begin{align*}
    \operatorname{MC}_O &= \{C^v \in \operatorname{MC}_{O^v} \mid \tilde{v} \notin C^v\} \\
    &\ \ \cup \{C^v[\tilde{v}/C_v] \mid \tilde{v} \in C^v \in \operatorname{MC}_{O^v}, C_v \in \operatorname{MC}_{O_v}\}.
    \end{align*}
\end{enumerate}
\end{lemma}

We are now in a position to prove the theorem. What follows is not deep mathematically, but it involves some heavy formula manipulation on $(V,\operOR,\operAND,\operSAND)$.

\begin{proof}
Since $h$ is an isomorphism of posets between $\dsemin{T}$ and $\dsemin{T^v}^- \cup (\dsemin{T^v}^+ \times \dsemin{T_v})$, it follows that

\begin{align}
\widecheck{\alpha}(T) &= \bigoperOR_{O \in \dsemin{T}} \bigoperAND_{C \in \operatorname{MC}_{O}} \bigoperSAND_{a \in C} \alpha(a) \nonumber \\
&= \bigoperOR_{O^v \in \dsemin{T^v}^-} \bigoperAND_{C \in \operatorname{MC}_{h(O^v)}} \bigoperSAND_{a \in C} \alpha(a) \label{eq:sat1}\\
& \ \ \operOR \bigoperOR_{\substack{O^v \in \dsemin{T^v}^+,\\ O_v \in \dsemin{T_v}}} \bigoperAND_{C \in \operatorname{MC}_{h(O^v,O_v)}} \bigoperSAND_{a \in C} \alpha(a). \nonumber
\end{align}
Focusing on the first part, from the definition of $h$ we find that
\begin{equation}
\bigoperOR_{O^v \in \dsemin{T^v}^-} \bigoperAND_{C \in \operatorname{MC}_{h(O^v)}} \bigoperSAND_{a \in C} \alpha(a) = \bigoperOR_{O^v \in \dsemin{T^v}^-} \bigoperAND_{C^v \in \operatorname{MC}_{O^v}} \bigoperSAND_{a \in C^v} \alpha^v(a). \label{eq:sat2}
\end{equation}
Furthermore, for $O^v \in \dsemin{T^v}^+$ and $O_v \in \dsemin{T_v}^-$, the maximal chains of $h(O^v,O_v) = O^v[\tilde{v}/O_v]$ are given by  \Cref{lemma:tchains}. It follows that the second part of \eqref{eq:sat1} can be written as
\begin{align}
&\bigoperOR_{\substack{O^v \in \dsemin{T^v}^+,\\ O_v \in \dsemin{T_v}}} \bigoperAND_{C \in \operatorname{MC}_{h(O^v,O_v)}} \bigoperSAND_{a \in C} \alpha(a)  \nonumber\\
&= \bigoperOR_{\substack{O^v \in \dsemin{T^v}^+,\\ O_v \in \dsemin{T_v}}} \left(\bigoperAND_{\substack{C^v \in \operatorname{MC}_{O^v}\colon \\ \tilde{v} \notin C^v}} \left(\bigoperSAND_{a \in C^v} \alpha^v(a)\right)\right. \label{eq:sat3}\\
&\ \ \ \ \ \left.\operAND  \bigoperAND_{\substack{C^v \in \operatorname{MC}_{O^v},\\ C_v \in \operatorname{MC}_{O_v}\colon \\ \tilde{v} \in C^v}}\left(\bigoperSAND_{a \in C^v \setminus \{\tilde{v}\}} \alpha^v(a) \operSAND \bigoperSAND_{a \in C_v} \alpha_v(a)\right)\right),\nonumber
\end{align}
where $\alpha_v$ is the restriction of $\alpha$ to $T_v$. For the second line we have
\begin{align}
& \bigoperAND_{\substack{C^v \in \operatorname{MC}_{O^v},\\ C_v \in \operatorname{MC}_{O_v}\colon \\ \tilde{v} \in C^v}}\left(\bigoperSAND_{a \in C^v \setminus \{\tilde{v}\}} \alpha^v(a) \operSAND \bigoperSAND_{a \in C_v} \alpha_v(a)\right) \nonumber \\
&= \bigoperAND_{\substack{C^v \in \operatorname{MC}_{O^v}\colon \\ \tilde{v} \in C^v}} \left(\bigoperSAND_{a \in C^v \setminus \{\tilde{v}\}} \alpha^v(a) \operSAND \bigoperAND_{C_v \in \operatorname{MC}_{O_v}} \bigoperSAND_{a \in C_v} \alpha_v(a)\right)\nonumber\\
&= \bigoperAND_{\substack{C^v \in \operatorname{MC}_{O^v}\colon \\ \tilde{v} \in C^v}} \left(\bigoperSAND_{a \in C^v \setminus \{\tilde{v}\}} \alpha(a) \operSAND \widehat{\alpha}_v(O_v)\right). \nonumber
\end{align}
Hence we can write
\begin{align}
&\bigoperOR_{\substack{O^v \in \dsemin{T^v}^+,\\ O_v \in \dsemin{T_v}}} \left(\bigoperAND_{\substack{C^v \in \operatorname{MC}_{O^v}\colon \\ \tilde{v} \notin C^v}} \left(\bigoperSAND_{a \in C^v} \alpha^v(a)\right)\right. \nonumber\\
& \ \ \ \left.\operAND  \bigoperAND_{\substack{C^v \in \operatorname{MC}_{O^v},\\ C_v \in \operatorname{MC}_{O_v}\colon \\ \tilde{v} \in C^v}}\left(\bigoperSAND_{a \in C^v \setminus \{\tilde{v}\}} \alpha^v(a) \operSAND \bigoperSAND_{a \in C_v} \alpha_v(a)\right)\right),\nonumber \\
&=\bigoperOR_{\substack{O^v \in \dsemin{T^v}^+,\\ O_v \in \dsemin{T_v}}} \left(\bigoperAND_{\substack{C^v \in \operatorname{MC}_{O^v}\colon \\ \tilde{v} \notin C^v}} \left(\bigoperSAND_{a \in C^v} \alpha^v(a)\right)\right. \nonumber\\
&\ \ \ \left.\operAND  \bigoperAND_{\substack{C^v \in \operatorname{MC}_{O^v}\colon \\ \tilde{v} \in C^v}} \left(\bigoperSAND_{a \in C^v \setminus \{\tilde{v}\}} \alpha^v(a) \operSAND \widehat{\alpha}_v(O_v)\right)\right)\nonumber\\
&=\bigoperOR_{O^v \in \dsemin{T^v}^+} \left(\bigoperAND_{\substack{C^v \in \operatorname{MC}_{O^v}\colon \\ \tilde{v} \notin C^v}} \left(\bigoperSAND_{a \in C^v} \alpha^v(a)\right)\right.\nonumber\\
&\ \ \ \left.\operAND  \bigoperAND_{\substack{C^v \in \operatorname{MC}_{O^v}\colon \\ \tilde{v} \in C^v}} \left(\bigoperSAND_{a \in C^v \setminus \{\tilde{v}\}} \alpha^v(a) \operSAND \bigoperOR_{O_v \in \dsemin{T_v}}  \widehat{\alpha}_v(O_v)\right)\right)\nonumber \\
&=\bigoperOR_{O^v \in \dsemin{T^v}^+} \left(\bigoperAND_{\substack{C^v \in \operatorname{MC}_{O^v}\colon \\ \tilde{v} \notin C^v}} \left(\bigoperSAND_{a \in C^v} \alpha^v(a)\right)\right.\nonumber\\
&\ \ \ \left.\operAND  \bigoperAND_{\substack{C^v \in \operatorname{MC}_{O^v}\colon \\ \tilde{v} \in C^v}} \left(\bigoperSAND_{a \in C^v \setminus \{\tilde{v}\}} \alpha^v(a) \operSAND \widecheck{\alpha}_v(T_v)\right)\right) \nonumber\\
&= \bigoperOR_{O^v \in \dsemin{T^v}^+} \bigoperAND_{C^v \in \operatorname{MC}_{O^v}}\bigoperSAND_{a \in C^v} \alpha^v(a). \nonumber
\end{align}

Combining this with \eqref{eq:sat1}, \eqref{eq:sat2} and \eqref{eq:sat3} we find
\begin{align*}
\widecheck{\alpha}(T) &= \bigoperOR_{O^v \in \dsemin{T^v}^-} \bigoperAND_{C^v \in \operatorname{MC}_{O^v}} \bigoperSAND_{a \in C^v} \alpha^v(a) \\
& \ \ \operOR \bigoperOR_{O^v \in \dsemin{T^v}^+} \bigoperAND_{C^v \in \operatorname{MC}_{O^v}}\bigoperSAND_{a \in C^v} \alpha^v(a)\\
&= \widecheck{\alpha}^v(T^v). \qedhere
\end{align*} 
\end{proof}

\section{Proofs of section \ref{sec:order}} \label[appendix]{app:order}

\begin{replemma}{lemma:LOSG}
\begin{enumerate}
    \item If $(X,\preceq,\operAND)$ is an LOSG, then $(X,\min,\operAND)$ is a semiring attribute domain.
    \item If $(X,\preceq,\operAND,\operSAND)$ is an DLOSG, then $(X,\min,\operAND,\operSAND)$ is a semiring dynamic attribute domain.
\end{enumerate}
\end{replemma}

\begin{proof}
We only prove the static case as the dynamic case is similar. It is clear that $\min$ is commutative and associative, so we have to prove distributivity. To do this, note that for all $x,y,z \in X$ one has
\begin{align*}
\min(x,y) = x &\Rightarrow x \preceq y \\
&\Rightarrow x \operAND z \preceq y \operAND z \\
&\Rightarrow \min(x \operAND z, y \operAND z)  = x \operAND z. 
\end{align*}
Analogously one shows that $\min(x \operAND z, y \operAND z) = y \operAND z$ whenever $\min(x,y) = y$, and so $\min(x,y) \operAND z = \min(x \operAND z, y \operAND z)$.
\end{proof}

\subsection{Uncertainty sets}

The statement here is a consequence of the fact that the maximal value of $\metr{T}$ is obtained when all $\alpha(a)$ are maximal. We can express this more formally as follows. Let $(X,\preceq,\operAND,\operSAND)$ be a DLOSG, and let $T$ be a DAT. Consider the map
\begin{align*}
W_T\colon X^{\BAS} &\rightarrow X \\
(\alpha(a))_{a \in \BAS} &\mapsto \metr{T}.
\end{align*}

\begin{lemma}
The map $W_T$ is nondecreasing in each $\alpha(a)$.
\end{lemma}

\begin{proof}
The operations $\operAND$ and $\operSAND$ are nondecreasing in both arguments by definition. Let $A \in \dsemin{T}$. Then $\metrA{A}$ is an expression in the $\alpha(a)$ and the operators $\operAND$ and $\operSAND$, and hence is nondecreasing in each $\alpha(a)$. It follows that $W_T(\alpha) = \min_{A \in \dsemin{T}}\metrA{A}$ is nondecreasing as well.
\end{proof}

\begin{lemma}
One has $\widecheck{\beta}(T) = (\widecheck{L}_T,\widecheck{U}_T)$.
\end{lemma}

\begin{proof}
By definition one has $\widecheck{\beta}(T) = (W_T((L_a)_{a \in \BAS}),W_T((U_a)_{a \in \BAS}))$. Since $W_T$ is nondecreasing in each argument one has
\begin{equation*}
W_T((U_a)_{a \in \BAS})) = \sup\{W_T(\alpha) \mid \forall a. \alpha(a) \leq U_a\} = \widecheck{U}_T.
\end{equation*}
The analogous statement holds for $\widecheck{L}_T$.
\end{proof}

\subsection{\texorpdfstring{$\boldsymbol{k}$}{k}-top metrics}

To prove \Cref{lemma:topk} we first need some auxiliary lemmas.

\begin{lemma} \label{lemma:topk:or}
Let $k,m \in \mathbb{N}$ with $k \leq m$. Then for all $M_1,M_2 \in \mathscr{M}(X)$ one has $\min^k(M_1 \uplus M_2) = \min^k(M_1 \uplus \min^m(M_2))$.
\end{lemma}

\begin{proof}
If $x \in \min^k(M_1 \uplus M_2) \setminus (M_1 \uplus \min^m(M_2))$, then $x \in M_2\setminus \min^m(M_2)$, so there are at least $m$ elements $x'$ of $M_2$ with $x' \prec x$. But these are elements of $M_1 \uplus M_2$ as well, so $x \notin \min^k(M_1 \uplus M_2)$. This is a contradiction, and so $\min^k(M_1 \uplus M_2) \subseteq \min^k(M_1 \uplus \min^m(M_2))$. Since these two multisets have the same cardinality it must be an equality.
\end{proof}

Define a binary operation $\star$ on $\mathscr{M}(X)$ by
\[
M_1 \star M_2 = \ldb x_1 \operAND x_2 \mid x_1 \in M_1, x_2 \in M_2 \rdb.
\]
Then $\star$ is commutative and associative, and $\star$ distributes over $\uplus$. Furthermore $M_1 \operAND^{k} M_2 = \min^k(M_1 \star M_2)$. Analogous to \Cref{lemma:topk:or} one can prove the following:

\begin{lemma} \label{lemma:topk:and}
Let $k,m \in \mathbb{N}$ with $k \leq m$. Then for all $M_1,M_2 \in \mathscr{M}(X)$ one has $\min^k(M_1 \star M_2) = \min^k(M_1 \star \min^m(M_2))$. \qedhere
\end{lemma}

This lemma has a straightforward analogon for $\operSAND$ rather than $\operAND$. We are now ready to prove this part's main result:

\begin{replemma}{lemma:topk}
The tuple $(\mathscr{M}^k(X),\operOR^k,\operAND^k,\operSAND^k)$ is a semiring dynamic attribute domain, and $\widecheck{\beta}(T) = \operatorname{Top}_k(T,\alpha)$.
\end{replemma}

\begin{proof}
It is clear that $\operOR^k$, $\operAND^k$ and $\operSAND^k$ are commutative. For the associativity of $\operOR^k$ we have, due to \Cref{lemma:topk:or},
\begin{align*}
M_1 \operOR^k (M_2 \operOR^k M_3) &= \operatorname{min}^k(M_1 \uplus \operatorname{min}^k(M_2 \uplus M_3)) \\
&= \operatorname{min}^k(M_1 \uplus (M_2 \uplus M_3)) \\
&= \operatorname{min}^k((M_1 \uplus M_2) \uplus M_3) \\
&= \operatorname{min}^k(\operatorname{min}^k(M_1 \uplus M_2) \uplus M_3) \\
&= (M_1 \operOR^k M_2) \operOR M_3.
\end{align*}
The associativity of the other operations are proven analogously. Similarly we only prove the distributivity of $\operAND^k$ over $\operSAND^k$. Here we have
\begin{align*}
M_1 \operAND^k (M_2 \operOR^k M_3) &= \operatorname{min}^k(M_1 \star \operatorname{min}^k(M_2 \uplus M_3)) \\
&= \operatorname{min}^k(M_1 \star (M_2 \uplus M_3)) \\
&= \operatorname{min}^k((M_1 \star M_2) \uplus (M_1 \star M_3)) \\
&= \operatorname{min}^k\big(\operatorname{min}^k(M_1 \star M_2) \uplus {} \\
	&\omit\hfill$\displaystyle\operatorname{min}^k(M_1 \star M_3)\big)$ \\
&= \operatorname{min}^k\big((M_1 \operAND^k M_2) \uplus (M_1 \operAND^k M_3)\big) \\
&= (M_1 \operAND^k M_2) \operOR^k (M_1 \operAND^k M_3).
\end{align*}
This shows that $(\mathscr{M}^k(X),\operOR^k,\operAND^k,\operSAND^k)$ is a semiring dynamic attribute domain. Furthermore it is easy to see that for $A \in \dsemin{T}$ one has $\widehat{\beta}(A) = \ldb \metrA{A} \rdb$, and so $\widecheck{\beta}(T) = \min^k\ldb \metrA{A} \mid A \in \dsemin{T} \rdb = \operatorname{Top}_k(T,\alpha)$.
\end{proof}

\subsection{The antichain semiring}

To prove \Cref{lemma:antichain} we first need a number of auxiliary results.

\begin{lemma}\label{lemma:paretocup}
For every $S_1,S_2 \in 2^X$ we have $m(S_1 \cup S_2) = m(S_1 \cup m(S_2))$.
\end{lemma}

\begin{proof}
Let $x \in m(S_1 \cup S_2)$. Then $x \in S_1 \cup S_2$; we show that $x \in S_1 \cup m(S_2)$. This is clear if $x \in S_1$, so assume $x \in S_2$. If $x \notin m(S_2)$, then there exists a $x' \in S_2$ with $x' \prec x$. But then $x' \in S_1 \cup S_2$, contradicting the fact that $x \in m(S_1 \cup S_2)$. So $x \in m(S_2) \subseteq S_1 \cup m(S_2)$. To prove that $x \in m(S_1 \cup m(S_2))$, suppose that this is not the case and that there exists a $x' \in S_1 \cup m(S_2))$ with $x' \prec x$. Then $x' \in S_1 \cup S_2$, contradicting the fact that $x \in m(S_1 \cup S_2)$. We conclude that $x \in m(S_1 \cup m(S_2))$ and $m(S_1 \cup S_2) \subseteq m(S_1 \cup m(S_2))$.

Now let $x \in m(S_1\cup m(S_2))$. Then $x \in S_1 \cup S_2$. Suppose $x \notin m(S_1 \cup S_2)$, and let $x' \in S_1 \cup S_2$ be such that $x' \prec x$. If $x' \in S_1$ then $x' \in S_1 \cup m(S_2)$, contradicting $x \in m(S_1 \cup m(S_2))$. If $x' \in S_2$, let $x'' \in m(S_2)$ be such that $x'' \preceq x'$. Then $x'' \in S_1 \cup m(S_2)$ and $x'' \prec x$, again contradicting $x \in m(S_1 \cup m(S_2))$. Either way, we find that such $x'$ cannot exist, and $x \in m(S_1 \cup S_2)$. Hence $m(S_1 \cup m(S_2)) \subseteq m(S_1 \cup S_2)$.
\end{proof}

Define a binary operation $\star$ on $2^X$ by
\[
S_1 \star S_2 = \{x_1 \operAND x_2 \mid x_1 \in S_1, x_2 \in S_2\}.
\]
Clearly $\star$ is commutative and associative, and $\star$ distributes over $\cup$. Furthermore $B \tilde{\operOR} B' = m(B \star B')$ for $B,B' \in \operatorname{AC}_X$. The operation $\star$ has the following property:

\begin{lemma} \label{lemma:paretostar}
For every $S_1,S_2 \in 2^X$ we have $m(S_1 \star S_2) = m(S_1 \star m(S_2))$.
\end{lemma}

\begin{proof}
Let $x \in m(S_1 \star S_2)$, and let $x_i \in S_i$ be such that $x = x_1 \operAND x_2$. Let $x_2' \in m(S_2)$ such that $x_2' \preceq x_2$. Since $X$ is a DPOSG one has $x_1 \operAND x_2' \preceq x_1 \operAND x_2 = x$. Since $x \in m(S_1 \star S_2)$ this means that $x = x_1 \operAND x_2'$, and so $x \in S_1 \star m(S_2)$. To prove that $x \in m(S_1 \star m(S_2))$, suppose that $x' \in S_1 \star m(S_2)$ satisfies $x' \preceq x$. Then $x' \in S_1 \star S_2$, hence $x' = x$ by the minimality of $x$ in $S_1 \star S_2$. We conclude that $x \in m(S_1 \star m(S_2))$, hence $m(S_1 \star S_2) \subseteq m(S_1 \star m(S_2))$.

Let $x \in m(S_1 \star m(S_2))$. Then $x \in S_1 \star S_2$, and we need to prove that $x \in m(S_1 \star S_2)$. Let $x' \in m(S_1 \star S_2)$ such that $x' \preceq x$. Write $x' = x_1' \operAND x_2'$ for $x_i' \in S_i$, and let $x_2'' \in m(S_2)$ be such that $x_2'' \preceq x_2'$. Then $x_1' \operAND x_2'' \preceq x' \preceq x$, and since $x$ is minimal in $S_1 \star m(S_2)$ we must have $x = x' = x_1 \operAND x_2''$. This shows that $x \in m(S_1 \star S_2)$, hence $m(S_1 \star m(S_2)) \subseteq m(S_1 \star S_2)$.
\end{proof}

Again the analogous statement for $\operSAND$ instead of $\operAND$ holds as well.

\begin{replemma}{lemma:antichain}
The tuple $(\operatorname{AC}_X,\operOR_{\AC},\operAND_{\AC})$ is a semiring attribute domain.
\end{replemma}

\begin{proof}
It is clear that the operators are commutative. By \Cref{lemma:paretostar} we have
\begin{align*}
B \tilde{\operAND} (B' \tilde{\operAND} B'') &= m(B \star m(B' \star B'')) \\
&= m(B \star (B' \star B'')) \\
&= m((B \star B')  \star B'') \\
&= m(m(B \star B') \star B'') \\
&= (B \tilde{\operAND} B')  \tilde{\operAND} B'',
\end{align*}
showing the associativity of $\tilde{\operAND}$. The associativity of $\tilde{\operOR}$ and $\tilde{\operSAND}$ is proven analogously, using \Cref{lemma:paretocup} for $\tilde{\operOR}$. We show distributivity only for $\tilde{\operAND}$ over $\tilde{\operOR}$, as the other cases are analogous:
\begin{align*}
B \tilde{\operAND} (B' \tilde{\operOR} B'') &= m(B \star m(B' \cup B'')) \\
&= m(B \star (B' \cup B'')) \\
&= m((B \star B') \cup (B \star B'')) \\
&= m(m(B \star B') \cup m(B \star B'')) \\
&= (B \tilde{\operOR} B') \tilde{\operAND} (B \tilde{\operOR} B''). \qedhere
\end{align*}
\end{proof}

\section{Proofs of Theorems \ref{theo:bottom_up_SAT}, \ref{theo:bottom_up_DAT}, \ref{theo:idempotent} and \ref{theo:bottom_up_BDD}} \label[appendix]{app:metrics}

In this appendix we prove the validity of the algorithms used to calculate metrics. For ATs of tree type this is essentially a direct consequence of \Cref{thm:mod}, as in a tree every node is a module. By \Cref{lemma:satisdat}, \Cref{theo:bottom_up_SAT} is a direct consequence of \Cref{theo:bottom_up_DAT}, so we only prove the latter. We first introduce one auxiliary lemma.

\begin{lemma} \label{lemma:indstep}
Let $T$ be a DAT where the children of $R_T$ are BASes.
\begin{enumerate}
    \item If $T = \OR(a,b)$, then $\metr{T} = \alpha(a) \operOR \alpha(b)$.
    \item If $T = \AND(a,b)$, then $\metr{T} = \alpha(a) \operAND \alpha(b)$.
    \item If $T = \SAND(a,b)$, then $\metr{T} = \alpha(a) \operSAND \alpha(b)$.
\end{enumerate}
\end{lemma}

\begin{proof}
Most of this follows directly from the definition of $\widecheck{\alpha}$ and \Cref{lemma:dsem}.
\begin{enumerate}
    \item In this case $\dsemin{T} = \{\poset[\{a\}][\varnothing],\poset[\{b\}][\varnothing]\}$, and from \eqref{eq:metrdat} we get $\metr{T} = \metrA{\poset[\{a\}][\varnothing]} \operOR \metrA{\poset[\{b\}][\varnothing]} = \alpha(a) \operOR \alpha(b)$.
    \item In this case $\dsemin{T} = \{\poset[\{a,b\}][\varnothing]\}$, and from \eqref{eq:metrdat} we get $\metr{T} = \metrA{\poset[\{a,b\}][\varnothing]} = \alpha(a) \operAND \alpha(b)$.
    \item In this case $\dsemin{T} = \{\poset[\{a,b\}][\{(a,b)\}]\}$, and from \eqref{eq:metrdat} we get $\metr{T} = \metrA{\poset[\{a,b\}][\{(a,b)\}]} = \alpha(a) \operSAND \alpha(b)$. \qedhere
\end{enumerate}
\end{proof}

\begin{reptheorem}{theo:bottom_up_SAT}
	Let \T be a static \AT with tree structure,
	$\attrOp$ an attribution on \Vdom,
	and $\domain=(\Vdom,\operOR,\operAND)$ a semiring attribute domain.
	Then $\metr{\T} = \BUSAT(\T,\ATroot,\attrOp,\domain)$.
\end{reptheorem}

\begin{reptheorem}{theo:bottom_up_DAT}
	Let \T be a \DAT with tree structure,
	$\attrOp$ an attribution on \Vdom,
	and $\domain=(\Vdom,\operOR,\operAND,\operSAND)$ a
	semiring dynamic attribute domain.
	Then $\metr{\T} = \BUDAT(\T,\ATroot,\attrOp,\domain)$.
\end{reptheorem}

\begin{proof}
We prove by induction on $v$ that $\metr{T_v} = \BUDAT(T,v,\attrOp,\domain)$. for convenience we assume that $T$ is binary. If $v$ is a BAS, then $\metr{T_v} = \metrA{\poset[\{a\}][\varnothing]} = \alpha(a)$. Now suppose $v = \OR(v_1,v_2)$. Since $T$ is a tree, every inner node is a module, in particular $v$. Let $\tilde{T}$ be the tree obtained by replacing $v_1$ and $v_2$ by BASes $\tilde{v}_i$, and let $\beta$ be an attribution on $\tilde{T}$ given by $\beta(\tilde{v}_i) = \metr{T_{v_i}}$; then by \Cref{thm:mod} one has $\widecheck{\beta}(\tilde{T}) = \metr{T_v}$. By the induction hypothesis and \Cref{lemma:indstep} we now have
\begin{align*}
\metr{T} &= \widecheck{\beta}(\tilde{T}) \\
&= \beta(\tilde{v}_1) \operOR \beta(\tilde{v}_2)\\
&= \metr{T_{v_1}} \operOR \metr{T_{v_2}} \\
&= \BUDAT(T,v_1,\attrOp,\domain) \operOR \BUDAT(T,v_2,\attrOp,\domain) \\
&= \BUDAT(T,v\attrOp,\domain).
\end{align*}
The proof for AND- and SAND-gates is analogous.
\end{proof}

\begin{reptheorem}{theo:idempotent}
Let $T = (N,E)$ be a SAT, let $D = (V,\operOR,\operAND)$ be an attribute domain, and let $\alpha\colon N \rightarrow V$ be an attribution. Suppose that $D$ is a both idempotent and absorbing. Then $\metr{T} = \mathtt{BU_{SAT}}(T,R_T,\alpha,D)$.
\end{reptheorem}

\begin{proof}
For a node $v$, let $\operatorname{Suc}_v \in \allSuites$ be the suite of attacks that activate $v$; then $\dsem{v}$ is the set of minimal elements in $\operatorname{Suc}_v$ under the partial order $\subseteq$. It follows from the absorption axiom that
\begin{equation*}
\bigoperOR_{A \in \dsem{v}} \bigoperAND_{a \in A} \alpha(a) = \bigoperOR_{A \in \operatorname{Suc}_v} \bigoperAND_{a \in A} \alpha(a).
\end{equation*}
We now prove the theorem by showing that $\widecheck{\alpha}(v) = \BUSAT(T,v,\alpha,D)$ by induction over $T$. It is clear for BASes. Suppose $v = \OR(v_1,\ldots,v_n)$; then $\operatorname{Suc}_{v_i} = \bigcup_{i} \operatorname{Suc}_{v_i}$. It follows that
\begin{align}
\widecheck{\alpha}(v) &= \bigoperOR_{A \in \bigcup_i \operatorname{Suc}_{v_i}} \widehat{\alpha}(A) \nonumber\\
&= \bigoperOR_{A \in \bigsqcup_i \operatorname{Suc}_{v_i}} \widehat{\alpha}(A) \label{eq:idempotent1}\\
&= \bigoperOR_{i \leq n} \bigoperOR_{A \in \operatorname{Suc}_{v_i}} \widehat{\alpha}(A) \nonumber\\
&= \bigoperOR_{i \leq n} \BUSAT(T,v_i,\alpha,D)
\label{eq:idempotent2} \\
&= \BUSAT(T,v,\alpha,D).\nonumber
\end{align}
Here \eqref{eq:idempotent1} follows from the idempotence of $\operOR$, and \eqref{eq:idempotent2} follows from the induction hypothesis.

Now suppose that $v = \AND(v_1,\ldots,v_n)$; then $\operatorname{Suc}_v = \bigcap_i \operatorname{Suc}_{v_i}$. We claim that
\begin{equation*}
\bigcap_i \operatorname{Suc}_{v_i} = \left\{\bigcup_i A_i \middle| \forall i\colon A_i \in \operatorname{Suc}_{v_i}\right\}.
\end{equation*}
The claim is proven as follows. On one hand, if $A \in \bigcap_i \operatorname{Suc}_{v_i}$, then we can write $A = \bigcap_i A$, which shows that it is part of the right hand side. On the other hand, each $\operatorname{Suc}_{v_i}$ is upward closed: if $A \in \operatorname{Suc}_{v_i}$ and $A \subseteq A'$, then $A' \in \operatorname{Suc}_{v_i}$. It follows that if $A_i \in \operatorname{Suc}_{v_i}$ for all $i$, then $\bigcup_i A_i \in \bigcap_i \operatorname{Suc}_{v_i}$. This proves the claim. From the claim we find
\begin{align}
\widecheck{\alpha}(v) &= \bigoperOR_{A \in \cap_i \operatorname{Suc}_{v_i}} \bigoperAND_{a \in A} \alpha(a) \nonumber\\
&= \bigoperOR_{\forall i\colon A_i \in \operatorname{Suc}_{v_i}} \bigoperAND_{a \in \bigcup_i A_i} \alpha(a) \nonumber\\
&= \bigoperOR_{\forall i\colon A_i \in \operatorname{Suc}_{v_i}} \bigoperAND_{a \in \bigsqcup_i A_i} \alpha(a) \label{eq:idempotent3}\\
&= \bigoperOR_{\forall i\colon A_i \in \operatorname{Suc}_{v_i}} \bigoperAND_{i \leq n} \bigoperAND_{a \in A_i} \alpha(a) \nonumber\\
&= \bigoperAND_{i \leq n} \bigoperOR_{A_i \in \operatorname{Suc}_{v_i}}  \bigoperAND_{a \in A_i} \alpha(a) \nonumber\\
&= \bigoperAND_{i \leq n} \BUSAT(T,v_i,\alpha,D) \label{eq:idempotent4}\\
&= \BUSAT(T,v,\alpha,D).\nonumber
\end{align}
Here we use the idempotence of $\operAND$ in \eqref{eq:idempotent3} and the induction hypothesis in \eqref{eq:idempotent4}.
\end{proof}

\begingroup

\def\root{\ensuremath{\BDDroot[{\bddT}]}\xspace}
\def\IC#1{\ensuremath{\mathrm{IC}_{#1}}\xspace}

\begin{reptheorem}{theo:bottom_up_BDD}
	Let \T be a static \AT, \bddT its \BDD encoding over \poset[\BAS][{<}],
	$\attrOp$ an attribution on \Vdom,
	and $\ndomain=(\Vdom,\operOR,\operAND,\ntOR,\ntAND)$ an absorbing unital semiring attribute domain.
	Then $\metr{\T} = \BUBDD(\bddT,\root,\attrOp,\ndomain)$.
\end{reptheorem}
\begin{proof}	\label{theo:bottom_up_BDD:proof}

Let $\bddT = (W,\low,\high,\BDDlab)$ be the BDD-encoding of $T$. For $w \in W$ and $p \in P(w)$, define $\metrA{p}$ and $\metr{w}$ as in \ref{eq:bddalg1} and \ref{eq:bddalg2}; we claim that $\metr{R_{B_T}} = \metr{T}$. To see this, note that by \Cref{thm:BDDpaths} we have 
\begin{align*}
	\metr{R_{B_T}}
		~= \bigoperOR_{p \in P(R_{B_T})} \metrA{p} 
		~= \bigoperOR_{A \in \pi(P(R_{B_T}))} \metrA{A}.
\end{align*}
Each element of $\pi(P(R_{B_T}))$ is succesful, and hence contains a minimal attack. Furthermore, $\ssem{T} \subseteq \pi(P(R_{B_T}))$. Since $D$ is absorbing, this implies 
\begin{equation*}
	\bigoperOR_{A \in \pi(P(R_{B_T}))} \metrA{A}
		~= \bigoperOR_{A \in \ssem{T}} \metrA{A}
		~=~ \metr{T},
\end{equation*}
completing the proof of our claim.

We complete the proof of the theorem by showing by induction that $\BUBDD(\bddT,w,\attrOp,\ndomain) = \metr{w}$ for all $w \in W$. It is certainly true for $w = \top$ for which $P(\top)$ has precisely one element, namely a path of length 0. For this path $p$ we find that $\metrA{p}$ is the empty $\operAND$-ation, so $\metrA{p} = \ntAND$ and $\metr{\top} = \ntAND$.  For the other base case we have $P(\bot) = \varnothing$, so $\metr{\bot} = \ntOR$. 

Now let us consider a non-terminal node $w$, and assume the induction hypothesis has been proven for its children $w_l = \low(w),w_h = \high(w)$. Then
\begin{align*}
P(w)
	&= \{p \cup \{(w,w_l)\} \mid p \in P(w_l)\}\\
	& \quad \quad \cup \{p \cup \{(w,w_h)\} \mid p \in P(w_h)\}.
\end{align*}
It follows that
\begin{align*}
\metr{w}
	&= \bigoperOR_{p \in P(w)} \metrA{p} \\
	&= \left(
		\bigoperOR_{p \in P(w_l)}\!\!\!\! \metrA{p\cup\{(w,w_l)\}}
	  \right)\!\!{\operOR}\!\!\left(
		\bigoperOR_{p \in P(w_h)}\!\!\!\! \metrA{p\cup\{(w,w_h)\}}
	  \right)\\
	&= \bigoperOR_{p \in P(w_l)} \bigoperAND_{\substack{w' \in W_n\colon\\ (w',\high(w')) \in p}}\\
	&\quad \quad \operOR \bigoperOR_{p \in P(w_h)}\left(\attr{\BDDlab(w)} \operAND \bigoperAND_{\substack{w' \in W_n\colon\\ (w',\high(w')) \in p}}\right)\\
	&= \metr{w_l} \operOR (\attr{\BDDlab(w)} \operAND \metr{w_h}) \\
	&= \BUBDD(\bddT,w_l,\attrOp,\ndomain) \\
		&\qquad \operOR (\attr{\BDDlab(w)} \operAND \BUBDD(\bddT,w_h,\attrOp,\ndomain)) \\
	&= \BUBDD(\bddT,w,\attrOp,\ndomain).
\qedhere
\end{align*}

\end{proof}
\endgroup  

\section{Proof of Theorem \ref{theo:NP_hard}} \label[appendix]{app:hard}

\begin{reptheorem}{theo:NP_hard}
	Given a DAG-structured static AT, the problem of computing any successful attack of minimal size is NP-hard.
\end{reptheorem}
\begin{proof}	\label{theo:NP_hard:proof}
	\def\minSAT{\textit{minSAT}\xspace}
	\def\CNFSAT{\textit{CNFSAT}\xspace}
	\def\posvar#1{\ensuremath{\tilde{#1}}\xspace}
	\def\posphi{\ensuremath{\widehat{\varphi}}\xspace}
	\def\minmap{\operatorname{\mathit{g}}}
	A static attack tree \T is equivalent to a logical formula whose atoms are the elements of \BAS: denote this formula \LT and note that none of its atoms appears negated.
	The problem of finding the smallest minimal attack in \T can thus be reformulated as finding the smallest $A\subseteq\BAS$ whose elements must evaluate to $\top$ to satisfy \LT.
	Denote this problem \minSAT.
	We now reduce \CNFSAT, the satisfiability problem for arbitrary logic formulae in conjunctive normal form, to solving \minSAT.
	Let $\varphi=\bigwedge_{i=1}^n c_i=\bigwedge_{i=1}^n\bigvee_{j=1}^{N_i}\ell_i^j$ be one such arbitrary formula with atoms $\mathrm{A}=\{a_k\}_{k=1}^m$.
	Define $\pos(\ell)\doteq\posvar{a}$ if the literal $\ell=\neg a$, and $\pos(\ell)\doteq a$ otherwise, where $\posvar{a}$ is a fresh non-negated (``positive'') atom.
	Now let $\hat{c}_i\doteq\bigvee_{j=1}^{N_i}\pos(\ell_i^j)$ and $\hat{a}\doteq(a\lor\posvar{a})$, and consider the formula $\posphi=\bigwedge_{i=1}^n\hat{c}_i\,\bigwedge_{k=1}^m\hat{a}_k$.
	Since no atom of \posphi is negated, by \minSAT we can find some $\minmap\from\mathrm{A}\cup\posvar{\mathrm{A}}\to\BB$ that satisfies \posphi, mapping to $\top$ a \emph{minimum amount of atoms} from $\mathrm{A}\cup\posvar{\mathrm{A}}$.
	Now consider the second part of the conjunction in \posphi: satisfying $\bigwedge_{k=1}^m\hat{a}_k$ requires, minimally, mapping $m$ atoms to $\top$, e.g.\ all the $\mathrm{A}$, or all the $\posvar{\mathrm{A}}$.
	But then:
	\begin{itemize}[topsep=0pt,parsep=.5ex,leftmargin=1em]
	\item	if $\minmap$ maps \emph{exactly $m$ atoms} to $\top$, then for every $\hat{a}_k$ it mapped either $a_k$ or $\posvar{a}_k$ to $\top$\ \;$\therefore$\;\ $\varphi$ is satisfiable;
	\item	else $\exists~\hat{c}_i,a_k$ s.t.\ $a_k\in\hat{c}_i$ and $\posvar{a}_k\in\hat{c}_i$ and $\minmap(a_k)=\minmap(\posvar{a}_k)=\top$\ \;$\therefore$\;\ $\varphi$ is unsatisfiable. \qedhere
	\end{itemize}

\end{proof}

\endgroup

\section{Proof of Lemma \ref{lem:metrcomp}} \label{app:metrcomp}

\begin{replemma}{lem:metrcomp}
Let $(\Vdom,\operOR,\operAND)$ be a semiring attribute domain, and $\attrOp\from \BAS_T \to \mathfrak{M}(\BAS_{\T})$ be given by $\attr{a} = \ldb a \rdb$. Then the multiset \metr{\T} is a set, and it is a succesful attack of minimal size.
\end{replemma}

\begin{proof}
It is clear that $\operOR$ and $\operAND$ are commutative and associative. Let $a,b,c \in V = \mathbb{Z}_{\geq 0}^n$; for distributivity we need to show that $\min(a+b,a+c) = a+\min(b,c)$. Assume $b \prec c$; then either $\sum_{i=1}^n b_i < \sum_{i=1}^n c_i$, or $\sum_{i=1}^n b_i = \sum_{i=1}^n c_i$ and $\mu(b) < \mu(c)$. In the former case one has
\begin{align*}
\sum_{i=1}^n (a+b)_i &= \sum_{i=1}^n a_i + \sum_{i=1}^n b_i \\
&< \sum_{i=1}^n a_i + \sum_{i=1}^n c_i \\
&= \sum_{i=1}^n (a+c)_i 
\end{align*}
and so $a+b \prec a+c$. In the latter case one has $\sum_{i=1}^n (a+b)_i = \sum_{i=1}^n (a+c)_i$. Furthermore, $\mu$ is multiplicative, and so 
\begin{align*}
\mu(a+b)
\,=\, \mu(a)\mu(b) 
\,<\, \mu(a)\mu(c)
\,=\, \mu(a+c).
\end{align*}
Again we conclude $a+b \prec a+c$; hence $\min(a+b,a+c) = a+b = a+\min(b,c)$ in both cases, and $(V,\operOR,\operAND)$ is a semiring domain. Furthermore, for $A \in \dsem{T}$ one has $\metrA{A} = A \in \mathscr{M}(\BAS_T)$, since $\operAND$ is just $\cup$ on disjoint sets. Hence $\bigoperOR_{A \in \dsem{T}} \metrA{A} = \min\{A \in \dsem{T}\}$. By definition of $\preceq$ this is an element of $\dsem{T}$ of minimal size.
\end{proof}

\end{document}